\def\submission{0}
\def\cameraready{0}

\ifnum\submission=0
\documentclass[11pt,letterpaper]{article}
\else
\documentclass[envcountsect]{llncs}
\makeatletter
\renewcommand{\@Opargbegintheorem}[4]{%
  #4\trivlist\item[\hskip\labelsep{#3#2\@thmcounterend}]}
\makeatother
\fi
\usepackage[utf8]{inputenc}
\usepackage{amsmath}
\usepackage{amsfonts}
\usepackage{xcolor}
\usepackage{breakcites}
\ifnum\submission=0
\usepackage{amsthm}
\fi
\usepackage{braket}
\usepackage{authblk}
\usepackage[colorlinks=true,allcolors=blue]{hyperref}
\ifnum\submission=0
\usepackage[margin=1in]{geometry}
\fi
\usepackage{xcolor}
\usepackage{algorithm}
\usepackage{mdframed} 
\mdfdefinestyle{figstyle}{ %
  linecolor=black!7, %
  backgroundcolor=black!7, %
  innertopmargin=10pt, %
  innerleftmargin=\textwidth, %
  innerrightmargin=\textwidth, %
  innerbottommargin=5pt %
}

\usepackage{comment}

\ifnum\submission=0
\newtheorem{theorem}{Theorem}[section]

\newtheorem{lemma}[theorem]{Lemma}

\newtheorem{definition}[theorem]{Definition}

\newtheorem{remark}{Remark}

\fi

\newtheorem{myclaim}[theorem]{Claim}

\def \sample { \overset{\hspace{0.1em}\mathsf{\scriptscriptstyle\$}}{\leftarrow} }

\newcommand{\ra}{\rightarrow}

\newcommand{\pro}{P}
\newcommand{\ver}{V}
\newcommand{\secpar}{\lambda}
\newcommand{\negl}{\mathsf{negl}}
\newcommand{\lang}{L}
\newcommand{\rel}{R}

\newcommand{\st}{\mathsf{st}}
\newcommand{\out}{\mathsf{out}}
\newcommand{\bit}{\{0,1\}}

\newcommand{\poly}{\mathsf{poly}}
\newcommand{\hil}{\mathcal{H}}

\newcommand{\defeq}{:=}

\newcommand{\Bad}{\mathsf{Bad}}
\newcommand{\bad}{\mathsf{bad}}

\newcommand{\redunderline}[1]{\textcolor{red}{\underline{\textcolor{black}{#1}}}}

\newcommand{\A}{\mathcal{A}}
\newcommand{\B}{\mathcal{B}}

\newcommand{\ext}{\mathsf{Ext}}
\newcommand{\Samp}{\mathsf{Samp}}
\newcommand{\SimSamp}{\mathsf{SimSamp}}

\newcommand{\calX}{\mathcal{X}}
\newcommand{\calY}{\mathcal{Y}}

\newcommand{\calD}{\mathcal{D}}

\newcommand{\ot}{\otimes}
\newcommand{\fail}{\mathsf{Fail}}

\newcommand{\NP}{\mathbf{NP}}
\newcommand{\QMA}{\mathbf{QMA}}
\newcommand{\BQP}{\mathbf{BQP}}

\newcommand{\rela}{\mathcal{R}}

\newcommand{\resp}{\mathtt{resp}}

\newcommand{\wipok}{\mathsf{WIPoK}}

\newcommand{\setup}{\mathsf{Setup}}

\newcommand*{\ipro}[2]{\langle #1|#2\rangle}
\newcommand{\TD}{\mathsf{TD}}

\newcommand{\reginp}{\mathsf{Inp}}

\newcommand*{\regR}{\mathbf{R}}
\newcommand*{\regX}{\mathbf{X}}
\newcommand*{\regY}{\mathbf{Y}}
\newcommand*{\regZ}{\mathbf{Z}}
\newcommand{\regW}{\mathbf{W}}

\newcommand*{\regM}{\mathbf{M}}
\newcommand*{\regout}{\mathbf{Out}}
\newcommand*{\reganc}{\mathbf{Anc}}
\newcommand*{\regaux}{\mathbf{Aux}}

\newcommand*{\regA}{\mathbf{A}}
\newcommand*{\regst}{\mathbf{ST}}
\newcommand*{\regB}{\mathbf{B}}
\newcommand*{\regD}{\mathbf{D}}
\newcommand{\regother}{\mathbf{Other}}
\newcommand{\regV}{\mathbf{V}}
\newcommand{\kext}{\mathcal{K}}
\newcommand{\test}{\mathsf{test}}
\newcommand{\siml}{\mathsf{Sim}}
\newcommand{\abort}{\mathsf{a}}
\newcommand{\nonabort}{\mathsf{na}}
\newcommand{\comb}{\mathsf{comb}}

\newcommand{\simlnaone}{\mathsf{Sim}_{\nonabort,1}(x,w,1^{\epsilon^{-1}},\ver^*_\secpar,\rho_\secpar)}
\newcommand{\simlnatwo}{\mathsf{Sim}_{\nonabort,2}(x,w,1^{\epsilon^{-1}},\ver^*_\secpar,\rho_\secpar)}
\newcommand{\simlnathree}{\mathsf{Sim}_{\nonabort,3}(x,w,1^{\epsilon^{-1}},\ver^*_\secpar,\rho_\secpar)}
\newcommand{\simlnafour}{\mathsf{Sim}_{\nonabort,4}(x,w,1^{\epsilon^{-1}},\ver^*_\secpar,\rho_\secpar)}
\newcommand{\simlnafive}{\mathsf{Sim}_{\nonabort,5}(x,w,1^{\epsilon^{-1}},\ver^*_\secpar,\rho_\secpar)}
\newcommand{\simlnai}{\mathsf{Sim}_{\nonabort,i}(x,w,1^{\epsilon^{-1}},\ver^*_\secpar,\rho_\secpar)}

\newcommand{\compind}{\stackrel{comp}{\approx}}
\newcommand{\statind}{\stackrel{stat}{\approx}}

\newcommand{\execution}[2]{\langle #1, #2 \rangle}

\newcommand{\psuc}{p_{\mathsf{real}}^{\mathsf{suc}}(x,w,1^{\epsilon^{-1}},\ver^*_\secpar,\rho_\secpar)}
\newcommand{\pasuc}{p_{\abort}^{\mathsf{suc}}(x,1^{\epsilon^{-1}},\ver^*_\secpar,\rho_\secpar)}
\newcommand{\pnasuc}{p_{\nonabort}^{\mathsf{suc}}(x,1^{\epsilon^{-1}},\ver^*_\secpar,\rho_\secpar)}
\newcommand{\pcombsuc}{p_{\comb}^{\mathsf{suc}}(x,1^{\epsilon^{-1}},\ver^*_\secpar,\rho_\secpar)}

\newcommand{\Exp}{\mathsf{Exp}}
\newcommand{\pp}{\mathsf{pp}}
\newcommand{\com}{\mathsf{com}}
\newcommand{\commit}{\mathtt{com}}
\newcommand{\open}{\mathtt{open}}
\newcommand{\Com}{\mathsf{Com}}
\newcommand{\CBCom}{\mathsf{CBCom}}

\newcommand{\Setup}{\mathsf{Setup}}
\newcommand{\Commit}{\mathsf{Commit}}

\newcommand{\CBSetup}{\CBCom.\mathsf{Setup}}

\newcommand{\CBCommit}{\CBCom.\mathsf{Commit}}
\newcommand{\SBCom}{\mathsf{SBCom}}
\newcommand{\SBSetup}{\SBCom.\mathsf{Setup}}
\newcommand{\SBCommit}{\SBCom.\mathsf{Commit}}

\newcommand{\COM}{\mathcal{C}\mathcal{O}\mathcal{M}}
\newcommand{\ppspace}{\mathcal{PP}}
\newcommand{\calM}{\mathcal{M}}
\newcommand{\calR}{\mathcal{R}}
\newcommand{\hiding}{\mathsf{hiding}}
\newcommand{\binding}{\mathsf{binding}}
\newcommand{\clbinding}{\mathsf{cl}\text{-}\binding}
\newcommand{\rand}{\mathsf{rand}}
\newcommand{\randspace}{\calR}
\newcommand{\OUT}{\mathsf{OUT}}
\newcommand{\sfQ}{\mathsf{Q}}
\newcommand{\sfR}{\mathsf{R}}
\newcommand{\real}{\mathsf{real}}
\newcommand{\final}{\mathsf{final}}
\newcommand{\hyb}{\mathsf{Hyb}}
\newcommand{\win}{\mathsf{win}}
\newcommand{\amp}{\mathsf{amp}}
\newcommand{\Amp}{\mathsf{Amp}}
\newcommand{\sequence}[1]{\{#1_\secpar\}_{\secpar\in\mathbb{N}}}

\newenvironment{boxfig}[2]{\begin{figure}[#1]\fbox{\begin{minipage}{0.97\linewidth}
                        \vspace{0.2em}
                        \makebox[0.025\linewidth]{}
                        \begin{minipage}{0.95\linewidth}
            {{
                        #2 }}
                        \end{minipage}
                        \vspace{0.2em}
                        \end{minipage}}}{\end{figure}}

\newcommand{\protocol}[4]{
\begin{boxfig}{h!}{\footnotesize 
\begin{center}
\centering{\textbf{#1}}
\end{center}
    #4
\vspace{0.2em} } \caption{\label{#3} #2}
\end{boxfig}
}

\newcommand{\nai}[1]{}
\newcommand{\km}[1]{}
\newcommand{\takashi}[1]{}
\newcommand{\revise}[1]{#1}

\ifnum\submission=0
\title{A Black-Box Approach to Post-Quantum Zero-Knowledge\\ in Constant Rounds}
\else
\title{A Black-Box Approach to Post-Quantum Zero-Knowledge in Constant Rounds}
\fi
\begin{document}

\ifnum\submission=0
\newcommand*{\email}[1]{\normalsize\href{mailto:#1}{#1}}
\author[1]{Nai-Hui Chia}
\author[2]{Kai-Min Chung}
\author[3]{Takashi Yamakawa}
\affil[1]{QuICS, University of Maryland \email{nchia@umd.edu}}
\affil[2]{Institute of Information Science, Academia Sinica \email{kmchung@iis.sinica.edu.tw}}
\affil[3]{NTT Secure Platform Laboratories \email{ takashi.yamakawa.ga@hco.ntt.co.jp}}
\else
\author{Nai-Hui Chia\inst{1,2} \and
Kai-Min Chung\inst{3} \and
Takashi Yamakawa\inst{4}}
\authorrunning{N. Chia et al.}
%
\institute{
QuICS, University of Maryland
\and 
Luddy School of Informatics, Computing, and Engineering, Indiana University Bloomington
\email{naichia@iu.edu}
\and
Institute of Information Science, Academia Sinica 
\email{kmchung@iis.sinica.edu.tw}
\and
NTT Secure Platform Laboratories 
\email{takashi.yamakawa.ga@hco.ntt.co.jp}}
\fi

\maketitle

\vspace{-7mm} 

\begin{abstract}
In a recent seminal work, Bitansky and Shmueli (STOC '20) gave the first construction of a constant round zero-knowledge argument for $\NP$ secure against quantum attacks. 
However, their construction has several drawbacks compared to the classical counterparts. 
Specifically, their construction only achieves computational soundness, requires strong assumptions of quantum hardness of learning with errors (QLWE assumption) and the existence of quantum fully homomorphic encryption (QFHE), and relies on non-black-box simulation.    

In this paper, we resolve these issues at the cost of weakening the notion of zero-knowledge to what is called $\epsilon$-zero-knowledge.
Concretely, we construct the following protocols:
\begin{itemize}
    \item We construct a constant round interactive proof for $\NP$ that satisfies \emph{statistical} soundness and  \emph{black-box} $\epsilon$-zero-knowledge against quantum attacks assuming the existence of \emph{collapsing hash functions}, which is a quantum counterpart of collision-resistant hash functions.
   Interestingly, this construction is just an adapted version of the classical protocol by Goldreich and Kahan (JoC '96) though the proof of $\epsilon$-zero-knowledge  property against quantum adversaries requires novel ideas.
      \item We construct a constant round interactive argument for $\NP$ that satisfies computational soundness and  \emph{black-box} $\epsilon$-zero-knowledge against quantum attacks only assuming the existence of post-quantum one-way functions.
\end{itemize}
At the heart of our results is a new quantum rewinding technique that enables a simulator to extract a committed message of a malicious verifier while simulating verifier's internal state in an appropriate sense.
\end{abstract}
\ifnum\submission=0
\thispagestyle{empty}
 \newpage
 \setcounter{page}{1}    
\fi

\ifnum\submission=0 
\section{Introduction}
\paragraph{Zero-Knowledge Proof.}
Zero-knowledge (ZK) proof \cite{GolMicRac89} is a fundamental cryptographic primitive, which enables a prover to convince  a verifier of a statement without giving any additional ``knowledge" beyond that the statement is true. 
In the classical setting, there have been many feasibility results on ZK proofs for specific languages including quadratic residuosity \cite{GolMicRac89}, graph isomorphism \cite{JACM:GMW91}, statistical difference problem \cite{SahVad03} etc., and for all $\NP$ languages assuming the existence of one-way functions (OWFs) \cite{JACM:GMW91,Blum86}.
On the other hand, van de Graaf \cite{Thesis:VanDeGraaf} pointed out that there is a technical difficulty to prove security of these protocols against quantum attacks. 
Roughly, the difficulty comes from the fact that security proofs of these results are based on a technique called \emph{rewinding}, which cannot be done when an adversary is  quantum due to the no-cloning theorem. Watrous \cite{SIAM:Watrous09} considered \textsl{post-quantum ZK proof}, which means a classical interactive proof that satisfies (computational) zero-knowledge property against quantum malicious verifiers, and showed that some of the classical constructions above are also post-quantum ZK. Especially, he introduced a new \emph{quantum rewinding technique} which is also applicable to quantum adversaries and proved that 3-coloring protocol of Goldreich, Micali, and Wigderson \cite{JACM:GMW91} is secure against quantum attacks assuming that the underlying OWF is post-quantum secure, i.e.,  uninvertible in quantum polynomial-time (QPT).\footnote{Strictly speaking, Watrous' assumption is a  statistically binding and post-quantum computationally hiding commitment scheme, and he did not claim that this can be constructed under the existence of post-quantum OWFs. However, we can see that  such a commitment scheme can be obtained by instantiating the construction of \cite{JC:Naor91,SICOMP:HILL99} with a post-quantum OWF.}
Since the 3-coloring problem is $\NP$-complete, this means that there exists a post-quantum ZK proof for all $\NP$ languages  assuming the existence of post-quantum OWFs.  

\paragraph{Round Complexity.}
An important complexity measure of ZK proofs is \emph{round complexity}, which is the number of interactions between a prover and verifier. 
In this aspect, the 3-coloring protocol \cite{JACM:GMW91} (and its quantumly secure version \cite{SIAM:Watrous09}) is not satisfactory since that requires super-constant number of rounds.\footnote{3-round suffices for achieving a constant soundness error, but super-constant times sequential repetitions are needed for achieving negligible soundness error (i.e., a cheating prover can let a  verifier accept on a false statement only with a negligible probability).
Negligible soundness error is a default requirement in this paper.} 
Goldreich and Kahan \cite{JC:GolKah96} gave the first construction of a constant round ZK proof for $\NP$ assuming the existence of collision-resistant hash function in the classical setting. 
However, Watrous' rewinding technique does not seem to work for this construction (as explained in Sec. \ref{sec:overview}), and it has been unknown if their protocol is secure against quantum attacks.   

Recently, Bitansky and Shmueli \cite{STOC:BitShm20} gave the first construction of post-quantum ZK \revise{\emph{argument} \cite{EC:BraCre89}} for $\NP$, which is a weakened version of post-quantum ZK proof where soundness holds only against computationally bounded adversaries.
In addition to weakening soundness to computational one, there are several drawbacks compared to classical counterparts.
First, they assume strong assumptions of quantum hardness of learning with erros (QLWE assumption) \cite{JACM:Regev09} and the existence of quantum fully homomorphic encryption (QFHE) \cite{FOCS:Mahadev18b,C:Brakerski18}.
Though the QLWE assumption is considered fairly standard due to reductions to worst-case lattice problems \cite{JACM:Regev09,STOC:Peikert09,STOC:BLPRS13}, a construction of  QFHE requires circular security of an QLWE-based encryption scheme, which has no theoretical evidence.
In contrast, a constant round \emph{classical} ZK argument for $\NP$ is known to exist under the minimal assumption of the existence of OWFs \cite{C:FeiSha89,TCC:PasWee09}. 
Second, their security proof of quantum ZK property relies on a novel \emph{non-black-box} simulation technique, which makes use of the actual description of malicious verifier instead of using it as a black-box.
In contrast, classical counterparts can be obtained by black-box simulation \cite{C:FeiSha89,JC:GolKah96,TCC:PasWee09}.
Therefore, it is of theoretical interest to ask if we can achieve constant round quantum ZK by black-box simulation.
Third, somewhat related to the second issue, their construction also uses building blocks in a non-black-box manner, which makes the actual efficiency of the protocol far from practical. 
Again, classical counterparts are known based on black-box constructions \cite{JC:GolKah96,TCC:PasWee09}. 

Given the state of affairs, it is natural to ask the following questions:
\begin{enumerate}
    \item Are there constant round post-quantum ZK proofs for $\NP$ instead of arguments?
    \item Are there constant round post-quantum ZK proofs/arguments for $\NP$ from weaker assumptions than those in \cite{STOC:BitShm20}?
    \item Are there constant round post-quantum ZK proofs/arguments for $\NP$ based on black-box simulation and/or black-box construction?
    \item Are known constructions of constant round classical ZK proofs/arguments for $\NP$ (e.g., \cite{C:FeiSha89,JC:GolKah96,TCC:PasWee09})  secure against quantum attacks if we instantiate them with post-quantum building blocks?
\end{enumerate}

\subsection{Our Results}
In this work, we partially answer the above questions affirmatively at the cost of weakening the quantum ZK property to \emph{quantum $\epsilon$-ZK}, which is the quantum version of $\epsilon$-ZK introduced in \cite{DNS04}.\footnote{$\epsilon$-ZK was originally called $\epsilon$-knowledge, but some later works \cite{STOC:BitKalPan18,EC:FleGoyJai18} call it $\epsilon$-ZK. We use $\epsilon$-ZK to clarify that this is a variant of ZK.}
\paragraph{Quantum $\epsilon$-Zero-Knowledge.}
The standard quantum ZK property roughly requires that for any QPT $\ver^*$, there exists a QPT simulator $\mathcal{S}$ that simulates the interaction between $\ver^*$ and an honest prover so that the simulation is indistinguishable from the real execution against any QPT distinguishers.
On the other hand, in quantum $\epsilon$-ZK, a simulator is allowed to depend on a  ``accuracy parameter" $\epsilon$.  
That is, it requires that for any QPT malicious verifier $\ver^*$ and a noticeable accuracy parameter $\epsilon$, there exists a QPT simulator $\mathcal{S}$ \emph{whose running time polynomially depends on $\epsilon^{-1}$}  
that simulates the interaction between $\ver^*$ and an honest prover so that no QPT distinguisher can distinguish it from real execution with advantage larger than $\epsilon$.
Though this is a significant relaxation of quantum ZK, this still captures meaningful security.
For example, we can see that quantum $\epsilon$-ZK implies both quantum versions of witness indistinguishability and witness hiding similarly to the analogous claims in the classical setting \cite{STOC:BitKhuPan19}.\footnote{Actually, \cite{STOC:BitKhuPan19} shows that even weaker notion called \emph{weak ZK} suffices for witness indistinguishability and witness hiding. See also Sec. \ref{sec:related_work}.}
Moreover, by extending  the observation in \cite{DNS04} to the quantum setting, we can see the following:
Suppose that a QPT malicious verifier solves some puzzle whose solution is efficiently checkable (e.g., finding a witness of an $\NP$ statement) after an interaction between an honest prover. Then, quantum $\epsilon$-ZK implies that if the verifier succeeds in solving the puzzle with noticeable probability $p$ after the interaction, then there is a QPT algorithm (whose running time polynomially depends on $p^{-1}$) that solves the same puzzle with noticeable probability (say, $p/2$) \emph{without interacting with the honest prover}. 
This captures the naive intuition of the ZK property that ``anything that can be done after the execution can be done without execution" in some sense, and this would be sufficient in many cryptographic applications.  Thus we believe that quantum $\epsilon$-ZK is conceptually a similar notion to the standard quantum ZK.
More discussion on (quantum) $\epsilon$-ZK and other related notions of ZK can be found in Sec. \ref{sec:related_work}.

\paragraph{Our Constructions.}
We give two constructions of constant round quantum $\epsilon$-ZK protocols.
\begin{itemize}
    \item We construct a constant round  quantum $\epsilon$-ZK \emph{proof} for $\NP$ assuming the existence of \emph{collapsing hash functions} \cite{EC:Unruh16,AC:Unruh16}, which is considered as a counterpart of collision-resistant hash functions in the quantum setting.
   Especially, we can instantiate the construction based on   the QLWE assumption. 
    Our construction is fully black-box in the sense that both simulation and construction rely on black-box usage of  building blocks and a malicious verifier.  
   Interestingly, this construction is just an adapted version of the classical protocol of \cite{JC:GolKah96} though the proof of quantum $\epsilon$-zero-knowledge  property  requires novel ideas.
      \item 
      We construct a constant round  quantum $\epsilon$-ZK argument for $\NP$ assuming the minimal assumption of the existence of \emph{post-quantum OWFs}. 
      This construction relies on black-box simulation, but the construction itself is non-black-box.
\end{itemize}
At the heart of our results is a new quantum rewinding technique that enables a simulator to extract a committed message of a malicious verifier while simulating verifier's internal state in some sense.
We formalize this technique as an \emph{extraction lemma}, which we believe is of independent interest.

\subsection{Technical Overview}\label{sec:overview}
Though we prove a general lemma which we call extraction lemma (Lemma \ref{lem:extraction}) and then prove quantum $\epsilon$-ZK  of our constructions based on that in the main body, we directly explain the proof of quantum $\epsilon$-ZK without going through such an abstraction in this overview. 
\paragraph{Known Classical Technique and Difficulty in Quantum Setting.}
First, we review a classical constant round ZK proof by Goldreich and Kahan \cite{JC:GolKah96} (referred to as GK protocol in the following), and explain why it is difficult to prove quantum ZK for this protocol by known techniques.
GK protocol is based on a special type of $3$-round proof system called $\Sigma$-protocol.\footnote{In this paper, we use $\Sigma$-protocol to mean a parallel repetition version where soundness error is reduced to negligible.}
In a $\Sigma$-protocol, a prover sends the first message $a$, a verifier sends the second message $e$ referred to as a \emph{challenge}, which is just a public randomness, and the prover sends the third message $z$.
A $\Sigma$-protocol satisfies a special type of honest-verifier ZK, which ensures that if a challenge $e$ is fixed, then one can simulate the transcript $(a,e,z)$ without using a witness. Though this may sound like almost the standard ZK property, a difficulty when proving ZK is that a malicious verifier may \emph{adaptively} choose $e$ depending on $a$, and thus we cannot fix $e$ at the beginning.
To resolve this issue, the idea of GK protocol is to let the verifier commit to a challenge $e$ at the beginning of the protocol.
That is, GK protocol roughly proceeds as follows:\footnote{We note that this construction is based on an earlier work of \cite{BCY91}.}
\begin{enumerate}
    \item A verifier sends a commitment $\com$ to a challenge $e$ of a $\Sigma$-protocol. \label{step:GK1}
    \item The prover sends the first message $a$ of the $\Sigma$-protocol.
    \label{step:GK2}
    \item The verifier opens $\com$ to open a challenge $e$ and its opening information $r$ (i.e., the randomness used for the commitment). \label{step:GK3}
    \item The prover aborts if the verifier's opening is invalid. Otherwise it sends the third message $z$ of the $\Sigma$-protocol.
\end{enumerate}
When proving the ZK property of GK protocol, they rely on a \emph{rewinding} argument. 
That is, a simulator first runs the protocol with a malicious verifier until Step \ref{step:GK3} to extract a committed message $e$ inside $\com$, and then rewind the verifier's state back to just after Step \ref{step:GK1}, and then simulates the transcript by using the extracted knowledge of $e$.  

On the other hand, this strategy does not work if we consider a quantum malicious verifier since  a quantum malicious verifier may perform measurements in Step \ref{step:GK3}, which is in general  not reversible. 
In other words, since we cannot copy the verifier's internal state after Step \ref{step:GK1} due to the no-cloning theorem, we cannot recover that state after running the protocol until Step \ref{step:GK3}.

Watrous \cite{SIAM:Watrous09} proved that we can apply a rewinding argument for quantum verifiers under a certain condition.
Roughly speaking, the condition is that there is a simulator that succeeds in simulation for quantum verifiers with a fixed (verifier-independent) and noticeable probability. 
For example, if the challenge space is polynomial size, then a simulator that simply guesses a challenge $e$ suffices.
However, for achieving negligible soundness error, the challenge space should be super-polynomial size, in which case it seems difficult to construct such a simulator.
Also, relaxing quantum ZK to quantum $\epsilon$-ZK does not seem to resolve the issue in any obvious way. 

\subsubsection{Quantum Analysis of GK Protocol.}
In spite of the above mentioned difficulty, we succeed in proving quantum $\epsilon$-ZK for a slight variant of GK protocol. In the following, we explain the idea for our results.
\paragraph{Simplified Goal: Simulation of Non-Aborting Case.}
First, we apply a general trick introduced in \cite{STOC:BitShm20}, which  simplifies the task of proving quantum ZK.
In GK protocol, we say that a verifier aborts if it fails to provide a valid opening to $\com$ in Step \ref{step:GK3}.
Then, for proving quantum ZK of the protocol, it suffices to construct two simulators $\siml_{\abort}$ and $\siml_{\nonabort}$ that work only when the verifier aborts and does not abort and they do not change the probability that the verifier aborts too much, respectively.
The reason is that if we randomly choose either of these two simulators and just run the chosen one, then the simulation succeeds with probability $1/2$ since the guess of if the verifier aborts is correct with probability $1/2$. Then, we can apply Watrous' rewinding technique to convert it to a full-fledged simulator.
Essentially the same trick also works for quantum $\epsilon$-ZK.

Moreover, it is easy to construct $\siml_{\abort}$ because the first message of a $\Sigma$-protocol can be simulated without witness, and one need not provide the third message to the verifier when it aborts. 
Therefore, the problem boils down to constructing a simulator $\siml_{\nonabort}$ that works only when the verifier does not abort.

\paragraph{Initial Observations.}
For explaining how to construct $\siml_{\nonabort}$, we start by considering the simplest case where a verifier never aborts.
Moreover, suppose that the commitment scheme used for committing to a challenge $e$ satisfies the strict-binding property \cite{EC:Unruh12}, i.e., for any commitment $\com$, there is at most one valid message and randomness. 
Then, a rewinding strategy similar to the classical case works since, in this case, the verifier's message in Step \ref{step:GK3} is information-theoretically determined, and such a deterministic computation does not collapse a quantum state in general.\footnote{This is also observed in \cite{STOC:BitShm20}.}
However, for ensuring statistical soundness, we have to use a statistically hiding commitment, which cannot be strict-binding. 
Fortunately, this problem can be resolved by using \emph{collapse-binding} commitments \cite{EC:Unruh16}, which roughly behave similarly to strict-binding commitments  for any 
\emph{computationally bounded} adversaries.\footnote{Strictly speaking, we need to use a slightly stronger variant of collapse-binding commitments which we call \emph{strong collapse-binding} commitments. Such commitments can be constructed under the QLWE assumption or the existence of collapsing hash functions in more general. See Sec. \ref{sec:commitment} 
\ifnum\submission=1 
\ifnum\cameraready=0
and \ref{sec:appendix_commitment}
\fi 
\fi 
for more details.} 
Since this is rather a standard technique, in the rest of this overview, we treat the commitment as if it satisfies the strict-binding property.

Next, we consider another toy example where a verifier sometimes aborts.
Suppose that a malicious verifier $\ver^*$ is given an initial state $\frac{1}{\sqrt{2}}(\ket{\psi_\abort}+\ket{\psi_{\nonabort}})$ in its internal register $\regV$
where $\ket{\psi_\abort}$ and $\ket{\psi_{\nonabort}}$ are orthogonal, 
and runs as follows:
\begin{enumerate}
    \item $\ver^*$ randomly picks $e$, honestly generates a commitment $\com$ to $e$, and sends it to the prover (just ignoring the initial state).
    \item After receiving $a$, $\ver^*$ performs a projective measurement $\{\ket{\psi_\abort}\bra{\psi_\abort},I-\ket{\psi_\abort}\bra{\psi_\abort}\}$ on $\regV$, and immediately aborts if $\ket{\psi_\abort}\bra{\psi_\abort}$ is applied, and otherwise honestly opens $(e,r)$.  
    \item After completing the protocol, $\ver^*$ outputs its internal state in $\regV$.
 \end{enumerate}
It is trivial to construct a simulator for this particular $\ver^*$ since it just ignores prover's messages. 
But for explaining our main idea,  we examine what happens if we apply the same rewinding strategy as the classical case to the above verifier.
After getting a commitment $\com$ from $\ver^*$, a simulator sends a random $a$ to  $\ver^*$ to extract $e$.
Since we are interested in constructing a simulator that works in the non-aborting case, suppose that $\ver^*$ does not abort, i.e., sends back a valid opening $(e,r)$.
At this point, $\ver^*$'s internal state collapses to $\ket{\psi_{\nonabort}}$.
Then the simulator cannot ``rewind" this state to the original verifier's state $\frac{1}{\sqrt{2}}(\ket{\psi_\abort}+\ket{\psi_{\nonabort}})$ in general, and thus the simulation seems to get stuck.
However, our key observation is that, 
conditioned on that $\ver^*$ does not abort, $\ver^*$'s state always collapses to $\ket{\psi_{\nonabort}}$ even in the real execution.
Since our goal is to construct $\siml_\nonabort$ that is only required to work for the non-aborting case, it does not matter if $\ver^*$'s state collapses to $\ket{\psi_{\nonabort}}$ when the simulator runs extraction.
More generally, extraction procedure may collapse verifier's internal state if a similar collapsing happens even in the real execution conditioned on that the verifier does not abort.  

\paragraph{Our Idea: Decompose Verifier's Space}
To generalize the above idea, we want to decompose verifier's internal state after Step \ref{step:GK1} into \emph{aborting part} and \emph{non-aborting part}. 
However, the definition of such a decomposition is non-trivial since a verifier may determine if it aborts depending on the prover's message $a$ in addition to its internal state.
Therefore, instead of decomposing it into always-aborting part and always-non-aborting part as in the example of the previous paragraph,  we set a noticeable threshold $t$ and decompose it into ``not-abort-with-probability $<t$ part" and ``not-abort-with-probability $\geq t$ part" over the randomness of $a$.

For implementing this idea, we rely on Jordan's lemma (e.g., see a lecture note by Regev \cite{Regev_lecture}) in a similar way to the work by Nagaj, Wocjan, and Zhang \cite{NWZ09} on the amplification theorem for $\QMA$.
Let $\Pi$ be a projection that corresponds to ``Step \ref{step:GK2} + Step \ref{step:GK3} + Check if the verifier does not abort" in GK protocol.
A little bit more formally, let $\regV$ be a register for verifier's internal state and $\regaux$ be an auxiliary register. 
Then  $\Pi$ is a projection over $\regV\ot \regaux$ that  works as follows:
\begin{enumerate}
    \item 
    Apply a unitary $U_{\mathsf{aux}}$ over $\regaux$ that maps $\ket{0}_\regaux$ to $\frac{1}{\sqrt{|\calR|}}\sum_{\rand \in \calR}\ket{\rand,a_{\rand}}_{\regaux}$ where $\calR$ is the randomness space to generate the first message of the $\Sigma$-protocol and $a_{\rand}$ is the first message derived from the randomness $\rand$.\footnote{$\regaux$ stores multiple qubits, but we denote by $\ket{0}_\regaux$ to mean $\ket{0^\ell}_\regaux$ for the appropriate length $\ell$  for notational simplicity.}
    \item Apply a unitary $U_V$ that corresponds to Step \ref{step:GK3} 
    for prover's message $a_{\rand}$ in $\regaux$ 
    except for measurement, \label{step:projection2}
    \item Apply a projection to the subspace spanned by states that contain  valid opening $(e,r)$ for $\com$ in designated output registers,
    \item 
    Apply $(U_V U_{\mathsf{aux}})^\dagger$.
\end{enumerate}
One can see that the probability that the verifier does not abort (i.e., sends a valid opening) is $\|\Pi\ket{\psi}_\regV\ket{0}_\regaux\|^2$ where $\ket{\psi}_\regV$ is verifier's internal state after Step \ref{step:GK1}.
Then  Jordan's lemma  gives 
an orthogonal decomposition of the Hilbert space of $\regV\ot\regaux$ into many one- or two-dimensional subspaces $S_1,...,S_N$ that are invariant under $\Pi$ and $\ket{0}_\regaux\bra{0}_\regaux$ such that we have the following:
\begin{enumerate}
    \item 
    For any 
    $j\in [N]$ and $\ket{\psi_j}_{\regV}\ket{0}_{\regaux}\in S_j$, the projection $\Pi$ succeeds with probability $p_j$, i.e., $\|\Pi\ket{\psi_j}_\regV\ket{0}_\regaux\|^2=p_j$.
    \item 
    A success probability of projection $\Pi$ is ``amplifiable" in each subspace.
    That is, there is an ``amplification procedure" $\Amp$ that maps any $\ket{\psi_j}_{\regV}\ket{0}_{\regaux}\in S_j$ to $\Pi\ket{\psi_j}_{\regV}\ket{0}_{\regaux}$ with overwhelming probability within $\poly(\secpar,p_j^{-1})$ times iteration of the same procedure (that does not depend on $j$) for any $j\in[N]$.
    Moreover, this procedure does not cause any interference between different subspaces.
\end{enumerate}


Then we define two subspaces 
\[
S_{<t}:=\bigoplus_{j:p_j<t}S_j,~~~ S_{\geq t}:=\bigoplus_{j:p_j\geq t}S_j.
\]
Then for any $\ket{\psi}_\regV$, we can decompose it as
\[
\ket{\psi}_\regV=\ket{\psi_{<t}}_\regV+ \ket{\psi_{\geq t}}_\regV
\]
by using (sub-normalized) states 
$\ket{\psi_{<t}}_\regV$ and $\ket{\psi_{\geq t}}_\regV$ such that 
$\ket{\psi_{<t}}_\regV\ket{0}_\regaux \in S_{<t}$ and $\ket{\psi_{\geq t}}_\regV\ket{0}_\regaux \in S_{\geq t}$. 
In this way, we can formally define a decomposition of verifier's internal state into 
``not-abort-with-probability $<t$ part" and ``not-abort-with-probability $\geq t$ part". 

\paragraph{Extraction and Simulation.}
Then we explain how we can use the above decomposition to implement  extraction of $e$ for simulation of non-aborting case.
First, we consider an easier case where the verifier's state after Step \ref{step:GK1} only has $S_{\geq t}$ component $\ket{\psi_{\geq t}}_\regV$.
In this case, we can use $\Amp$ to map  $\ket{\psi_{\geq t}}_\regV\ket{0}_\regaux$ onto the span of $\Pi$ within  $\poly(\secpar,t^{-1})$ times iteration. 
After mapped to $\Pi$, we can extract $(e,r)$ without collapsing the state by the definition of $\Pi$ and our assumption that the commitment is strict-binding.
This means that given $\ket{\psi_{\geq t}}_\regV$, we can extract $(e,r)$, which is information theoretically determined by $\com$, with overwhelming probability.
In general, such a deterministic computation can be implemented in a reversible manner, and thus we can extract $(e,r)$ from $\ket{\psi_{\geq t}}_\regV$ almost without damaging the state.

On the other hand, the same procedure does not work for $\ket{\psi_{< t}}_\regV$ since $\poly(\secpar,t^{-1})$ times iteration is not sufficient for amplifying the success probability of $\Pi$ to overwhelming in this subspace. 
Our idea is to let a simulator run the above extraction procedure in superposition even though $S_{<t}$ component may be damaged.

Specifically, our extraction procedure $\ext$ works as follows:
\begin{enumerate}
    \item Given a verifier's internal state $\ket{\psi}_\regV$ after Step \ref{step:GK1}, initialize $\regaux$ to $\ket{0}_\regaux$ and runs $\Amp$ for  $\poly(\secpar,t^{-1})$ times iteration.
    Abort if a mapping onto $\Pi$ does not succeed. Otherwise, proceed  to the next step.
    \label{Step_extraction_1}
    \item Apply $U_V U_{\mathsf{aux}}$, measure designated output registers to obtain $(e_\ext,r_\ext)$, and apply $(U_V U_{\mathsf{aux}})^\dagger$.
    We note that $(e_\ext,r_\ext)$ is always a valid opening of $\com$ since $\ext$ runs this step only if it succeeds in mapping the state onto $\Pi$ in the previous step.
    We also note that this step does not collapse the state at all by the strict-binding property of the commitment.
    \item Uncompute Step \ref{Step_extraction_1} and measure $\regaux$.
    Abort if  the measurement outcome is not $0$. Otherwise, proceed to the next step.
    \item Output the extracted opening $(e_\ext,r_\ext)$ along with a ``post-extraction state"  $\ket{\psi'}_\regV$ in register $\regV$.
    For convenience, we express $\ket{\psi'}_\regV$ as a sub-normalized state whose norm is the probability that $\ext$ does not abort and the post-extraction state conditioned on that the extraction succeeds is $\frac{\ket{\psi'}_\regV}{\|\ket{\psi'}_\regV\|}$.
\end{enumerate}

In the following, we analyze $\ext$. 
We consider the decomposition of $\ket{\psi}_\regV$ as defined in the previous paragraph:
\[
\ket{\psi}_\regV=\ket{\psi_{<t}}_\regV+ \ket{\psi_{\geq t}}_\regV.
\]
Suppose that $\ext$ does not abort, i.e., it outputs a valid opening $(e_\ext,r_\ext)$ along with a post-extraction state $\ket{\psi'}_\regV$.
Then, $\ket{\psi'}_\regV$ can be expressed as  
\[
\ket{\psi'}_\regV=\ket{\psi'_{<t}}_\regV+ \ket{\psi'_{\geq t}}_\regV
\]
for some $\ket{\psi'_{<t}}_\regV$ and $\ket{\psi'_{\geq t}}_\regV$ such that $\ket{\psi'_{<t}}_\regV\ket{0}_\regaux \in S_{<t}$,  $\ket{\psi'_{\geq t}}_\regV\ket{0}_\regaux \in S_{\geq t}$, and $\ket{\psi_{\geq t}}_\regV\approx \ket{\psi'_{\geq t}}_\regV$ since there is no interference between $S_{<t}$ and $S_{\geq t}$ when running $\Amp$ and $S_{\geq t}$ component hardly changes as observed above.
This is not even a close state to the original state $\ket{\psi}_\regV$ in general since the $S_{<t}$ component may be completely different. 
However, our key observation is that, conditioned on that the verifier does not abort, at most ``$t$-fraction" of $S_{<t}$ component survives even in the real execution by the definition of the subspace $S_{<t}$. 
That is, in the verifier's final output state conditioned on that it does not abort, the average squared norm of a portion that comes from $S_{<t}$ component is at most $t$.
Thus, even if a simulator fails to simulate this portion, this only impacts the accuracy of the simulation by a certain function of $t$, which is shown to be $O(\sqrt{t})$ in the main body.

With this observation in mind, 
 the non-aborting case simulator $\siml_{\nonabort}$ works as follows.
 \begin{enumerate}
     \item Run Step \ref{step:GK1} of the verifier to obtain $\com$ and let $\ket{\psi}_\regV$ be verifier's internal state at this point.
     \item Run $\ext$ on input $\ket{\psi}_\regV$.
     Abort     if $\ext$ aborts.
     Otherwise,  obtain an extracted opening $(e_\ext,r_\ext)$ and a post-extraction state $\ket{\psi'}_\regV$, and proceed to the next step.
     \item  Simulate a transcript $(a,e_\ext,z)$ by the honest-verifier ZK property of the $\Sigma$-protocol.
     \item Send $a$ to the verifier whose internal state is replaced with  $\ket{\psi'}_\regV$. Let $(e,r)$ be the verifier's response.
     Abort if $(e,r)$ is not a valid opening to $\com$. 
     Otherwise send $z$ to the verifier.
     \item Output the verifier's final output.
 \end{enumerate}
 By the above analysis, we can see that $\siml_\nonabort$'s output distribution is close to the real verifier's output distribution with an approximation error $O(\sqrt{t})$   conditioned on that the verifier does not abort. Furthermore, the probability that the verifier does not abort can only be changed by at most $O(\sqrt{t})$. 
If we could set $t$ to be a negligible function, then we would be able to achieve quantum ZK rather than quantum $\epsilon$-ZK.
However, since we have to ensure that $\Amp$'s running time $\poly(\secpar,t^{-1})$ is polynomial in $\secpar$, we can only set $t$ to be noticeable.
Since we can set $t$ to be an arbitrarily small noticeable function,  we can make the approximation error $O(\sqrt{t})$ be an arbitrarily small noticeable function.
This means that the protocol satisfies quantum $\epsilon$-ZK. 

\paragraph{Black-Box Simulation.}
So far, we did not pay attention to the black-box property of simulation.
We briefly explain the definition of black-box quantum ZK and that our simulator satisfies it. 
First, we define black-box quantum ZK by borrowing the definition of quantum oracle machine by Unruh \cite{EC:Unruh12}.
Roughly, we say that a simulator is black-box if it only accesses unitary part of a verifier and its inverse in a black-box manner, and does not directly act on the verifier's internal registers.
With this definition, one part where it is unclear if our simulator is black-box is the amplification procedure $\Amp$. 
However, by a close inspection, we can see that $\Amp$ actually just performs sequential measurements $\{\Pi,I_{\regV,\regaux}-\Pi\}$ and $\{\ket{0}_{\regaux}\bra{0}_\regaux,I_{\regV,\regaux}-\ket{0}_{\regaux}\bra{0}_\regaux\}$, which can be done by black-box access to the verifier as seen from the definition of $\Pi$. 
Therefore, we can see that our simulator is black-box.

\paragraph{A Remark on Underlying $\Sigma$-Protocol.}
In the original GK protocol, any $\Sigma$-Protocol can be used as a building block.
However, in our technique, we need to use \emph{delayed-witness} $\Sigma$-protocol where the first message $a$ can be generated without knowledge of a witness due to a technical reason.
An example of delayed-witness $\Sigma$-protocol is Blum's Graph Hamiltonicity protocol \cite{Blum86}. 
Roughly, the reason to require this additional property is for ensuring that a simulator can perfectly simulate the first message $a$ of the $\Sigma$-protocol when running the extraction procedure.
In the classical setting, a computationally indistinguishable simulation of $a$ works, but we could not prove an analogous claim in our setting. 

\subsubsection{OWF-based Construction.}
Next, we briefly explain our OWF-based quantum $\epsilon$-ZK argument. 
The reason why we need a stronger assumption in our first construction is that we need to implement the commitment for the challenge by a constant round statistically hiding commitment, which is not known to exist from OWF. 
Then, a natural idea is to relax it to computationally hiding one if we only need computational soundness.
We can show that the extraction technique as explained above also works for statistically binding commitments with a small tweak.
However, we cannot prove soundness of the protocol without any modification due to a malleability issue.
For explaining this, we recall that the first message $a$ of a $\Sigma$-protocol itself is also implemented as a commitment. 
Then, the computational hiding of commitment does not prevent a computationally bounded prover, which is given a commitment $\com$ to $e$, from generating a ``commitment" $a$ whose committed message depends on $e$.
Such a dependence leads to an attack against soundness.
To prevent this, an extractable commitment scheme is used to generate $a$ in the classical setting \cite{TCC:PasWee09}.
However, since it is unclear if the extractable commitment scheme used in \cite{TCC:PasWee09} is secure against quantum adversaries, we take an alternative approach that we let a prover prove that it knows a committed message inside $a$ by using a proof of knowledge before a verifier opens a challenge as is done in \cite[Sec.4.9]{Foundation:Volume1} (see also
\cite[App.C.3]{Foundation:Volume2}).
A naive approach to implement this idea would be to use ZK proof of knowledge, but this does not work since a constant round ZK argument is what we are trying to construct.
Fortunately, we can instead use witness indistinguishable proof of knowledge (WIPoK) with a simple OR proof trick.
Specifically, we let a prover prove that ``I know committed message in $a$" OR ``I know witness $w$ for $x$" where $x$ is the statement being proven in the protocol.
In the proof of soundness, since we assume $x$ is a false statement, a witness for the latter statement does not exist.
Then we can extract a committed message inside $a$ to break the hiding property of the commitment scheme used by the verifier if the committed message depends on $e$. 
On the other hand, in the proof of $\epsilon$-ZK property, we can use the real witness $w$ in an intermediate hybrid to simulate WIPoK without using knowledge of a committed message. In such a hybrid, we can rely on honest-verifier ZK of the $\Sigma$-protocol to change $a$ to a simulated one for an extracted challenge $e$. 

Finally, we remark that though we are not aware of any work that explicitly claims the existence of a constant round WIPoK that works for quantum provers from OWFs, we observe that a combination of known works easily yields such a construction. (See 
\ifnum\cameraready=1
the full version 
\else
Sec. \ref{sec:wipok} 
\fi  
for more details.) 
As a result, we obtain constant round quantum $\epsilon$-ZK argument from OWFs.

\subsection{Related Work}\label{sec:related_work}
\paragraph{$\epsilon$-Zero-Knowledge and Related Notions.}
Though we are the first to consider $\epsilon$-ZK in the quantum setting, there are several works that consider $\epsilon$-ZK in the classical setting.
We briefly review them.
We note that all of these results are in the classical setting, and it is unknown if similar results hold in the quantum setting.
The notion of $\epsilon$-ZK (originally called $\epsilon$-knowledge) was introduced by Dwork, Naor, and Sahai \cite{DNS04} in the context of concurrent ZK proofs.
Bitansky, Kalai, and Paneth \cite{STOC:BitKalPan18} gave a construction of 4-round $\epsilon$-ZK proof  for $\NP$ assuming the existence of key-less multi-collision resistant hash function.\footnote{The protocol achieves full-fledged ZK if we allow the simulator to take  non-uniform advice or assume a super-polynomial assumption.}
Barak and Lindell \cite{STOC:BarLin02} showed the impossibility of constant round black-box ZK proof with strict-polynomial time simulation, and observed that strict-polynomial time simulation is possible if we relax ZK to $\epsilon$-ZK.
This can be understood as a theoretical separation between ZK and $\epsilon$-ZK.
On the other hand, Fleischhacker, Goyal, and Jain \cite{EC:FleGoyJai18} showed that there does not exist $3$-round $\epsilon$-ZK proof for $\NP$ even with non-black-box simulation under some computational assumptions, which is the same lower bound as that for ZK proofs if we allow non-black-box simulation.

Another relaxation of ZK is \emph{super-polynomial simulation (SPS)-ZK} \cite{EC:Pass03}, where a simulator is allowed to run in super-polynomial time.
One may find a similarity between $\epsilon$-ZK and SPS-ZK in the sense that the latter can be seen as a variant of $\epsilon$-ZK where we set the accuracy parameter $\epsilon$ to be negligible. 
On the other hand,  it has been considered that $\epsilon$-ZK is much more difficult to achieve than SPS-ZK.
For example, the work of Bitansky, Khurana, and Paneth \cite{STOC:BitKhuPan19} gave a construction of a 2-round argument for $\NP$ that achieves a weaker notion of ZK than $\epsilon$-ZK, and the result is considered a significant breakthrough in the area even though there is a simple construction of 2-round SPS-ZK argument for $\NP$ \cite{EC:Pass03}.

Several works considered other weakened notions of ZK \cite{DNRS03,TCC:BitPan12,TCC:ChuLuiPas15,C:JKKR17,STOC:BitKhuPan19}.
Some of them are weaker than $\epsilon$-ZK, and others are incomparable.
For example, ``weak ZK" in \cite{TCC:BitPan12,TCC:ChuLuiPas15} is incomparable to $\epsilon$-ZK whereas ``weak ZK" in \cite{STOC:BitKhuPan19} is weaker than $\epsilon$-ZK.

\paragraph{Post-Quantum Zero-Knowledge with Classical Computational Soundness.}
Ananth and La Placa \cite{TCC:AnaLap20} gave a construction of post-quantum ZK argument for $\NP$ with \emph{classical} computational soundness assuming the QLWE assumption. 
Though such a protocol would be easy to obtain if we assume average-case classical hardness of certain problems in $\BQP$ (e.g., factoring) in addition to the QLWE assumption, what is interesting in \cite{TCC:AnaLap20} is that they only assume the QLWE assumption.

\paragraph{Post-Quantum Zero-Knowledge with Trusted Setup.}
Several works studied (non-interactive) post-quantum ZK proofs for $\NP$ in the common random/reference string model \cite{Kobayashi03,C:DamFehSal04,C:PeiShi19}.
Among them, Peikert and Shiehian \cite{C:PeiShi19} proved that there exists non-interactive post-quantum ZK proof for $\NP$ in the common reference string model assuming the QLWE assumption.\footnote{In \cite{C:PeiShi19}, they do not explicitly claim ZK against quantum adversaries. However, since their security proof does not rely on rewinding, it immediately extends to post-quantum security if we assume the underlying assumption against quantum adversaries.}  

\paragraph{Zero-Knowledge for $\QMA$.}
The complexity class $\QMA$ is a quantum analogue of $\NP$.
Broadbent, Ji, Song, and Watrous \cite{BJSW20} gave a construction of a ZK proof for $\QMA$. 
Recently, Broadbent and Grilo \cite{FOCS:BroGri20} gave an alternative simpler construction of a ZK proof for $\QMA$.
Bitansky and Shmueli \cite{STOC:BitShm20} gave a constant round ZK argument for $\QMA$ by combining the construction of \cite{FOCS:BroGri20} and their post-quantum ZK argument for $\NP$.
We believe that our technique can be used to construct a constant round $\epsilon$-ZK proof for $\QMA$ by replacing the delayed-witness $\Sigma$-protocol for $\NP$ with the delayed-witness quantum $\Sigma$-protocol for $\QMA$ recently proposed by Brakerski and Yuen \cite{BraYue20}.\footnote{Actually, their protocol is delayed-input, i.e., the first message generation does not use the statement either.}
This is beyond the scope of this paper, and we leave a formal proof as a future work.

Several works studied non-interactive ZK proofs/arguments for $\QMA$ in preprocessing models \cite{C:ColVidZha20,FOCS:BroGri20,Shmueli20,TCC:ACGH20}.

\paragraph{Collapsing Hash Functions.}
The notion of collapsing hash functions was introduced by Unruh \cite{EC:Unruh16} for a replacement of collision-resistant hash functions in post-quantum setting.
Unruh \cite{AC:Unruh16} gave a construction of a collapsing hash function under the QLWE assumption. Actually, the construction is generic based on any lossy function  with sufficiently large ``lossy rate".\footnote{A lossy function is defined similarly to a lossy trapdoor function \cite{STOC:PeiWat08} except that we do not require the existence of trapdoor.}
Currently, we are not aware of any other construction of collapsing hash function based on standard assumptions, but any new construction of collapsing hash function yields a new instantiation of our first construction.

Zhandry \cite{EC:Zhandry19b} proved that any collision-resistant hash function that is not collapsing yields a stronger variant of public-key quantum money (with infinitely often security).
Given the difficulty of constructing public key quantum money, he suggested that most natural post-quantum collision-resistant hash functions are likely already collapsing.

\paragraph{Relation to \cite{TCC:ChiChuYam20}.}
Our idea of decomposing a verifier's internal space into ``aborting space" and ``non-aborting space" is inspired by a recent work of Chia, Chung, and Yamakawa \cite{TCC:ChiChuYam20}.
In \cite{TCC:ChiChuYam20}, the authors consider a decomposition of a prover's internal space into ``know-answer space" and ``not-know-answer space" to prove soundness of parallel repetition version of Mahadev's classical verification of quantum computation protocol \cite{FOCS:Mahadev18a}.
Though the conceptual idea and some technical tools are similar, the ways of applying them to actual problems are quite different.
For example, in our case, we need a careful analysis to make sure that a post-extraction state is close to the original one in some sense while such an argument does not appear in their work since their goal is proving soundness rather than ZK. 
On the other hand, their technical core is a approximated projection to each subspace, which is not needed in this paper. 

\paragraph{Subsequent work.}
Subsequently to this work, Chia, Chung, Liu, and Yamakawa \cite{CCLY21_arxiv} proved that there does not exist a constant round  post-quantum ZK argument for $\NP$ unless $\NP\in \mathbf{BQP}$, which is highly unlikely. 
This justifies the relaxation to $\epsilon$-ZK in our constructions.
\section{Preliminaries}
\paragraph{Basic Notations.}
We use $\secpar$ to denote the security parameter throughout the paper.
For a positive integer $n\in\mathbb{N}$, $[n]$ denotes a set $\{1,2,...,n\}$.
For a finite set $\calX$, $x\sample \calX$ means that $x$ is uniformly chosen from $\calX$.
A function $f:\mathbb{N}\ra [0,1]$ is said to be negligible if for all polynomial $p$ and sufficiently large $\secpar \in \mathbb{N}$, we have $f(\secpar)< 1/p(\secpar)$, said to be overwhelming if $1-f$ is negligible, and said to be noticeable if there is a polynomial $p$ such that we have $f(\secpar)\geq  1/p(\secpar)$ for sufficiently large $\secpar\in \mathbb{N}$.
We denote by $\poly$ an unspecified polynomial and by $\negl$ an unspecified negligible function.
We use PPT and QPT to mean (classical) probabilistic polynomial time and quantum polynomial time, respectively.
For a classical probabilistic or quantum algorithm $\A$, $y\sample \A(x)$ means that $\A$ is run on input $x$ and outputs $y$.
When $\A$ is classical probabilistic algorithm, we denote by $\A(x;r)$ to mean the execution of $\A$ on input $x$ and a randomness $r$.
When $\A$ is a quantum algorithm that takes a quantum advice, we denote by $\A(x;\rho)$ to mean the execution of $\A$ on input $x$ and an advice $\rho$.
For a quantum algorithm $\A$, a unitary part of $\A$ means the unitary obtained by deferring all measurements by $\A$ and omitting these measurements.
We use the bold font (like $\regX$) to denote quantum registers, and $\hil_\regX$ to mean the Hilbert space corresponding to the register $\regX$. 
For a quantum state $\rho$, $M_{\regX}\circ \rho$ means a measurement in the computational basis on the register $\regX$ of $\rho$.
For quantum states $\rho$ and $\rho'$, $\TD(\rho,\rho')$ denotes trace distance between them. 
For a pure state $\ket{\psi}$, $\|\ket{\psi}\|$ denotes its Euclidean norm. 
When we consider a sequence $\sequence{X}$ of some objects (e.g., bit strings, quantum states, sets, Hilbert spaces etc.) indexed by the security parameter $\secpar$, we often simply write $X$ to mean $X_\secpar$ or $\sequence{X}$, which will be clear from the context.
Similarly, for a function $f$ in the security parameter $\secpar$, we often simply write $f$ to mean $f(\secpar)$.

\paragraph{Standard Computational Models.}
\begin{itemize}
\item A PPT algorithm is a probabilistic polynomial time (classical) Turing machine.
A PPT algorithm is also often seen as a sequence of uniform polynomial-size circuits.
\item A QPT algorithm is a polynomial time quantum Turing machine. 
A QPT algorithm is also often seen as a sequence of uniform polynomial-size quantum circuits.
\item 
An adversary (or malicious party) is modeled as a non-uniform QPT algorithm $\A$ (with quantum advice) that is specified by sequences of polynomial-size quantum circuits $\{\A_\secpar\}_{\secpar\in\mathbb{N}}$ and polynomial-size quantum advice $\{\rho_\secpar\}_{\secpar\in \mathbb{N}}$.
When $\A$ takes an input of $\secpar$-bit, $\A$ runs $\A_{\secpar}$ taking $\rho_\secpar$ as an advice. 

\end{itemize}

\paragraph{Interactive Quantum Machine and Oracle-Aided Quantum Machine.}
We rely on the definition of an interactive quantum machine and oracle-aided quantum machine that is given oracle access to an interactive quantum machine following \cite{EC:Unruh12}.
Roughly, an interactive quantum machine $\A$ is formalized by a unitary over registers $\regM$ for receiving and sending messages  and $\regA$ for maintaining $\A$'s internal state.
For two interactive quantum machines $\A$ and $\B$ that share the same message register $\regM$, an interaction between $\A$ and $\B$ proceeds by alternating invocations of $\A$ and $\B$ while exchanging messages over  $\regM$.

An oracle-aided quantum machine $\mathcal{S}$ given oracle access to an interactive quantum machine $\A$ with an initial internal state $\rho$ (denoted by $\mathcal{S}^{\A(\rho)}$) is allowed to apply unitary part of $\A$ and its inverse in a black-box manner where $\mathcal{S}$ can act on $\A$'s internal register $\regA$ only through oracle access.
We refer to \cite{EC:Unruh12} for more formal definitions of interactive quantum machines and black-box access to them.


\paragraph{Indistinguishability of Quantum States.}
We define computational and statistical indistinguishability of quantum states similarly to \cite{STOC:BitShm20}.

We may consider random variables over bit strings or over quantum states. 
This will be clear from the context. 
For ensembles of random variables $\mathcal{X}=\{X_i\}_{\secpar\in \mathbb{N},i\in I_\secpar}$ and $\mathcal{Y}=\{Y_i\}_{\secpar\in \mathbb{N},i\in I_\secpar}$  over the same set of indices $I=\bigcup_{\secpar\in\mathbb{N}}I_\secpar$ and a function $\delta$,       
we write $\mathcal{X}\compind_{\delta}\mathcal{Y}$ to mean that for any non-uniform QPT algorithm $\A=\{\A_\secpar,\rho_\secpar\}$, there exists a negligible function $\negl$ such that for all $\secpar\in\mathbb{N}$, $i\in I_\secpar$, we have
\[
|\Pr[\A_\secpar(X_i;\rho_\secpar)]-\Pr[\A_\secpar(Y_i;\rho_\secpar)]|\leq \delta(\secpar) + \negl(\secpar).
\]
Especially, when we have the above for $\delta=0$, we say that $\calX$ and $\calY$ are computationally indistinguishable, and simply write $\calX\compind \calY$.

Similarly, we write $\calX\statind_{\delta}\calY$ to mean that for any unbounded time  algorithm $\A$, there exists a negligible function $\negl$ such that for all $\secpar\in\mathbb{N}$, $i\in I_\secpar$, we have 
\[
|\Pr[\A(X_i)]-\Pr[\A(Y_i)]|\leq \delta(\secpar) + \negl(\secpar).\footnote{In other words, $\calX\statind_{\delta}\calY$ means that there exists a negligible function $\negl$ such that the trace distance between $\rho_{X_i}$ and $\rho_{Y_i}$ is at most $\delta(\secpar) + \negl(\secpar)$ for all $\secpar\in \mathbb{N}$ and $i\in I_\secpar$ where $\rho_{X_i}$ and $\rho_{Y_i}$ denote density matrices corresponding to $X_{i}$ and $Y_{i}$.}
\]
Especially, when we have the above for $\delta=0$, we say that $\calX$ and $\calY$ are statistically indistinguishable, and simply write $\calX\statind \calY$.
Moreover, 
we write $\calX \equiv \calY$ to mean
that $X_i$ and $Y_i$ are distributed identically for all $i\in I$

\subsection{Post-Quantum One-Way Functions and Collapsing Hash Functions}
A post-quantum one-way function (OWF) is a classically computable function that is hard to invert in QPT. 
A collapsing hash function is a quantum counterpart of collision-resistant hash function introduced by Unruh \cite{EC:Unruh16}. 
Unruh \cite{AC:Unruh16} gave a construction of collapsing hash functions based on the QLWE assumption. 
We give formal definitions in  
\ifnum\cameraready=1
the full version
\else
Appendix
\ifnum\submission=0
 \ref{app:omitted_definition}  
\else
 \ref{app:omitted_definition_OWF_Collapsing}
\fi
\fi
since they are only used for constructing other cryptographic primitives and not directly used in our constructions.

\subsection{Commitment}\label{sec:commitment}
\ifnum\submission=1 
\revise{We use commitments in our constructions. 
Though they are mostly standard, we need one new security notion which we call \emph{strong collapse-binding}, which is a stronger variant of  collapse-biding introduced by Unruh \cite{EC:Unruh16}.
Roughly speaking, this security requires that for any superposition of messages and randomness corresponding the same commitment generated by an adversary, the adversary cannot distinguish if the message and randomness registers are measured or not.
The difference from the original collapse-binding property is that both message and randomness registers are measured rather than only the message register. 
We observe that the collapse-binding commitment  based on collapsing hash functions  in \cite{EC:Unruh16} also satisfies the strong collapse-binding property.
Especially, there exists a strong collapse-binding commitment under the QLWE assumption.
\ifnum\cameraready=1
See the full version for deails of the definition and construction of strong collapse-binding commitments. 
\else
See Sec. \ref{sec:appendix_commitment} for details of the definition and Sec. \ref{sec:strong_collapse_binding} for the construction of strong collapse-binding commitments.
\fi
}
\else 
We give definitions of commitments and their security.
Though mostly standard, we introduce one new security notion which we call \emph{strong collapse-binding}, which is a stronger variant of  collapse-biding introduced by Unruh \cite{EC:Unruh16}.
As shown in Appendix \ref{sec:strong_collapse_binding}, we can see that Unruh's construction of collapse-binding commitments actually also satisfies strong collapse-binding with almost the same (or even simpler) security proof.

\begin{definition}[Commitment.]\label{def:commitment}
A (two-message) commitment scheme   with message space $\calM$, randomness space $\calR$, commitment space $\COM$, and a public parameter space $\ppspace$  consists of two classical PPT algorithms $(\Setup,\Commit)$:
\begin{description}
\item[$\Setup(1^\secpar)$:] The setup algorithm takes the security parameter $1^\secpar$ as input and outputs a public parameter $\pp\in \ppspace$.
\item[$\Commit(\pp,m)$:] The committing algorithm takes a public parameter $\pp\in \ppspace$ and a message $m\in \calM$ as input and outputs a commitment $\com\in \COM$. 
\end{description}
We say that a commitment scheme is non-interactive if a public parameter $\pp$ generated by $\Setup(1^\secpar)$ is always just the security parameter $1^\secpar$. For such a scheme, we omit to write $\Setup$.

We define the following security notions for a commitment scheme.

\paragraph{Statistical/Computational Hiding.}
For an adversary $\A$, we consider an experiment $\Exp_{\A}^{\hiding}(1^\secpar)$ defined below:
\begin{enumerate}
    \item $\A$ is given the security parameter $1^\secpar$ and sends a (possibly malformed) public parameter $\pp\in\ppspace$ and $(m_0,m_1)\in \calM^2$ to the challenger 
    \item The challenger randomly picks $b\sample \bit$, computes $\com\sample \Commit(\pp,m_b)$, and sends $\com$ to $\A$. 
    \item $\A$ is given a commitment $\com$ and outputs $b'\in \bit$.
    The experiment outputs $1$ if $b=b'$ and $0$ otherwise.
\end{enumerate}
We say that a commitment scheme satisfies statistical (resp. computational) hiding if for any unbounded-time (resp. non-uniform QPT) adversary $\A$, we have
\begin{align*}
    |\Pr[1\sample \Exp_{\A}^{\hiding}(1^\secpar)]-1/2|=\negl(\secpar).
\end{align*}
\begin{remark}
Our definition of hiding requires that the security should hold even if $\pp$ is maliciously generated. Thus, the hiding property holds even if a receiver  runs the setup algorithm. 
\end{remark}

\paragraph{Binding.}
\begin{itemize}
\item \textbf{Perfect/Statistical/Computational Binding.}
We say that a non-interactive commitment scheme satisfies statistical (resp. computational) binding if for any unbounded-time (resp. non-uniform QPT) adversary $\A$, we have  
\ifnum\submission=0
\begin{align*}
    \Pr[\Commit(\pp,m;r)=\Commit(\pp,m';r')\land m\neq m' : \pp \sample \Setup(1^\secpar),(m,m',r,r')\sample \A(\pp)]=\negl(\secpar).
\end{align*}
\else
\begin{align*}
    \Pr\left[
    \begin{array}{ll}
     &\Commit(\pp,m;r)=\Commit(\pp,m';r')\\
     &\land~ m\neq m' 
    \end{array}
    : 
    \begin{array}{ll}
    &\pp \sample \Setup(1^\secpar),\\
    &(m,m',r,r')\sample \A(\pp)
    \end{array}
    \right]=\negl(\secpar).
\end{align*}
\fi
We say that a scheme satisfies perfect binding if the above probability is $0$ for all unbounded-time adversary $\A$.

\item \textbf{Strong Collapse-Binding.}
For an adversary $\A$, we define an experiment $\Exp_{\A}^{\clbinding}(1^\secpar)$ as follows:
\begin{enumerate}
    \item The challenger generates $\pp\sample \Setup(1^\secpar)$.
    \item $\A$ is given the public parameter $\pp$ as input and generates a commitment $\com\in\COM$ and a quantum state $\sigma$ over registers $(\regM,\regR,\regA)$ where
    $\regM$ stores an element of $\calM$, $\regR$ stores an element of $\calR$, and $\regA$ is $\A$'s internal register.
    Then it sends $\com$ and registers $(\regM,\regR)$ to the challenger, and keeps $\regA$ on its side.
    \item The challenger picks $b\sample \bit$ If $b=0$, the challenger does nothing and if $b=1$, the challenger measures registers $(\regM,\regR)$ in the computational basis.
    The challenger returns registers $(\regM,\regR)$ to $\A$
    \label{step:collapse_binding_game_measure}
    \item $\A$ outputs a bit $b'$. The experiment outputs $1$ if $b'=b$ and $0$ otherwise. 
\end{enumerate}
We say that $\A$ is a valid adversary if we we have 
\[
\Pr[\Commit(\pp,m;r)=\com:\pp\sample\setup(1^\secpar),(\com,\sigma)\sample \A(\pp),(m,r)\leftarrow M_{\regM,\regR}\circ \sigma]=1.
\]
We say that a  commitment is strongly collapse-binding if for any non-uniform QPT valid adversary $\A$,  we have
\[
|\Pr[1\sample \Exp_{\A}^{\clbinding}(1^\secpar)]-1/2|= \negl(\secpar).
\]
\end{itemize}
\begin{remark}\label{rem:statistical_to_collapse_binding}
The difference of strong collapse-binding  from the original collapse-binding  is that the challenger measures both registers $(\regM,\regR)$ in Step \ref{step:collapse_binding_game_measure} in the case of $b=1$ whereas the challenger of the original collapse-binding game only measures $\regM$.
We note that the statistical binding property immediately implies the (original) collapse-binding property, but it does not imply the strong collapse-binding property.
\end{remark}
\begin{remark}\label{rem:collapse_to_computational_binding}
One can easily see that the strong collapse-binding property implies the computational binding property. 
Indeed, if one can find $(m,r)\neq (m',r')$ such that $\Commit(\pp,m;r)=\Commit(\pp,m';r')=\com$, then we can break the strong collapse-binding property by sending 
$\com$ and
$\ket{\psi}:=\frac{1}{\sqrt{2}}(\ket{m,r}_{\regM,\regR}+\ket{m',r'}_{\regM,\regR})$ to the challenger and performing a measurement $(\ket{\psi}\bra{\psi},I-\ket{\psi}\bra{\psi})$ on the returned state to distinguish if the state is measured.
\end{remark}
\end{definition}

We introduce the following definition for convenience.
\begin{definition}[Binding Public Parameter.]\label{def:binding_pp}
We say that $\pp\in \ppspace$ is \emph{binding}
if for any commitment $\com\in \COM$, there is at most one $m\in\calM$ such that $\Commit(\pp,m;r)=\com$ for some $r\in \calR$.
\end{definition}
The following lemma is easy to see.
\begin{lemma}\label{lem:overwhelming_fraction_is_binding}
If a commitment scheme is statistically binding, then overwhelming fraction of $\pp$ generated by $\Setup(1^\secpar)$ is binding.
\end{lemma}

We also consider an additional security definition.
\begin{definition}[Unpredictability.]\label{def:unpredictability}
For an adversary $\A$, we consider an experiment $\Exp_{\A}^{\mathsf{unpre}}(1^\secpar)$ defined below:
\begin{enumerate}
    \item $\A$ is given the security parameter $1^\secpar$ and sends a (possibly malformed) public parameter $\pp\in\ppspace$ to the challenger.
    \item The challenger randomly picks $m\sample \calM$, computes $\com\sample \Commit(\pp,m)$, and sends $\com$ to $\A$. 
    \item $\A$ returns $m^*$.
    The experiment outputs $1$ if $m=m^*$ and $0$ otherwise.
\end{enumerate}
We say that a commitment scheme is  unpredictable if for any non-uniform QPT adversary $\A$, we have
\begin{align*}
    \Pr[1\sample \Exp_{\A}^{\mathsf{unpre}}(1^\secpar)]=\negl(\secpar).
\end{align*}
\end{definition}
The following lemma is a folklore, and easy to prove.
\begin{lemma}\label{lem:hiding_to_unpredictable}
If a commitment scheme is computationally hiding and $|\calM|=2^{\omega(\secpar)}$, then the scheme is unpredictable. 
\end{lemma}

\paragraph{Instantiations.}
A computationally hiding and statistically binding commitment scheme exists under the existence of OWF \cite{JC:Naor91,SICOMP:HILL99}. 
A computationally  hiding and perfectly binding non-interactive commitment scheme exists under  the QLWE assumption \cite{TCC:GHKW17,ePrint:LomSch19}.

A statistically hiding and strong collapse-binding commitment scheme exists assuming the existence of collapsing hash functions (and thus under the QLWE assumption) \cite{EC:Unruh16,AC:Unruh16}. 
This can be seen by observing that the proof of  (original) collapse-binding property from collapsing hash functions in \cite{EC:Unruh16,AC:Unruh16} already implicitly proves the strong collapse-binding property.
For completeness, we give a proof in  Appendix \ref{sec:strong_collapse_binding}.
\fi

\subsection{Interactive Proof and Argument.}
We define interactive proofs and arguments similarly to \cite{STOC:BitShm20}. 
\paragraph{Notations.}
For an $\NP$ language $\lang$ and $x\in \lang$, $\rel_{\lang}(x)$ is the set that consists of all (classical) witnesses $w$ such that the verification machine for $L$ accepts $(x,w)$.

A (classical) interactive protocol is modeled as an interaction between interactive quantum machines $\pro$ referred to as a prover and $\ver$ referred to as a verifier that can be implemented by PPT algorithms.
We denote by $\execution{\pro(x_{\pro})}{\ver(x_{\ver})}(x)$ an execution of the protocol where $x$ is a common input, $x_\pro$ is $\pro$'s private input, and $x_\ver$ is $\ver$'s private input.
We denote by $\OUT_\ver\execution{\pro(x_{\pro})}{\ver(x_{\ver})}(x)$ the fianl output of $\ver$ in the execution. 
An honest verifier's output is $\top$ indicating acceptance or $\bot$ indicating rejection, and a quantum malicious verifier's output may be an arbitrary quantum state.  

\begin{definition}[Interactive Proof and Argument for $\NP$]
An interactive proof or argument for an $\NP$ language $\lang$ is an interactive protocol between a PPT prover $\pro$ and a PPT verifier $\ver$ that satisfies the following: 
\paragraph{Perfect Completeness.}
For any $x\in L$,  and $w\in R_L(x)$, we have 
\begin{align*}
    \Pr[\OUT_\ver\execution{\pro(w)}{\ver}(x)=\top]=1
\end{align*}
\paragraph{Statistical/Computational Soundness.}
We say that an interactive protocol is statistically (resp. computationally) sound if for any unbounded-time (resp. non-uniform QPT) cheating prover $\pro^*$, there exists a negligible function $\negl$ such that for any $\secpar \in \mathbb{N}$ and any $x\in \bit^\secpar\setminus \lang $, we have  
\begin{align*}
    \Pr[\OUT_\ver\execution{\pro^*}{\ver}(x)=\top]\leq \negl(\secpar).
\end{align*}
We call an interactive protocol with statistical (resp. computational) soundness an interactive proof (resp. argument).
\end{definition}

\ifnum\submission=0 
\subsubsection{Witness Indistinguishable Proof of Knowledge}\label{sec:wipok}
\begin{definition}[Witness Indistinguishable Proof of Knowledge]\label{def:WIPoK}
A witness indistinguishable proof of knowledge for an $\NP$ language $\lang$ is an interactive proof for $\lang$ that satisfies the following properties (in addition to perfect completeness and statistical soundness):
\paragraph{Witness Indistinguishability.}
For any non-uniform QPT malicious verifier $\ver^*$, we have 
\begin{align*}
\{\OUT_{\ver^*}\execution{\pro(w_0)}{\ver^*}(x)\}_{\secpar,x,w_0,w_1}\compind \{\OUT_{\ver^*}\execution{\pro(w_1)}{\ver^*}(x)\}_{\secpar,x,w_0,w_1}
\end{align*}
where $\secpar\in \mathbb{N}$, $x\in \lang \cap \bit^\secpar$,  and $w_0,w_1\in \rel_\lang(x)$.
\paragraph{(Non-Adaptive) Knowledge Extractability.}
There is an oracle-aided QPT algorithm $\kext$, a polynomial $\poly$, a negligible function $\negl$, and a constant $d\in \mathbb{N}$ such that for any quantum unbounded-time  malicious prover $\pro^*=\{\pro^*_\secpar,\rho_\secpar\}_{\secpar\in \mathbb{N}}$, $\secpar \in \mathbb{N}$, and $x\in \bit^\secpar$, we have   
\begin{align*}
\Pr[w\in \rel_L(x):w\sample \kext^{\pro^*_\secpar(\rho_\secpar)}(x)]
    \geq \frac{1}{\poly(\secpar)}\cdot \Pr[\OUT_{\ver}\execution{\pro^*_\secpar(\rho_\secpar)}{\ver}(x)=\top]^d - \negl(\secpar).
\end{align*}
\end{definition}
\paragraph{Instantiations.}
We can construct a constant round (actually $4$-round) witness indistinguishable proof of knowledge for $\NP$ only assuming the existence of post-quantum OWFs by some tweak of existing works \cite{EC:Unruh12,EC:Unruh16}. We briefly explain this below.

A constant round witness indistinguishable proof of knowledge that satisfies the above requirements was first constructed by Unruh \cite{EC:Unruh12} based on \emph{strict-binding commitments}, where a commitment perfectly binds not only a message but also randomness. 
Due to the usage of strict-binding commitment, an instantiation of this protocol requires one-to-one OWF, for which there is no post-quantum candidate under standard assumptions. 
Later, Unruh \cite{EC:Unruh16} proved that the protocol in \cite{EC:Unruh12} can be instantiated using collapse-binding commitments instead of strict-binding commitments if we relax the knowledge extractability requirement to computational one.
Since the statistical binding property trivially implies the collapse-binding property as noted in Remark \ref{rem:statistical_to_collapse_binding}, we can just use statistically binding commitments as  collapse-binding commitments in the construction of \cite{EC:Unruh16}.
Moreover, since the statistically binding property can be seen as a ``statistical version" of collapse-binding, we obtain statistical knowledge extractability.\footnote{If one is not convinced by this informal explanation, one can think of our claim as the existence of a constant round witness indistinguishable \emph{argument} of knowledge for $\NP$ under the existence of OWF. Actually, this computational version of knowledge extractability suffices for the purpose of this paper.} 
In summary, the construction in \cite{EC:Unruh16} instantiated with statistically binding (and computational hiding) commitments suffices for our purpose.

We also give another more concrete explanation.
The protocol in \cite{EC:Unruh12} is a modification of Blum's Graph Hamiltonicity protocol \cite{Blum86}. For proving the knowledge extractability, Unruh introduced a rewinding technique that enables the extractor to run a prover twice for different challenges.
In his extraction strategy, the extractor records \emph{both committed messages and randomness} when it runs the prover for the first time. 
For ensuring that this does not collapse the prover's state too much, he assumed the strict-binding property.
Here, we observe that the extractor actually need not record both committed message and randomness, and it only need to record the committed message. (Indeed, the security proof in \cite{EC:Unruh16} does so).
In this case, the statistical binding property instead of the strict-binding property suffices to ensure that the prover's state is not collapsed too much since the randomness register is not measured by the extractor.
\fi

\subsubsection{Delayed-Witness \texorpdfstring{$\Sigma$}{Sigma}-Protocol}\label{sec:delayed_sigma}
We introduce a special type of $\Sigma$-protocol which we call  \emph{delayed-witness $\Sigma$-protocol} where the first message can be generated without witness.
\begin{definition}[Delayed-Witness $\Sigma$-protocol]\label{def:delayed_witness_sigma_protocol}
A (post-quantum) delayed-witness $\Sigma$-protocol for an $\NP$ language $L$ is a 3-round interactive proof for $\NP$  with the following syntax.
\end{definition}
\noindent\textbf{Common Input:} An instance $x\in \lang \cap \bit^{\secpar}$ for security parameter $\secpar \in \mathbb{N}$.\\
\textbf{$\pro$'s Private Input:} A classical witness $w\in \rel_\lang(x)$ for $x$. 
\begin{enumerate}
\item $\pro$ generates a ``commitment" $a$ and a state $\st$. 
For this part, $\pro$ only uses the statement $x$ and does not use any witness $w$.
We denote this procedure by $(a,\st)\sample \Sigma.\pro_1(x)$.
Then it sends $a$ to the verifier, and keeps $\st$ as its internal state.
\item $\ver$ chooses a``challenge" $e\sample \bit^\secpar$ and sends $e$ to $\pro$.
\item $\pro$ generates a ``response" $z$ from $\st$, witness $w$, and $e$.
We denote this procedure by $z\sample \Sigma.\pro_3(\st,w,e)$.
Then it sends $z$ to $\ver$.
\item $\ver$ verifies the transcript $(a,e,z)$ and outputs $\top$ indicating acceptance or $\bot$ indicating rejection.
We denote this procedure by $\top/\bot \sample \Sigma.\ver(x,a,e,z)$.
\end{enumerate}
We require a delayed-witness $\Sigma$-protocol to satisfy the following property in addition to perfect completeness and statistical soundness.\footnote{We do not require \emph{special soundness}, which is often a default requirement of $\Sigma$-protocol.}
\paragraph{Special Honest-Verifier Zero-Knowledge.}
There exists a PPT simulator $\siml_{\Sigma}$ such that we have
\begin{align*}
&\{(a,z):(a,\st)\sample \Sigma.\pro_1(x),z\sample \Sigma.\pro_3(\st,w,e)\}_{\secpar,x,w,e}
\compind
\{(a,z):(a,z)\sample \siml_{\Sigma}(x,e)\}_{\secpar,x,w,e}
\end{align*}
where $x\in \lang \cap \bit^\secpar$, $w\in \rel_\lang(x)$, and $e\in \bit^\secpar$.

\paragraph{Instantiations.}
An example of a dealyed-witness $\Sigma$-protocol is a parallel repetition version of Blum's Graph Hamiltonicity protocol \cite{Blum86}.
In the ptorocol, we need a computationally hiding  and perfectly binding non-interactive commitment scheme, which exists under the QLWE assumption as noted in Sec. \ref{sec:commitment}. 
In summary, a delayed-input $\Sigma$-protocol for all $\NP$ languages exists under the QLWE assumption.


\subsubsection{Quantum \texorpdfstring{$\epsilon$}{epsilon}-Zero-Knowledge Proof and Argument}
Here, we define quantum black-box $\epsilon$-zero-knowledge proofs and arguments.
The difference from the definition of quantum zero-knowledge in \cite{STOC:BitShm20} are:
\begin{enumerate}
    \item(\textbf{$\epsilon$-Zero-Knowledge}) We allow the simulator to depend on a noticeable ``accuracy parameter" $\epsilon$, and allows its running time to polynomially depend on $\epsilon^{-1}$, and 
    \item(\textbf{Black-Box Simulation})
    the simulator is only given black-box access to a malicious verifier.
\end{enumerate}
\begin{definition}[Post-Quantum Black-Box $\epsilon$-Zero-Knowledge Proof and Argument]\label{def:post_quantum_ZK}
A post-quantum black-box  $\epsilon$-zero-knowledge proof (resp. argument) for an $\NP$ language $\lang$ is an interactive proof (resp. argument) for $\lang$ that satisfies the following property in addition to perfect completeness and statistical (resp. computational) soundness:
\paragraph{Quantum Black-Box $\epsilon$-Zero-Knowledge.}
There exists an oracle-aided QPT simulator $\siml$ such that for any non-uniform QPT malicious verifier $\ver^*=\{\ver^*_\secpar,\rho_\secpar\}_{\secpar\in \mathbb{N}}$ and any noticeable function $\epsilon(\secpar)$, we have
\[
\{\OUT_{\ver^*_\secpar}\execution{\pro(w)}{\ver^*_\secpar(\rho_\secpar)}(x)\}_{\secpar,x,w}
\compind_{\epsilon}
\{\OUT_{\ver^*_\secpar}(\siml^{\ver^*_\secpar(\rho_\secpar)}(x,1^{\epsilon^{-1}}))\}_{\secpar,x,w}
\]
where $\secpar \in \mathbb{N}$, $x\in \lang\cap \bit^\secpar$, $w\in \rel_\lang(\secpar)$, and  $\OUT_{\ver^*_\secpar}(\siml^{\ver^*_\secpar(\rho_\secpar)}(x))$ is the state in the output register of $\ver^*_\secpar$ after the simulated execution of $\ver^*_\secpar$ by $\siml$.
\end{definition}
\begin{remark}
In the above definition of quantum black-box $\epsilon$-zero-knowledge, we do not consider an entanglement between auxiliary input of a malicious verifier and distinguisher unlike the original definition of quantum zero-knowledge by Watrous  \cite{SIAM:Watrous09}. 
However, in  
\ifnum\cameraready=1
the full version
\else
Appendix \ref{app:equivalence_ZK}, 
\fi 
we show that the above definition implies indistinguishability against a distinguisher that may get an entangled state to verifier's auxiliary input by taking advantage of black-box simulation.
\end{remark}

\ifnum\submission=1 
\subsubsection{Witness Indistinguishable Proof of Knowledge}
The definition of witness   indistinguishable proof of knowledge is given in  \ifnum\cameraready=1
the full version
\else
Appendix \ref{sec:wipok}.
\fi
\fi

\subsection{Quantum Rewinding Lemma}
Watrous \cite{SIAM:Watrous09} proved a lemma that enables us to amplify the success probability of a quantum algorithm under certain conditions.
The following form of the lemma is based on that in  \cite[Lemma 2.1]{STOC:BitShm20}.
\begin{lemma}[\cite{SIAM:Watrous09,STOC:BitShm20}]\label{lem:quantum_rewinding}
There is an oracle-aided quantum algorithm $\sfR$ that gets as input the following:
\begin{itemize}
\item A quantum circuit $\sfQ$ that takes $n$-input qubits in register $\reginp$ and outputs a classical bit $b$ (in a register outside $\reginp$)  and an $m$ output qubits.  
\item An $n$-qubit state $\rho$ in register $\reginp$.
\item A number $T\in \mathbb{N}$ in unary.
\end{itemize}

$\sfR(1^T,\sfQ,\rho)$ executes in time $T\cdot|\sfQ|$  and outputs a distribution over $m$-qubit states  $D_{\rho}\defeq \sfR(1^T,\sfQ,\rho)$  with the following guarantees.

For an $n$-qubit state $\rho$, denote by $\sfQ_{\rho}$ the conditional distribution of the output distribution $\sfQ(\rho)$,
conditioned on $b = 0$, and denote by $p(\rho)$ the probability that $b = 0$. If there exist $p_0, q \in (0,1)$, $\gamma \in (0,\frac{1}{2})$
such that:
\begin{itemize}
    \item  Amplification executes for enough time: $T\geq \frac{\log (1/\gamma)}{4p_0(1-p_0)}$,
    \item  There is some minimal probability that $b = 0$: For every $n$-qubit state $\rho$, $p_0\leq p(\rho)$,
    \item  $p(\rho)$ is input-independent, up to $\gamma$ distance: For every $n$-qubit state $\rho$, $|p(\rho)-q|<\gamma$, and
    \item  $q$ is closer to $\frac{1}{2}$: $p_0(1-p_0)\leq q(1-q)$,
\end{itemize}
then for every $n$-qubit state $\rho$,
\begin{align*}
    \TD(\sfQ_{\rho},D_{\rho})\leq 4\sqrt{\gamma}\frac{\log(1/\gamma)}{p_0(1-p_0)}.
\end{align*}
Moreover, $\sfR(1^T,\sfQ,\rho)$ works in the following manner: 
It  uses $\sfQ$ for only implementing oracles that perform the unitary part of $\sfQ$ and its inverse, acts on $\reginp$ only through these oracles, and the output of $\sfR$ is the state in the output register of $\sfQ$ after the simulated execution.
We note that $\sfR$ may directly act on $\sfQ$'s internal registers other than $\reginp$. 
\end{lemma}
\begin{remark}
The final claim of the lemma (``Moreover...") is not explicitly stated in previous works.
In the description of $\mathsf{R}$ in \cite{SIAM:Watrous09}, the first qubit of $\reginp$ is designated to output $b$,  and thus the above requirement is not satisfied.
However, this can be easily avoided by just letting $\sfQ$ output $b$ in a register outside $\reginp$ as required above.
Then one can see that $\sfR$ acts on the input register only through $\sfQ$ as seen from the description of $\sfR$ in \cite{SIAM:Watrous09} (with the above modification in mind).
Looking ahead, this is needed to show our $\epsilon$-zero-knowledge simulators are black-box.
\end{remark}

\section{Technical Lemmas}
In this section, we introduce three lemmas that are used in the proof of the extraction lemma (Lemma \ref{lem:extraction}) in Sec. \ref{sec:extraction}.

\revise{
\begin{lemma}\label{lem:state-close} \takashi{I introduced this lemma for simplifying the proof of Claim \ref{claim:hybtwo_to_real}. 
A similar lemma was used in CCLY.}
Let $\ket{\phi_b}=\ket{\phi_{b,0}}+\ket{\phi_{b,1}}$ be a normalized quantum state and $\Pi$ be a projector over a Hilbert space $\hil$ such that  
$\bra{\phi_{b,0}}\Pi \ket{\phi_{b',1}}=0$  for $b,b'\in \bit$, 
$\left\|\Pi \ket{\phi_{b,0}}\right\|^2\le \gamma$  
 for $b\in \bit$, and 
 $\|\ket{\phi_{1,1}}-\ket{\phi_{0,1}}\|
 \le \delta$ for some real numbers $\gamma,\delta$. 
Let $F$ be a quantum algorithm  that takes a state in $\hil$ as input, applies the projective measurement $(\Pi,I-\Pi)$, and outputs the resulting state if the measurement outcome is $0$ i.e., the state is projected onto $\Pi$, and otherwise outputs $\bot$.

Then it holds that 
\begin{align*}
    \TD(F(\ket{\phi_0}),F(\ket{\phi_1})) 
    \leq 
    \sqrt{4\gamma+2\delta}.
\end{align*}
\end{lemma}
\begin{proof}
If $\sqrt{4\gamma+2\delta}>1$, then the desired inequality trivially holds. Thus,  we assume $\sqrt{4\gamma+2\delta}\le 1$ in the rest of the proof. 
We consider an additional one-qubit register and define 
\begin{align*}
\ket{\psi_b}\defeq \sqrt{1-p_b}\ket{0}\ket{0^m} + \ket{1}\Pi\ket{\phi_b}
\end{align*}
for $b\in \bit$ where 
$m$ is the number of qubits in the register for $\ket{\phi_b}$ and 
\begin{align*}
p_b \defeq \|\Pi\ket{\phi_b}\|^2.
\end{align*}
Without loss of generality, we assume $p_0\geq p_1$. 
It suffices to prove 
\begin{align}
\label{eq:conclusion_trace_distance_pure}
   \TD(\ket{\psi_0}\bra{\psi_0},\ket{\psi_1}\bra{\psi_1})\leq  \sqrt{4\gamma+2\delta}
\end{align}
because a distinguisher that distinguishes $F(\ket{\phi_0})$ and $F(\ket{\phi_1})$ can be easily converted into a distinguisher that distinguishes 
$\ket{\psi_0}$ and $\ket{\psi_1}$ with the same advantage. 

We have 
\begin{align*}
    p_0&=\|\Pi\ket{\phi_0}\|^2\\
    &= \|\Pi\ket{\phi_{0,0}}\|^2+\|\Pi\ket{\phi_{0,1}}\|^2 \\
    &\leq \|\Pi\ket{\phi_{0,1}}\|^2+ \gamma 
\end{align*} 
where we used 
the assumption that $\bra{\phi_{0,0}}\Pi \ket{\phi_{0,1}}=0$  in the second equality and
the assumption that $\left\|\Pi\ket{\phi_{0,0}}\right\|^2\le \gamma$
in the final inequality. 

Thus, we have 
\begin{align}
\label{eq:useful_properties_new}
\|\Pi\ket{\phi_{0,1}}\|^2\geq p_0-\gamma.
\end{align}

We give a lower bound for $|\langle\psi_0|\psi_1\rangle|$.
By the definition of $\ket{\psi_b}$, 
\begin{align*}
|\langle\psi_0|\psi_1\rangle|
&=|\sqrt{(1-p_0)(1-p_1)}+ \bra{\phi_0}\Pi\ket{\phi_1}| \notag\\
&=|\sqrt{(1-p_0)(1-p_1)}+ \bra{\phi_{0,0}}\Pi\ket{\phi_{1,0}}+\bra{\phi_{0,1}}\Pi\ket{\phi_{1,1}}| \notag\\ 
&=|\sqrt{(1-p_0)(1-p_1)}+ \bra{\phi_{0,0}}\Pi\ket{\phi_{1,0}}+\bra{\phi_{0,1}}\Pi\ket{\phi_{0,1}}+\bra{\phi_{0,1}}\Pi\left(\ket{\phi_{1,1}}-\ket{\phi_{0,1}}\right)| \notag\\
&\geq (1-p_0)+\|\Pi\ket{\phi_{0,1}}\|^2- \left\|\Pi\ket{\phi_{0,0}}\right\|\cdot \|\Pi\ket{\phi_{1,0}}\|-\|\ket{\phi_{1,1}}-\ket{\phi_{0,1}}\|\\
&\geq 
(1-p_0)+(p_0-\gamma)-\gamma-\delta\\
&=1-(2\gamma+\delta)  
\end{align*}
where we used 
the assumption that $\bra{\phi_{b,0}}\Pi \ket{\phi_{b',1}}=0$  for $b,b'\in \bit$ in the first equality, 
the assumptions that $p_0\geq p_1$ and $\ket{\phi_{0,1}}=\ket{\phi_{1,1}}$ in the first inequality, 
and Eq. 
\ref{eq:useful_properties_new} and the assumptions that $\left\|\Pi\ket{\phi_{b,0}}\right\|^2\le \gamma$ for $b\in \bit$
and 
$\|\ket{\phi_{1,1}}-\ket{\phi_{0,1}}\|\le \delta$
 in the second inequality.
We note that $1-(2\gamma+\delta)>0$ since we assume $\sqrt{4\gamma+2\delta}\le 1$. 

Then, we have 
\begin{align*}
 \TD(\ket{\psi_0}\bra{\psi_0},\ket{\psi_1}\bra{\psi_1})
 &=\sqrt{1-|\langle\psi_0|\psi_1\rangle|^2}\\
 &\leq  \sqrt{4\gamma+2\delta}
\end{align*}
This completes the proof of Lemma \ref{lem:state-close}. 

\end{proof}
}

\revise{
The second lemma is the following variant of the gentle measurement lemma. 
\begin{lemma}\label{lem:gentle_measurement} \takashi{The previous version also used this lemma in the case of $\nu=\negl$ by just saying "the gentle measurement lemma".}
    Let $\ket{\psi}_\regX$ be a (not necessarily normalized) state over register $\regX$ and $U$ be a unitary over registers $(\regX,\regY,\regZ)$.  
    Suppose that a measurement of register $\regZ$ of $U\ket{\psi}_\regX\ket{0}_{\regY,\regZ}$ results in a deterministic value except for probability $\nu$, i.e., 
    there is $z^*$ such that 
    \begin{align*}
        \|(I-\ket{z^*}\bra{z^*})_{\regZ} U\ket{\psi}_\regX\ket{0}_{\regY,\regZ}\|^2\le \nu. 
    \end{align*}
    If we let $R:= (\ket{0}\bra{0})_{\regY,\regZ}U^\dagger (\ket{z^*}\bra{z^*})_{\regZ} U$, 
then we have 
\begin{align*}
    \|
    \ket{\psi}_\regX\ket{0}_{\regY,\regZ}-
   R \ket{\psi}_\regX\ket{0}_{\regY,\regZ}\|\le \sqrt{\nu}.
\end{align*}
\end{lemma}
\begin{proof}
Let $\Pi_{z^*}:=(\ket{z^*}\bra{z^*})_{\regZ}$.
We have 
\begin{align*}
\ket{\psi}_\regX\ket{0}_{\regY,\regZ}&=(\ket{0}\bra{0})_{\regY,\regZ}U^\dagger  U\ket{\psi}_\regX\ket{0}_{\regY,\regZ}\\
&= R\ket{\psi}_\regX\ket{0}_{\regY,\regZ}+(\ket{0}\bra{0})_{\regY,\regZ}U^\dagger (I-\Pi_{z^*}) U\ket{\psi}_\regX\ket{0}_{\regY,\regZ}.
\end{align*}
Thus, we have 
\begin{align*}
    \|
    \ket{\psi}_\regX\ket{0}_{\regY,\regZ}-
   R \ket{\psi}_\regX\ket{0}_{\regY,\regZ}\|
   &=\|(\ket{0}\bra{0})_{\regY,\regZ}U^\dagger (I-\Pi_{z^*}) U\ket{\psi}_\regX\ket{0}_{\regY,\regZ}\|\\ 
   &\le \|(I-\Pi_{z^*}) U\ket{\psi}_\regX\ket{0}_{\regY,\regZ}\|\le \sqrt{\nu}.
\end{align*}
\end{proof}
}
 
 The third lemma is about amplifying the success probability of a projection.
Very roughly speaking, the lemma states that for any projection $\Pi$ and a ``threshold"  $0<t<1$, we can decompose the Hilbert space into two subspaces $S_{<t}$ and $S_{\geq t}$ so that 
\begin{enumerate}
    \item $\Pi$ succeeds with probability $<t$ (resp. $\geq t$) in $S_{<t}$ (resp. $S_{\geq t}$).
    \item There is an efficient procedure $\Amp$ that runs in time $T=O( t^{-1})$ that maps any state $\ket{\psi}\in S_{\geq t}$ onto the span of $\Pi$ with probability almost $1$. We note that this does not necessarily map $\ket{\psi}$ to $\Pi\ket{\psi}$.
    \item Each subspace is invariant under $\Pi$ and $\Amp$.  
\end{enumerate}
The formal statement of our lemma is given below:
\begin{lemma}\label{lem:amplification}
Let $\Pi$ be a projection over a Hilbert space $\hil_\regX \ot \hil_\regY$. 
For any  \revise{noticeable function $t=t(\secpar)$,} 
there exists an orthogonal decomposition $(S_{<t}, S_{\geq t})$ of $\hil_\regX \ot \hil_\regY$ that satisfies the following:
\begin{enumerate}
\item(\textbf{$S_{<t}$ and $S_{\geq t}$ are invariant under $\Pi$ and $(\ket{0}\bra{0})_{\regY}$.}) \label{item:amplification_invariance_projection}
For any $\ket{\psi}_{\regX,\regY}\in S_{<t}$, we have 
\begin{align*}
\Pi \ket{\psi}_{\regX,\regY}\in S_{<t},~~~~~ (I_{\regX}\otimes(\ket{0}\bra{0})_{\regY})\ket{\psi}_{\regX,\regY}\in S_{<t}.
\end{align*}
Similarly, 
 for any $\ket{\psi}_{\regX,\regY}\in S_{\geq t}$, we have 
\begin{align*}
\Pi \ket{\psi}_{\regX,\regY}\in S_{\geq t},~~~~~ (I_{\regX}\otimes(\ket{0}\bra{0})_{\regY})\ket{\psi}_{\regX,\regY}\in S_{\geq t}.
\end{align*}
\item(\textbf{$\Pi$ succeeds with probability $<t$ and $\geq t$ in $S_{<t}$ and $S_{\geq t}$.})  \label{item:amplification_success_probability}
For any quantum state $\ket{\phi}_{\regX}\in \hil_{\regX}$ s.t.  $\ket{\phi}_{\regX}\ket{0}_{\regY}\in S_{<t}$ we have 
\begin{align*}
\|\Pi\ket{\phi}_{\regX}\ket{0}_{\regY}\|^2< t. 
\end{align*}
Similarly, for any quantum state $\ket{\phi}_{\regX}\in \hil_{\regX}$ s.t.  $\ket{\phi}_{\regX}\ket{0}_{\regY}\in S_{\geq t}$ we have 
\begin{align*}
\|\Pi\ket{\phi}_{\regX}\ket{0}_{\regY}\|^2\geq  t. 
\end{align*} 
 \item(\textbf{Unitary for amplification.})
For any $T\in \mathbb{N}$, there exists a unitary  $U_{\amp,T}$ over $\hil_\regX \ot \hil_\regY \ot \hil_\regB \ot \hil_\reganc$ where $\regB$ is a register to store a qubit and $\reganc$ is a register to store ancillary qubits with the following properties:
 \label{item:amplification}
 \begin{enumerate}
 \item(\textbf{Mapped onto $\Pi$ when $\regB$ contains $1$.})
For any quantum state $\ket{\psi}_{\regX,\regY}\in \hil_{\regX}\otimes \hil_{\regY}$, we can write
 \[
 \ket{1}\bra{1}_{\regB} U_{\amp,T}\ket{\psi}_{\regX,\regY}\ket{0}_{\regB,\reganc}=\sum_{anc}\ket{\psi'_{anc}}_{\regX,\regY}\ket{1}_{\regB}\ket{anc}_{\reganc}
 \]
 by using sub-normalized states $\ket{\psi'_{anc}}_{\regX,\regY}$ that are in the span of $\Pi$.
 \label{item:amplification_map_to_pi}
\item(\textbf{Amplification of success probability in $S_{\geq t}$.})  \label{item:amplification_amplification}
\revise{For any noticeable function $\nu=\nu(\secpar)$, there is $T=\poly(\secpar)$ such that}
for any quantum state $\ket{\phi}_{\regX}\in \hil_{\regX}$ s.t.  $\ket{\phi}_{\regX}\ket{0}_{\regY}\in S_{\geq t}$, we have\footnote{\revise{In the previous versions of this paper, we claimed that the lower bound is $1-(1-2t+2t^2)^{T-1}(1-t)$. 
However, this was based on a false claim that $(1-2t+2t^2)^{T-1}(1-t)$ is decreasing in $t\in[0,1]$ for any fixed $T\ge 1$.} 
} 
 \[
 \|\ket{1}\bra{1}_\regB U_{\amp,T}\ket{\phi}_{\regX}\ket{0}_{\regY}\ket{0}_{\regB,\reganc}\|^2
 \geq \revise{1-\nu}.
 \]
 \item(\textbf{$S_{<t}$ and $S_{\geq t}$ are invariant under $U_{\amp,T}$}). 
 For any  quantum state $\ket{\psi_{<t}}_{\regX,\regY}\in  S_{<t}$ and any $b,anc$,
we can write 
  \[
U_{\amp,T}\ket{\psi_{< t}}_{\regX,\regY}\ket{b,anc}_{\regB,\reganc}=\sum_{b',anc'}\ket{\psi'_{<t,b',anc'}}_{\regX,\regY}\ket{b',anc'}_{\regB,\reganc}
  \]
  by using sub-normalized states $\ket{\psi'_{<t, b',anc'}}_{\regX,\regY} \in S_{<t}$.
  
  Similarly, 
 for any  quantum state $\ket{\psi_{\geq t}}_{\regX,\regY}\in  S_{\geq t}$ and any $b,anc$,
we can write 
  \[
U_{\amp,T}\ket{\psi_{\geq t}}_{\regX,\regY}\ket{b, anc}_{\regB,\reganc}=\sum_{b', anc'}\ket{\psi'_{\geq t, b', anc'}}_{\regX,\regY}\ket{b', anc'}_{\regB,\reganc}
  \]
  by using sub-normalized states $\ket{\psi'_{\geq t, b',anc'}}_{\regX,\regY} \in S_{\geq t}$. 
  \label{item:amplification_invariance}
 \end{enumerate}
 \item(\textbf{Efficient Implementation of $U_{\amp,T}$}.) 
 There exists a QPT algorithm $\Amp$ (whose description is independent of $\Pi$) that takes as input  $1^T$, a description of quantum circuit that perform a measurement $(\Pi, I_{\regX,\regY}-\Pi)$, and a state $\ket{\psi}_{\regX,\regY,\regB,\reganc}$, and outputs $U_{\amp,T}\ket{\psi}_{\regX,\regY,\regB,\reganc}$.
 Moreover, $\Amp$ uses the measurement circuit for only implementing an oracle that apply unitary to write a measurement result in a designated register in $\reganc$, and it acts on $\regX$ only through the oracle access.
 \label{item:amplification_efficiency}
\end{enumerate}
\end{lemma}

Since the above lemma can be proven by a similar usage of Jordan's lemma to existing works \cite{NWZ09,TCC:ChiChuYam20}, we give the proof in Appendix \ref{sec:proof_amplification}.

\section{Extraction Lemma}\label{sec:extraction}
In this section, we prove our main technical lemma, which we call the \emph{extraction lemma}.
Before giving a formal statement, we give an intuitive explanation.
Suppose that we have a two-stage quantum algorithm $\A=(\A_\commit,\A_\open)$ that works as follows.
$\A_\commit$ is given $\pp$ of a commitment scheme and generates a commitment $\com$, and passes a quantum state $\rho_\st$ in its internal register to $\A_\open$. $\A_\open$ is given the internal state $\rho_\st$, and outputs a message-randomness pair $(m,r)$ (which is not necessarily a valid opening to $\com$) along with a classical output $\out$, and let $\rho'_\st$ be its internal state after the execution.
We call a successive execution of $\A_\commit$ and $\A_\open$ a real experiment.
On the other hand, we consider an \emph{extraction experiment} where an ``extractor" $\ext$ runs on input $\rho_\st$ in between   $\A_\commit$ and $\A_\open$ to ``extract" a committed message $m_\ext$ while generating a simulated $\A$'s internal state $\rho_\ext$. Then we run $\A_\open$ with the internal state $\rho_\ext$ instead of $\rho_\st$ to complete the extraction experiment.
Roughly, the extraction lemma claims that 
if the commitment scheme is strong collapse-binding (resp. statistically binding),
then there exists an extractor $\ext$ such that we have $m=m_\ext$ with high probability and distributions of $(m,r,\out,\rho'_\st)$ in real and extraction experiments are computationally (resp. statistically) indistinguishable \emph{conditioned on that $(m,r)$ is a valid opening to $\com$}. 

The formal statement is given below.

\begin{definition}[Extraction Experiments] \label{def:extraction_experiments}
Let $\Com=(\Setup,\Commit)$ be a commitment scheme with message space $\calM$,  
randomness space $\calR$,  
commitment space $\COM$, 
and a public parameter space $\ppspace$. 
Let 
$\A=\{\A_{\commit,\secpar},\A_{\open,\secpar},\rho_{\secpar}\}_{\secpar\in \mathbb{N}}$ be a sequence of two-stage non-uniform QPT algorithms with the following syntax: 
\begin{description}
\item[$\A_{\commit,\secpar}(\pp;\rho_\secpar)\rightarrow (\com,\rho_\st)$:]
It takes as input  $\pp\in\ppspace$  
and an advice $\rho_\secpar$,  and outputs $\com\in \COM$ and a quantum state $\rho_\st$ in register $\regst$.
\item[$\A_{\open,\secpar}(\rho_\st)\rightarrow (m,r,\out,\rho'_\st)$:]
It takes as input  a quantum state $\rho_\st$ in register $\regst$,  and outputs $m\in \calM$, $r\in\calR$, a classical string $\out$, and a quantum state $\rho'_\st$ in register $\regst$.
\end{description}
Let $\ext$ be a QPT algorithm and $\delta$ be a function in $\secpar$.
Then we define following experiments:\\

\ifnum\submission=0
\begin{tabular}{l|l}
    \begin{minipage}[t]{0.40\textwidth}
\underline{$\Exp_\real[\Com,\A](\secpar)$}\\
$\pp\sample \Setup(1^\secpar)$,\\
$(\com, \rho_\st)\sample \A_{\commit,\secpar}(\pp;\rho_\secpar)$,\\
~\vspace{1.6mm}\\
$(m,r,\out,\rho'_\st)\sample \A_{\open,\secpar}(\rho_\st)$,\\
\emph{If} $\Commit(\pp,m;r)\neq \com$,\\
$~~~~$\emph{Output} $\bot$\\
\emph{Else Output} $(\pp,\com,m,r,\out,\rho'_\st)$.
    \end{minipage} 
        &~
    \begin{minipage}[t]{0.55\textwidth}
\underline{$\Exp_{\mathsf{ext}}[\Com,\A,\ext](\secpar,\delta)$}\\
$\pp\sample \Setup(1^\secpar)$,\\
$(\com, \rho_\st)\sample \A_{\commit,\secpar}(\pp;\rho_\secpar)$,\\
$(m_\ext , \rho_\ext )\sample \ext(1^\secpar,1^{\delta^{-1}},\pp,\com,\A_{\open,\secpar},\rho_\st)$,\\
$(m,r,\out,\rho'_\st)\sample \A_{\open,\secpar}(\rho_\ext),$\\
\emph{If} $\Commit(\pp,m;r)\neq \com \lor m \neq m_\ext$,\\
$~~~~$\emph{Output} $\bot$\\
\emph{Else Output} $(\pp,\com,m,r,\out,\rho'_\st)$.
    \end{minipage}
    \end{tabular}\\
\else 
{\footnotesize
\begin{tabular}{l|l}
    \begin{minipage}[t]{0.45\textwidth}
\underline{$\Exp_\real[\Com,\A](\secpar)$}\\
$\pp\sample \Setup(1^\secpar)$,\\
$(\com, \rho_\st)\sample \A_{\commit,\secpar}(\pp;\rho_\secpar)$,\\
~\vspace{1.6mm}\\
$(m,r,\out,\rho'_\st)\sample \A_{\open,\secpar}(\rho_\st)$,\\
\emph{If} $\Commit(\pp,m;r)\neq \com$,\\
$~~~~$\emph{Output} $\bot$\\
\emph{Else Output} $(\pp,\com,m,r,\out,\rho'_\st)$.
    \end{minipage} 
        &~
    \begin{minipage}[t]{0.55\textwidth}
\underline{$\Exp_{\mathsf{ext}}[\Com,\A,\ext](\secpar,\delta)$}\\
$\pp\sample \Setup(1^\secpar)$,\\
$(\com, \rho_\st)\sample \A_{\commit,\secpar}(\pp;\rho_\secpar)$,\\
$(m_\ext , \rho_\ext )\sample \ext(1^\secpar,1^{\delta^{-1}},\pp,\com,\A_{\open,\secpar},\rho_\st)$,\\
$(m,r,\out,\rho'_\st)\sample \A_{\open,\secpar}(\rho_\ext),$\\
\emph{If} $\Commit(\pp,m;r)\neq \com \lor m \neq m_\ext$,\\
$~~~~$\emph{Output} $\bot$\\
\emph{Else Output} $(\pp,\com,m,r,\out,\rho'_\st)$.
    \end{minipage}
    \end{tabular}\\
}
\fi
\end{definition}

\begin{lemma}[Extraction Lemma]\label{lem:extraction}
For any strong collapse-binding commitment scheme $\Com=(\Setup,\Commit)$, there exists a QPT algorithm $\ext$ such that 
for any noticeable
function $\delta(\secpar)$ and $\A=\{\A_{\commit,\secpar},\A_{\open,\secpar},\rho_{\secpar}\}_{\secpar\in \mathbb{N}}$ as in Definition \ref{def:extraction_experiments}, we have 

\ifnum\submission=1
 \begin{align*}
\{\Exp_\real[\Com,\A](\secpar)\}_{\secpar\in \mathbb{N}}
\compind_{\delta} 
 \{\Exp_{\mathsf{ext}}[\Com,\A,\ext](\secpar,\delta)\}_{\secpar\in \mathbb{N}}.
 \end{align*}
\else
 \begin{align}
\{\Exp_\real[\Com,\A](\secpar)\}_{\secpar\in \mathbb{N}}
\compind_{\delta} 
 \{\Exp_{\mathsf{ext}}[\Com,\A,\ext](\secpar,\delta)\}_{\secpar\in \mathbb{N}}.
 \label{eq:extraction_conclusion}
 \end{align}
\fi

If  $\Com$ is statistically binding instead of strong collapse-binding,  we have
\ifnum\submission=1
 \begin{align*}
\{\Exp_\real[\Com,\A](\secpar)\}_{\secpar\in \mathbb{N}}
\statind_{\delta} 
 \{\Exp_{\mathsf{ext}}[\Com,\A,\ext](\secpar,\delta)\}_{\secpar\in \mathbb{N}}. 
 \end{align*}
\else
 \begin{align}
\{\Exp_\real[\Com,\A](\secpar)\}_{\secpar\in \mathbb{N}}
\statind_{\delta} 
 \{\Exp_{\mathsf{ext}}[\Com,\A,\ext](\secpar,\delta)\}_{\secpar\in \mathbb{N}}.
  \label{eq:extraction_conclusion_statistical}
 \end{align}
\fi
Moreover, $\ext(1^\secpar,1^{\delta^{-1}},\pp,\com,\A_{\open,\secpar},\rho_\st)$ works in the following manner: It  uses $\A_{\open,\secpar}$ for only implementing oracles that perform unitary part of $\A_{\open,\secpar}$ and its inverse, and acts on $\regst$ only through black-box access to the oracles.
The second output $\rho_\ext$ of $\ext$ is the state in $\regst$ after the execution.  
We note that $\ext$ may directly act on   internal registers of $\A_{\open,\secpar}$ other than $\regst$.
\end{lemma}

\ifnum\submission=1 
\revise{The above lemma abstracts our technical core, which is extraction of the verifier's committed challenge without collapsing verifier's internal state too much. 
(One can think of $\A$ in the above lemma as the verifier and $\rho_{\st}$ and $\rho'_{\st}$ as verifier's internal states before and after opening the commitment, respectively, in our constant round $\epsilon$-zero-knowledge proofs/arguments.)
Since the intuition of the proof is already explained in Sec. \ref{sec:overview}, we defer the proof to 
\ifnum\cameraready=1
the full version. 
\else
Appendix \ref{sec:appendix_proof_extraction}.}
\fi
\else

\subsection{Proof of Extraction Lemma (Lemma \ref{lem:extraction})}\label{sec:proof_extraction}
\begin{proof}   (of Lemma \ref{lem:extraction})
We first prove the lemma for the case of strong collapse-binding commitments.
We explain how to modify the proof to prove the lemma for statistically binding commitments at the end of the proof.


In the proof, we need to consider sequences of many objects (e.g., unitary, Hilbert space, projection, etc.) indexed by $\secpar$.
For the sake of simplicity, we will often ignore the indexing by $\secpar$.

Let $U_{\open}$ be the unitary that represents $\A_{\open,\secpar}$. 
More precisely, we define $U_{\open}$ over a Hilbert space $\hil_\regA\defeq  \hil_\regst \ot  \hil_\regW \ot \hil_\regM \ot \hil_\regR \ot \hil_\regout$ so that $\A_{\open,\secpar}$ can be described as follows:
\begin{description}
\item[$\A_{\open,\secpar}(\rho_\st)$:] 
It stores a quantum state $\rho_\st$ in the register $\regst$ and initializes registers $\regW$, $\regM$, $\regR$, and $\regout$ to be $\ket{0}_{\regW,\regM,\regR,\regout}$.
Then it applies a unitary $U_{\open}$, measures registers $\regM$, $\regR$, and $\regout$ in the standard basis to obtain $m$, $r$, and $\out$, and outputs $m$, $r$, $\out$, and a quantum state $\rho'_\st$ in register $\regst$ tracing out register $\regW$. 
\end{description}

For any $\pp\in\ppspace$ and $\com\in\COM$,  we define a projection $\Pi^{\pp,\com}$ over $\hil_\regA$ as
\begin{align}\label{eq:definition_pi}
 \Pi^{\pp,\com} \defeq U_{\open}^\dagger \Pi^{\pp,\com}_{\test}  U_{\open}
\end{align}
where 
\begin{align*}
 \Pi^{\pp,\com}_{\test} \defeq  \left(\sum_{(m,r):\Commit(\pp,m;r)=\com}\left(\ket{m,r}\bra{m,r}\right)_{\regM,\regR}\right).
\end{align*}
\revise{We apply Lemma \ref{lem:amplification} for $\hil_\regX:=\hil_\regst$, $\hil_\regY:= \hil_\regW \ot \hil_\regM \ot \hil_\regR \ot \hil_\regout$, $\Pi:=\Pi^{\pp,\com}$, $t:=\delta^2/6$, and $T=\poly(\secpar)$ is chosen in such a way that Item \ref{item:amplification_amplification} of Lemma \ref{lem:amplification} holds for $\nu=t^2$.}
Then we have a decomposition $(S_{<t}^{\pp,\com}, S_{\geq t}^{\pp,\com})$ of $\hil_\regA$ and a unitary $U_{\amp,T}^{\pp,\com}$ over $\hil_\regX\ot \hil_\regY \ot \hil_\regB \ot \hil_\reganc$ that satisfies the requirements in Lemma \ref{lem:amplification} where we write $\pp,\com$ in superscript  to clarify the dependence on them.

Then we construct an extractor $\ext$ as follows:
\begin{description}
\item[$\ext(1^\secpar,1^{\delta^{-1}},\pp,\com,\A_{\open,\secpar},\rho_\st)$:]~
\begin{enumerate}
\item 
Store  a quantum state $\rho_\st$ in register $\regst$ and
 initialize registers $\regW$, $\regM$, $\regR$, $\regout$, $\regB$, and  $\reganc$ to be all $\ket{0}$.
\item Apply $U_{\amp,T}^{\pp,\com}$ by using the algorithm $\Amp$ in Item \ref{item:amplification_efficiency} of Lemma \ref{lem:amplification}.
\item Measure register $\regB$  and let $b$ be the outcome.
If $b=0$, then return $\bot$ and immediately halt.\footnote{More precisely, it returns $(m_\ext,\rho_\ext)\defeq (\bot,\ket{\bot}\bra{\bot})$. The same remark also applies to Step \ref{step:ext_final}}
Otherwise, proceed to the next step.   \label{step:ext_measureB}
\item 
Apply $U_{\open}$, measure registers $\regM$ and $\regR$ to obtain an outcome $(m_{\ext},r_{\ext})$, and apply $U_{\open}^\dagger$.
\label{step:ext_measure_rewind}
\item Apply ${U_{\amp,T}^{\pp,\com}}^\dagger$ by using the algorithm $\Amp$ in Item \ref{item:amplification_efficiency} of Lemma \ref{lem:amplification}.
\item Measure all registers $\regW$, $\regM$, $\regR$, $\regout$, $\regB$, and $\reganc$. 
If the outcome is not all $0$,  return $\bot$. 
Otherwise, output $m_\ext$ and the state $\rho_\ext$ in the register $\regst$.
\label{step:ext_final}
\end{enumerate}
\end{description}
We say that $\ext$ fails if it outputs $\bot$. 
We can see that $\ext$ runs in QPT and satisfies the syntactic requirements 
noting that $\Amp$ can be implemented by black-box access to $\A_{\open,\secpar}$ 
by the definition of $\Pi^{\pp,\com}$ and  Item \ref{item:amplification_efficiency} of Lemma \ref{lem:amplification}.

For $\ext$ as constructed above, we consider the following sequence of hybrid experiments:
\begin{description}
\item[$\Exp_{\mathsf{ext}}{[}\Com,\A,\ext{]}(\secpar,\delta)$:] This is the experiment as defined in Definition \ref{def:extraction_experiments}.
\item[$\Exp_{\hyb_1}{[}\Com,\A,\ext{]}(\secpar,\delta)$:] This experiment is identical to the previous one except that the experiment only checks if $\ext$ fails (i.e., $\ext$ returns $\bot$) instead of checking $m\neq m_\ext$ for the decision of outputting $\bot$.
 More concretely, we replace ``If $\Commit(\pp,m;r)\neq \com$ $\lor$  $m\neq m_\ext$'' with ``If $\Commit(\pp,m;r)\neq \com$ $\lor$ $\ext$ fails''.
 \item[$\Exp_{\hyb_2}{[}\Com,\A,\ext'{]}(\secpar,\delta)$:] This experiment is identical to the previous one except that instead of $\ext$, we use $\ext'$ that works similarly to $\ext$ except that Step \ref{step:ext_measure_rewind} is deleted  and   $m_\ext$ is omitted from the output.
 We note that the experiment is well-defined since $m_\ext$ is no longer used due to the modification made in the previous hybrid. 
 \item[$\Exp_\real{[}\Com,\A{]}(\secpar)$:] This is the experiment as defined in Definition \ref{def:extraction_experiments}.
\end{description}

We prove that output distributions of each neighboring experiments are close. 
\begin{myclaim}\label{claim:ext_to_hybone}
If $\Com$ is strong collapse-binding,  we have
\[
\{\Exp_{\mathsf{ext}}[\Com,\A,\ext](\secpar,\delta)\}_{\secpar\in\mathbb{N}}\compind  \{\Exp_{\hyb_1}[\Com,\A,\ext](\secpar,\delta)\}_{\secpar \in \mathbb{N}}.
\]
\end{myclaim}
\begin{proof}
For this part, we only need the computational binding property. 
(As remarked in Remark \ref{rem:collapse_to_computational_binding}, the strong collapse-binding property implies the computational binding property.)

We can see that the difference between these two experiments may happen only when $\Commit(\pp,m;r)= \com$, $\ext$ does not fail,  and $m\neq m_\ext$. 
We denote this event by $\Bad$.
We prove that $\Bad$ happens with a negligible probability.
When $\ext$ does not fail, we  have $b=1$ in Step \ref{step:ext_measureB} of $\ext$.
When this happens, at this point, the state in the registers $\regst$, $\regW$, $\regM$, $\regR$, $\regout$  is in the span  of $\Pi^{\pp,\com}$ by Item \ref{item:amplification_map_to_pi} of Lemma \ref{lem:amplification}.
When this happens, for $(m_\ext,r_\ext)$ obtained in Step \ref{step:ext_measure_rewind}, we have $\Commit(\pp,m_\ext;r_\ext)=\com$ by the definition of $\Pi^{\pp,\com}$. 
Therefore, when $\Bad$ happens, we have $\Commit(\pp,m;r)= \Commit(\pp,m_\ext;r_\ext)=\com$ and $m \neq m_\ext$.
Thus, if this happens with non-negligible probability, we can use $\A$ to break the computational binding of the commitment scheme.
Therefore, assuming the computational binding property of $\Com$ (which follows from the strong collapse-binding),  this happens with a negligible probability.
\end{proof}

\begin{myclaim}\label{claim:hybone_to_hybtwo}
If $\Com$ is strong collapse-binding,  we have
\[
\{\Exp_{\hyb_1}[\Com,\A,\ext](\secpar,\delta)\}_{\secpar \in \mathbb{N}}\compind  \{\Exp_{\hyb_2}[\Com,\A,\ext'](\secpar,\delta)\}_{\secpar \in \mathbb{N}}
\]
\end{myclaim}
\begin{proof}
As observed in the proof of Claim \ref{claim:ext_to_hybone}, if we have $b=1$ in Step \ref{step:ext_measureB} of $\ext$, at this point, the state in the registers $\regst$, $\regW$, $\regM$, $\regR$, $\regout$  is in the span  of $\Pi^{\pp,\com}$ by Item \ref{item:amplification_map_to_pi} of Lemma \ref{lem:amplification}.
This means that conditioned on that this happens, the registers $\regM$ and $\regR$ contain a valid opening $(m,r)$ for $\com$ under the public parameter $\pp$ (i.e., we have $\Commit(\pp,m;r)=\com$) by the definition of $\Pi^{\pp,\com}$.
Therefore, by the strong collapse-binding property of $\Com$, the experiment is computationally indistinguishable even if we omit the measurement of registers $\regM$ and $\regR$ in Step \ref{step:ext_measure_rewind}.
If we omit the measurement, then the Step \ref{step:ext_measure_rewind} just applies $U_{\open}$ followed by $U_{\open}^{\dagger}$, which is equivalent to doing nothing.
Therefore, the experiment is indistinguishable even if we delete the Step \ref{step:ext_measure_rewind}  of $\ext$ and the claim is proven.
\end{proof}

\begin{myclaim}\label{claim:hybtwo_to_real}
we have
\[
\{\Exp_{\hyb_2}[\Com,\A,\ext'](\secpar,\delta)\}_{\secpar \in \mathbb{N}}\statind_{\delta}  \{\Exp_{\real}[\Com,\A]\}(\secpar)\}_{\secpar \in \mathbb{N}}
\]
\end{myclaim}
We give a proof of Claim \ref{claim:hybtwo_to_real} in next subsection.
Combining Claim \ref{claim:ext_to_hybone}, \ref{claim:hybone_to_hybtwo}, and \ref{claim:hybtwo_to_real},
we obtain Eq. \ref{eq:extraction_conclusion}.   
 This completes the proof of Lemma \ref{lem:extraction} for the strong collapse-binding case.\\
 
 \noindent\textbf{Statistically binding case.}
 Here, we briefly explain how to modify the proof to prove Lemma \ref{lem:extraction} when $\Com$ is statistically binding instead of strong collapse-binding commitment.
The construction of $\ext$ is the same as the above strong collapse-binding  case except that it only measures the register $\regM$ and does not measure the register $\regR$ in Step \ref{step:ext_measure_rewind}.
Then, the rest of the proof is done similarly to the strong collapse-binding  case.
We explain how we can use statistical binding instead of strong collapse-binding.
In the above proof, we use strong collapse-binding property for bounding the difference between $\Exp_{\mathsf{ext}}$ and $\Exp_{\hyb_1}$ (Claim \ref{claim:ext_to_hybone}) and bounding the difference between $\Exp_{\hyb_1}$ and $\Exp_{\hyb_2}$ (Claim \ref{claim:hybone_to_hybtwo}).

For  bounding the difference between $\Exp_{\mathsf{ext}}$ and $\Exp_{\hyb_1}$, we observe that for $m_{\ext}$ extracted by $\ext$, there must exist $r_{\ext}$ such that $\Commit(\pp,m_{\ext};r_{\ext})=\com$ by the construction of $\ext$ (though $r_{\ext}$ is not measured by $\ext$ due to the modification explained above).
Therefore, if  $\Commit(\pp,m;r)=\com$, then we must have $m_{\ext}=m$ assuming that $\pp$ is binding.\footnote{See Defitinion \ref{def:binding_pp} for the definition of $\pp$ being binding.} 
Thus, replacing the check of $m_{\ext}=m$ with the check of if $\ext$ fails does not change the experiment unless $\pp$ is not binding, which happens with negligible probability as shown in Lemma \ref{lem:overwhelming_fraction_is_binding}.
 
 For bounding the difference between $\Exp_{\hyb_1}$ and $\Exp_{\hyb_2}$,  we observe that Step \ref{step:ext_measure_rewind} of $\ext$ (with the modification that it only measures the register $\regM$) does not collapse the state assuming that $\pp$ is binding. 
Therefore, we can prove the counterpart of Claim \ref{claim:hybone_to_hybtwo} based on statistical binding.

We note that an upper bound of the difference between $\Exp_{\hyb_2}$ and $\Exp_{\real}$ (Claim \ref{claim:hybtwo_to_real}) can be proven by the exactly same proof since we do not use security of commitment for this part as seen in next subsection.
\end{proof}

 \subsection{Proof of Claim \ref{claim:hybtwo_to_real}}
 In this subsection, we give a proof of Claim \ref{claim:hybtwo_to_real}, which was used in the proof of Lemma \ref{lem:extraction} in Sec. \ref{sec:proof_extraction}.
 \begin{proof}(of Claim \ref{claim:hybtwo_to_real})
We prove a stronger claim that two experiments $\Exp_{\real}[\Com,\A](\secpar)$ and $\Exp_{\hyb_2}[\Com,\A,\ext'](\secpar,\delta)$ are statistically close for any fixed $\pp$, $\com$ and $\rho_\st$. 
More precisely, for any fixed  $\pp$, $\com$ and $\rho_\st$, we consider the following two experiments:\\

\begin{tabular}{l|l}
    \begin{minipage}[t]{0.40\textwidth}
\underline{$\widetilde{\Exp}_\real^{\pp,\com}[\Com,\A](\secpar,\rho_\st)$}\\
~\vspace{3mm}\\
$(m,r,\out,\rho'_\st)\sample \A_{\open,\secpar}(\rho_\st)$,\\
\emph{If} $\Commit(\pp,m;r)\neq \com$,\\
$~~~~$\emph{Output} $\bot$\\
\emph{Else Output} $(m,r,\out,\rho'_\st)$.
    \end{minipage} 
        &~
    \begin{minipage}[t]{0.55\textwidth}
\underline{$\widetilde{\Exp}_{\hyb_2}^{\pp,\com}[\Com,\A,\ext'](\secpar,\rho_\st,\delta)$}\\
$\rho_\ext \sample \ext'(1^\secpar,1^{\delta^{-1}},\pp,\com,\A_{\open,\secpar},\rho_\st)$,\\
$(m,r,\out,\rho'_\st)\sample \A_{\open,\secpar}(\rho_\ext)$,\\
\emph{If} $\Commit(\pp,m;r)\neq \com$ \emph{or} $\ext'$ \emph{fails}\\
$~~~~$\emph{Output} $\bot$\\
\emph{Else Output} $(m,r,\out,\rho'_\st)$.
    \end{minipage}
    \end{tabular}\\
Recall that $\ext'$ is an algorithm that works similarly to $\ext$ except that it deletes Step \ref{step:ext_measure_rewind} and does not output $m_\ext$ as introduced in $\Exp_{\hyb_2}$.
We prove that for any fixed  $\pp$, $\com$, and $\rho_\st$, we have 
\begin{align}
\{\widetilde{\Exp}_\real^{\pp,\com}[\Com,\A](\secpar,\rho_\st)\}_{\secpar \in \mathbb{N}}\statind_{\delta}  
\{\widetilde{\Exp}_{\hyb_2}^{\pp,\com}[\Com,\A,\ext'](\secpar,\rho_\st,\delta)\}_{\secpar \in \mathbb{N}}.
\label{eq:hybtwo_to_real}
\end{align}
It is easy to see that if Eq. \ref{eq:hybtwo_to_real} holds for all  $\pp$, $\com$ and $\rho_\st$, then Claim \ref{claim:hybtwo_to_real} follows by averaging over  $\pp$, $\com$ and $\rho_\st$.
Moreover, since any mixed state can be understood as a probability distribution over pure states, it suffices to prove Eq. \ref{eq:hybtwo_to_real} assuming $\rho_\st$ is a pure state. 
Since we assume it is a pure state, we denote it by $\ket{\phi_{\st}}_{\regst}$ instead of $\rho_\st$.
Since we fix $\pp$ and $\com$, we omit to write $\pp,\com$ in superscripts of $\Pi$, $\Pi_{\test}$, $S_{<t}$, $S_{\geq t}$, and $U_{\amp,T}$ for notational simplicity.

\revise{In the following, we denote by $\regother$ to mean registers $(\regW,\regM,\regR,\regout,\regB,\reganc)$.
Let $R$ be an operator defined as follows:
\begin{align*}
R:=(\ket{0}\bra{0})_{\regother}U_{\amp,T}^{\dagger}(\ket{1}\bra{1})_{\regB}U_{\amp,T}.
\end{align*}
Let $\Pi_{<t}$ and $\Pi_{\ge t}$ be projections onto $S_{<t}$ and $S_{\ge t}$, respectively.  
To apply Lemma \ref{lem:state-close}, 
we define states $\ket{\phi_0}=\ket{\phi_{0,0}}+\ket{\phi_{0,1}}$ and $\ket{\phi_1}=\ket{\phi_{1,0}}+\ket{\phi_{1,1}}$ over $(\regD,\regst,\regother)$ 
where $\regD$ is an additional one-qubit register as follows:   
\begin{align*}
&\ket{\phi_{0}}:= \ket{1}_{\regD}\ket{\phi_\st}_{\regst}\ket{0}_{\regother},\\
&\ket{\phi_{0,0}}:= \ket{1}_{\regD}\Pi_{< t}\ket{\phi_\st}_{\regst}\ket{0}_{\regother},\\
&\ket{\phi_{0,1}}:= \ket{1}_{\regD}\Pi_{\ge t}\ket{\phi_\st}_{\regst}\ket{0}_{\regother},\\
&\ket{\phi_{1}}:=\ket{1}_{\regD} R \ket{\phi_\st}_{\regst}\ket{0}_{\regother}+\alpha \ket{0}_{\regD}\ket{0}_{\regst}\ket{0}_{\regother},\\
&\ket{\phi_{1,0}}:=\ket{1}_{\regD} R \Pi_{< t}  \ket{\phi_\st}_{\regst}\ket{0}_{\regother}+\alpha \ket{0}_{\regD}\ket{0}_{\regst}\ket{0}_{\regother},\\
&\ket{\phi_{1,1}}:= \ket{1}_{\regD}R \Pi_{\ge t} \ket{\phi_\st}_{\regst}\ket{0}_{\regother}
\end{align*}
for $\alpha:=\sqrt{1-\|R \ket{\phi_\st}_{\regst}\ket{0}_{\regother}\|^2}$ (so that $\ket{\phi_{1}}$ is a normalized state). 
Let $\Pi'$ be a projector over $(\regD,\regst,\regother)$ defined as 
\begin{align*}
\Pi':= (\ket{1}\bra{1})_{\regD} \ot \Pi_{\regst,\regW,\regM,\regR,\regout}\ot I_{\regB,\reganc}
\end{align*}
where $\Pi$ is as defined in Eq. \ref{eq:definition_pi}. (Note that we are omitting the superscript $\pp,\com$ here.)
Let $F$ be the quantum algorithm as in  Lemma \ref{lem:state-close} with respect to the projection $\Pi'$ as defined above. That is, $F$ is the algorithm that takes a state over 
$(\regD,\regst,\regother)$, 
applies the projective measurement $(\Pi',I-\Pi')$, and outputs the resulting state if the measurement outcome is $0$, i.e., the state is projected onto $\Pi'$ and otherwise outputs $\bot$.  
Then we have 
\begin{align}\label{eq:TD_bound}
    \TD(\widetilde{\Exp}_\real^{\pp,\com}[\Com,\A](\secpar,\ket{\phi_\st}_\regst),\widetilde{\Exp}_{\hyb_2}^{\pp,\com}[\Com,\A,\ext'](\secpar,\ket{\phi_\st}_\regst,\delta))\le \TD(F(\ket{\phi_0}),F(\ket{\phi_1}))
\end{align}
Indeed, this can be seen by the following observation. Let $G$ be a quantum algorithm that works as follows:
If its input is $\bot$, then $G$ outputs $\bot$. Otherwise, $G$ parses the input as a state over $(\regD,\regst,\regother)$, applies $U_{\open}$, measures registers $(\regM,\regR,\regout)$, and outputs registers $(\regM,\regR,\regout,\regst)$ tracing out all the other registers.   
Noting that $\Pi=U^\dagger_\open \Pi_\test U_\open$, 
 it is easy to see that  $G$ maps $F(\ket{\phi_0})$ and $F(\ket{\phi_1})$ to $\widetilde{\Exp}_\real^{\pp,\com}[\Com,\A](\secpar,\ket{\phi_\st}_\regst)$ and $\widetilde{\Exp}_{\hyb_2}^{\pp,\com}[\Com,\A,\ext'](\secpar,\ket{\phi_\st}_\regst,\delta))$, respectively. Thus Eq. \ref{eq:TD_bound} follows from monotonicity of trace distance. 
Thus, it suffices to prove 
\begin{align*}
\TD(F(\ket{\phi_0}),F(\ket{\phi_1}))\le \delta. 
\end{align*}
To show this by using Lemma \ref{lem:state-close}, we prove the following claim. 
\begin{myclaim}\label{claim:condition_check}
The following hold: 
\begin{enumerate}
    \item $\bra{\phi_{b,0}} \Pi' \ket{\phi_{b',1}}=0$ for $b,b'\in \bit$.
    \item $\|\Pi'\ket{\phi_{b,0}}\|^2\le t$ for $b\in \bit$. 
    \item $\|\ket{\phi_{1,1}}-\ket{\phi_{0,1}}\|\le \sqrt{\nu} $.
\end{enumerate}
\end{myclaim}
\begin{proof}[Proof of Claim \ref{claim:condition_check}]
The first item immediately follows from the definition. 
The second item for $b=0$ immediately follows from  Item \ref{item:amplification_success_probability} of Lemma \ref{lem:amplification}. 
To see the second item for the case of $b=1$, we observe that $R \Pi_{< t}  \ket{\phi_\st}_{\regst}\ket{0}_{\regother}$ is in the intersection of the spans of  $\Pi_{< t}\ot I_{\regB,\reganc}$ and $(\ket{0}\bra{0})_{\regother}$ by  Item \ref{item:amplification_invariance_projection} and \ref{item:amplification_invariance} of Lemma \ref{lem:amplification} and the definition of $R$. 
Then the desired inequality follows from Item \ref{item:amplification_success_probability} of Lemma \ref{lem:amplification} similarly to the case of $b=0$. 
To see the third item, we observe that 
Item
\ref{item:amplification_amplification} of Lemma \ref{lem:amplification} implies 
$$
\|(\ket{1}\bra{1})_{\regB}U_{\amp,T}\Pi_{\ge t}\ket{\phi_\st}_{\regst}\ket{0}_{\regother}\|^2\le \nu.
$$
Thus, Lemma \ref{lem:gentle_measurement} implies
\begin{align*}
   \|\Pi_{\ge t}\ket{\phi_\st}_{\regst}\ket{0}_{\regother}-R\Pi_{\ge t}\ket{\phi_\st}_{\regst}\ket{0}_{\regother}\|\le \sqrt{\nu}. 
\end{align*}
This immediately implies the third item.
\end{proof}
By Claim \ref{claim:condition_check}, we can apply Lemma \ref{lem:state-close} to obtain
\begin{align*}
\TD(F(\ket{\phi_0}),F(\ket{\phi_1}))\le 
\sqrt{4t+2\sqrt{\nu}}
=
\delta 
\end{align*}
where the final inequality follows from
$t=\delta^2/6$ and $\nu=t^2$. 
This completes the proof of Claim \ref{claim:hybtwo_to_real}.
}
\end{proof}

\fi
\ifnum\submission=0 
\section{Post-Quantum \texorpdfstring{$\epsilon$}{epsilon}-Zero-Knowledge Proof}\label{sec:ZK_from_Collapse}
In this section, we prove the following theorem.
\begin{theorem}\label{thm:ZK_proof_from_QLWE}
If the QLWE assumption holds, then there exists a $5$-round post-quantum black-box $\epsilon$-zero-knowledge proof for all $\NP$ languages.
\end{theorem}
Then we generalize it to obtain the following theorem in Sec. \ref{sec:generalization_collapsing}.
\begin{theorem}\label{thm:ZK_proof_from_collapsing}
If a collapsing hash function exists, then there exists a $5$-round post-quantum black-box $\epsilon$-zero-knowledge proof for all $\NP$ languages.
\end{theorem}
\else 
\section{Post-Quantum \texorpdfstring{$\epsilon$}{epsilon}-Zero-Knowledge Proof and Argument}\label{sec:ZK_from_Collapse}
In this section, we prove the following theorems.
\begin{theorem}\label{thm:ZK_proof_from_QLWE}
If the QLWE assumption holds, then there exists a $5$-round post-quantum black-box $\epsilon$-zero-knowledge proof for all $\NP$ languages.
\end{theorem}
\begin{theorem}\label{thm:ZK_proof_from_collapsing}
If a collapsing hash function exists, then there exists a $5$-round post-quantum black-box $\epsilon$-zero-knowledge proof for all $\NP$ languages.
\end{theorem}
\begin{theorem}\label{thm:ZK_argument_from_OWF}
If post-quantunm OWF exists, then there exists a 9-round  post-quantum black-box $\epsilon$-zero-knowledge argument for all $\NP$ languages.
\end{theorem}
In the rest of this section, we prove Theorem \ref{thm:ZK_proof_from_QLWE} and \ref{thm:ZK_proof_from_collapsing}.
The proof of Theorem \ref{thm:ZK_argument_from_OWF} is given in
\ifnum\cameraready=1
the full version. 
\else
Appendix \ref{sec:ZK_from_OWF}.
\fi
\fi

\subsection{Construction}\label{sec:construction_proof}
Our construction is the same as the Golderich-Kahan protocol \cite{JC:GolKah96} except that we instantiate the verifier's commitment with a strong collapse-binding commitment and we rely on a post-quantum delayed-witness $\Sigma$-protocol.
Specifically, our construction is built on the following ingredients:
\begin{itemize}
    \item A commitment scheme $(\CBSetup,\CBCommit)$ that is statistical hiding and strong collapse-binding with message space $\bit^\secpar$ and randomness space $\calR$. 
    As noted in Sec.  
    \ref{sec:commitment},    
    such a commitment scheme exists under the QLWE assumption.
    \item A delayed-witness $\Sigma$-protocol $(\Sigma.\pro_1,\Sigma.\pro_3,\Sigma.\ver)$ for an $\NP$ language $\lang$ as defined in Definition \ref{def:delayed_witness_sigma_protocol}.
     As noted in Sec. \ref{sec:delayed_sigma}, such a protocol exists under the QLWE assumption.
\end{itemize}

Then our construction of post-quantum black-box $\epsilon$-zero-knowledge proof is given in Figure \ref{fig:ZK_from_Collapse}.

\protocol
{Protocol \ref{fig:ZK_from_Collapse}}
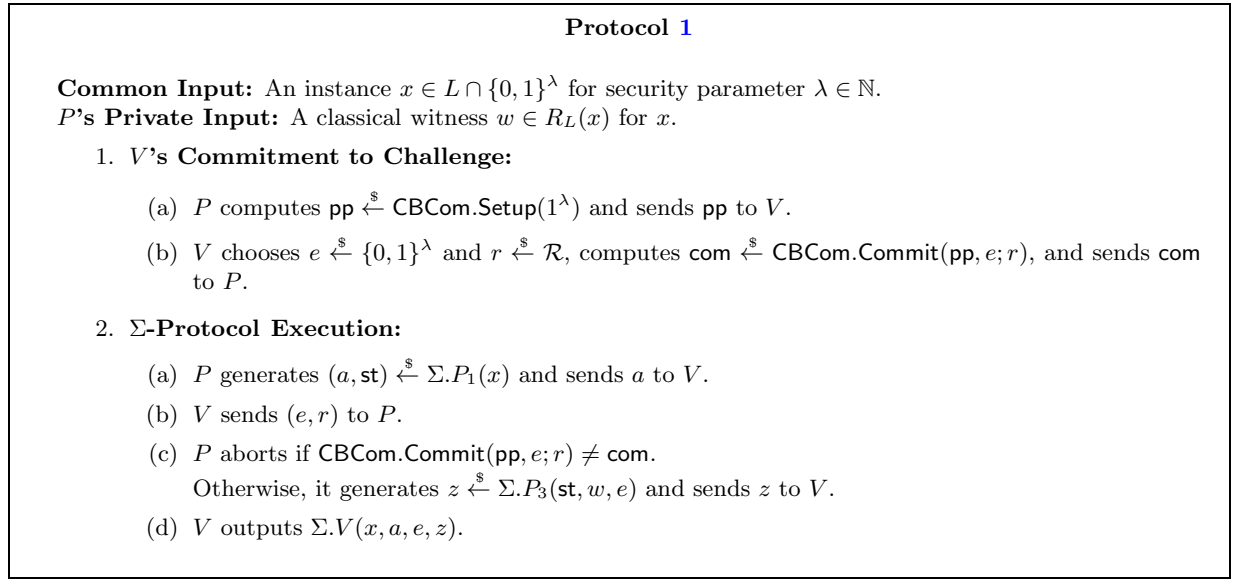
{Constant-Round Post-Quantum $\epsilon$-Zero-Knowledge Proof for $\lang \in \NP$}
{fig:ZK_from_Collapse}
{
\textbf{Common Input:} An instance $x\in \lang \cap \bit^{\secpar}$ for security parameter $\secpar \in \mathbb{N}$.\\
\textbf{$\pro$'s Private Input:} A classical witness $w\in \rel_\lang(x)$ for $x$. 
\begin{enumerate}
\item \textbf{$\ver$'s Commitment to Challenge:}
\begin{enumerate}
    \item $\pro$ computes $\pp\sample \CBSetup(1^\secpar)$ and sends $\pp$ to $\ver$.
    \item $\ver$ chooses $e\sample \bit^\secpar$ and $r\sample \randspace$, computes $\com \sample \CBCommit(\pp,e;r)$, and sends $\com$ to $\pro$.
\end{enumerate}
\item \textbf{$\Sigma$-Protocol Execution:}
\begin{enumerate}
    \item $\pro$ generates $(a,\st)\sample \Sigma.\pro_1(x)$ and sends $a$ to $\ver$.
    \item $\ver$ sends $(e,r)$ to $\pro$.
    \label{Step_protocol_ZK_proof_open}
    \item $\pro$ aborts if $\CBCommit(\pp,e;r)\neq \com$.\\ Otherwise, it generates $z\sample \Sigma.\pro_3(\st,w,e)$ and sends $z$ to $\ver$.
       \label{Step_protocol_ZK_proof_check}
    \item $\ver$ outputs $\Sigma.\ver(x,a,e,z)$. 
\end{enumerate}
\end{enumerate}
}

The completeness of the protocol clearly follows from that of the underlying $\Sigma$-protocol.
In Sec. \ref{sec:statistical_soundness} and \ref{sec:proof_ZK_collapsing}, we prove that this protocol satisfies statistical soundness and quantum black-box $\epsilon$-zero-knowledge.
Then we obtain Theorem \ref{thm:ZK_proof_from_QLWE}.

\subsection{Statistical Soundness}\label{sec:statistical_soundness}
This is essentially the same as the proof in  \cite{JC:GolKah96}, but we give a proof for completeness.

For $x\notin \lang$ an unbounded-time cheating prover $\pro^*$, we consider the following sequence of hybrids.
We denote by $\win_i$ the event that $\pro^*$ wins in $\hyb_i$.

\begin{description}
\item[$\hyb_1$:]
This is the original game. That is,
\begin{enumerate}
    \item $\pro^*$ sends $\pp$ to $\ver$.
    \item $\ver$ chooses $e\sample \bit^\secpar$ and $r\sample \randspace$, computes $\com \sample \CBCommit(\pp,e;r)$, and sends $\com$ to $\pro^*$.
    \label{step:soundness_experiment_send_com}
    \item $\pro^*$ sends $a$ to $\ver$.
    \item $\ver$ sends $(e,r)$ to $\pro^*$
    \label{step:soundness_experiment_send_er}
    \item $\pro^*$ sends $z$ to $\ver$.
\end{enumerate}
We say that $\pro^*$ wins if we have $\Sigma.\ver(x,a,e,z)=\top$.
\item[$\hyb_2$:]
This hybrid is identical to the previous one except that in Step \ref{step:soundness_experiment_send_er}, $\ver$ uniformly chooses $r'$ such that $\com=\CBCommit(\pp,e;r')$ and sends $(e,r')$ to $\pro^*$ instead of $(e,r)$.
We note that this procedure may be inefficient.

This is just a conceptual change and thus we have $\Pr[\win_1]=\Pr[\win_2]$.

\item[$\hyb_3$:]
This hybrid is identical to the previous one except that in Step \ref{step:soundness_experiment_send_com}, $\ver$ sends $\com\sample \CBCommit(\pp,0^\ell;r)$ and the generation of $e$ is delayed to Step \ref{step:soundness_experiment_send_er}.

Since no information of $r$ is given to $\pro^*$ due to the modification made in $\hyb_2$, 
by the statistical hiding property of $\CBCom$, we have $|\Pr[\win_3]-\Pr[\win_2]|=\negl(\secpar)$.

Now, it is easy to prove $\Pr[\win_3]=\negl(\secpar)$ by reducing it to the statistical soundness of the $\Sigma$-protocol.
Namely, we consider a cheating prover $\Sigma.\pro^*$ against the $\Sigma$-protocol that works as follows.
\begin{enumerate}
    \item $\Sigma.\pro^*$ runs $\pro^*$ to get the first message $\pp$.
    \item $\Sigma.\pro^*$ computes  $\com\sample \CBCommit(\pp,0^\ell;r)$, sends $\com$ to $\pro^*$, and gets the third message $a$.
    Then $\Sigma.\pro^*$ sends $a$ to its own external challenger as the first message of the $\Sigma$-protocol.
    \item Upon receiving a challenge $e$ from the external challenger, $\Sigma.\pro^*$ uniformly chooses $r'$ such that $\com=\CBCommit(\pp,e;r')$, sends $(e,r')$ to $\pro^*$, and gets the $\pro^*$'s final message $z$.
    Then $\Sigma.\pro^*$ sends $z$ to the external challenger. 
\end{enumerate}
It is easy to see that $\Sigma.\pro^*$ perfectly simulates the environment in $\hyb_3$ for $\pro^*$. Therefore, $\Sigma.\pro^*$'s winning probability is equal to $\Pr[\win_3]$.
On the other hand, by soundness of the $\Sigma$-protocol, $\Sigma.\pro^*$'s winning probability is  $\negl(\secpar)$.
Therefore we have $\Pr[\win_3]=\negl(\secpar)$.
\end{description}

Combining the above, we have $\Pr[\win_1]=\negl(\secpar)$, which means that the protocol satisfies the statistical soundness.
\subsection{Quantum Black-Box \texorpdfstring{$\epsilon$}{epsilon}-Zero-Knowledge}\label{sec:proof_ZK_collapsing}
\noindent\textbf{Structure of the Proof.}
A high-level structure of our proof is similar to that of \cite{STOC:BitShm20}.
Specifically, we first construct simulators $\siml_\abort$ and $\siml_\nonabort$ that simulate the ``aborting case" and ``non-aborting case", respectively.
More precisely, $\siml_\abort$ correctly simulates the verifier's view if the verifier aborts and  otherwise returns a failure symbol $\fail$ and $\siml_\nonabort$ correctly simulates the verifier's view if the verifier does not abort and  otherwise returns a failure symbol $\fail$.
Then we consider a combined simulator $\siml_\comb$ that runs either of $\siml_\abort$ or  $\siml_\nonabort$ with equal probability.
Then $\siml_\comb$ correctly simulates  the verifier's view conditioned on that the output is not $\fail$, and it returns $\fail$ with probability almost $1/2$.
By applying the Watrous' quantum rewinding lemma (Lemma \ref{lem:quantum_rewinding}) to $\siml_\comb$, we can convert it to a full-fledged simulator.

Though the above high-level structure is similar to \cite{STOC:BitShm20}, the analyses of simulators $\siml_\abort$ and  $\siml_\nonabort$ are completely different from \cite{STOC:BitShm20} since we consider different protocols.
While the analysis of $\siml_\abort$ is easy, the analysis of $\siml_\nonabort$ is a little more complicated as it requires the extraction lemma (Lemma \ref{lem:extraction}), which was developed in Sec. \ref{sec:extraction}.\\

\noindent\textbf{Proof of Quantum Black-Box $\epsilon$-Zero-Knowledge.}
For clarity of exposition, we first show the quantum $\epsilon$-zero-knowledge property ignoring that the simulator should be black-box.  That is, 
we give the full description of the malicious verifier and its quantum advice as part of the simulator’s input instead of only the oracle access to the verifier. At the end of the proof, we explain that the simulator is indeed  black-box.

In quantum $\epsilon$-zero-knowledge, we need to show a simulator $\siml$ that takes an accuracy parameter $1^{\epsilon^{-1}}$ as part of its input.
We assume $\epsilon(\secpar)= o(1)$ without loss of generality since the other case trivially follows from this case.
Without loss of generality, we can assume that a malicious verifier $\ver^*$ does not terminate the protocol before the prover aborts since it does not gain anything by declaring the termination.
We say that $\ver^*$ aborts if it fails to provide a valid opening $(e,r)$ to $\com$ in Step \ref{Step_protocol_ZK_proof_open} (i.e., the prover aborts in Step \ref{Step_protocol_ZK_proof_check}).

First, we construct a simulator $\siml_\comb$, which returns a special symbol $\fail$ with probability roughly $1/2$ but almost correctly simulates the output of $\ver^*_\secpar$ conditioned on that it does not return $\fail$.
The simulator $\siml_\comb$ uses simulators $\siml_\abort$ and $\siml_\nonabort$ as sub-protocols:

\begin{description}
\item[$\siml_\comb(x,1^{\epsilon^{-1}},\ver^*_\secpar,\rho_\secpar)$:]~
\begin{enumerate}
    \item Choose $\mathsf{mode}\sample \{\abort,\nonabort\}$.
    \item Run $\siml_{\mathsf{mode}}(x,1^{\epsilon^{-1}},\ver^*_\secpar,\rho_\secpar)$. 
    \item Output what $\siml_{\mathsf{mode}}$ outputs.
\end{enumerate}
\item[$\siml_\abort(x,1^{\epsilon^{-1}},\ver^*_\secpar,\rho_\secpar)$:]\footnote{Though $\siml_\abort$ does not depend on $\epsilon$, we include $1^{\epsilon^{-1}}$ in the input for notational uniformity.}~
\begin{enumerate}
   \item Set $\ver^*_\secpar$'s internal state to $\rho_\secpar$.
    \item Compute $\pp\sample \CBSetup(1^\secpar)$ and send $\pp$ to $\ver^*_\secpar$.
    \item $\ver^*_\secpar$ returns $\com$.
    \item Compute $(a,\st)\sample \Sigma.\pro_1(x)$ and send $a$ to $\ver^*_\secpar$.
    \item $\ver^*_\secpar$ returns $(e,r)$.
    \item Return $\fail$ and abort if $\CBCommit(\pp,e;r)= \com$.\\
    Otherwise, let $\ver^*_\secpar$ output the final output notifying that the prover aborts.
    \item The final output of $\ver^*_\secpar$ is treated as the output $\siml_\abort$.  
\end{enumerate}
\item[$\siml_\nonabort(x,1^{\epsilon^{-1}},\ver^*_\secpar, \rho_\secpar)$:]~
\begin{enumerate}
\item Set $\ver^*_\secpar$'s internal state to $\rho_\secpar$.
   \item Compute $\pp\sample \CBSetup(1^\secpar)$ and send $\pp$ to $\ver^*_\secpar$.
    \item $\ver^*_\secpar$ returns $\com$. 
    Let $\rho_\st$ be the internal state of $\ver^*_\secpar$ at this point.
    \label{step:simna_generate_com}
    \item Compute $(e_\ext,\rho_\ext)\sample \ext(1^\secpar,1^{\delta^{-1}},\pp,\com,\A_{\open,\secpar},\rho_{\st})$ where $\ext$ is as in Lemma \ref{lem:extraction} for the commitment scheme $\CBCom$, $\delta\defeq \frac{\epsilon^2}{3600\log^4(\secpar)}$, and $\A=(\A_{\commit,\secpar},\A_{\open,\secpar})$ as defined below:
    \begin{description}
    \item[$\A_{\commit,\secpar}(\pp;\rho_\secpar)$:]
    It sets  $\ver^*_\secpar$'s internal state  to $\rho_\secpar$ and sends $\pp$ to $\ver^*_\secpar$. Let $\com$ be the response by $\ver^*_\secpar$ and $\rho_\st$ be the internal state of $\ver^*_\secpar$ at this point. 
    It outputs $(\com,\rho_\st)$. 
    \item[$\A_{\open,\secpar}(\rho_\st)$:] It generates $(a,\st)\sample \Sigma.\pro_1(x)$,\footnote{
        We note that we consider $x$ to be hardwired into $\A_{\open,\secpar}$. 
    We also note that though $\A_{\open,\secpar}$ does not take explicit randomness, it can generate randomness by say, applying Hadamard on its working register and then measuring it.
    }  sets $\ver^*_\secpar$'s internal state to $\rho_\st$, 
    and
    sends $a$ to $\ver^*_\secpar$.
    Let $(e,r)$ be the response by $\ver^*_\secpar$ and let $\rho'_{\st}$ be the internal state of $\ver^*_\secpar$ at this point. 
    It outputs $(e,r,\out:=(a,\st),\rho'_{\st})$.
    \end{description}
   Here, we remark that 
   $\ver^*_\secpar$'s internal register corresponds to $\regst$ and 
   $e$ corresponds to $m$  in the notation of Lemma \ref{lem:extraction}.
        \label{step:simna_ext}
   \item Set the verifier's internal state to $\rho_\ext$. 
 \label{step:simna_reinitialize}
    \item Compute $(a,z)\sample \siml_{\Sigma}(x,e_{\ext})$ and send $a$ to $\ver^*_\secpar$. 
    \label{step:simna_simlate_Sigma}
    \item $\ver^*_\secpar$ returns $(e,r)$.  
    \label{step:simna_generate_er}
    \item Return $\fail$ and abort if $e\neq e_\ext$ or $\CBCommit(\pp,e;r)\neq \com$.\\
    Otherwise, send $z$ to $\ver^*_\secpar$.
    \label{step:simna_send_z}
    \item The final output of $\ver^*_\secpar$ is treated as the output $\siml_\nonabort$.   \label{step:simna_final}
\end{enumerate}
\end{description}

Intuitively, $\siml_\abort$ (resp. $\siml_\nonabort$) is a simulator that simulates the verifier's view in the case that verifier aborts (resp. does not abort).

More formally, we prove the following lemmas.
\begin{lemma}[$\siml_\abort$ simulates the aborting case.]\label{lem:siml_abort}
For any non-uniform QPT malicious verifier $\ver^*=\{\ver^*_\secpar,\rho_\secpar\}_{\secpar\in\mathbb{N}}$, let $\OUT_{\ver^*_{\abort}}\execution{\pro(w)}{\ver^*_\secpar(\rho_\secpar)}(x)$ be the $\ver^*_\secpar$'s final output that  is replaced with $\fail$ if $\ver^*_\secpar$ does not abort.
Then we have
\begin{align*}
  \{\OUT_{\ver^*_{\abort}}\execution{\pro(w)}{\ver^*_\secpar(\rho_\secpar)}(x)\}_{\secpar,x,w}  \equiv
  \{\siml_\abort(x,1^{\epsilon^{-1}},\ver^*_\secpar,\rho_\secpar)\}_{\secpar,x,w}.
\end{align*}
where $\secpar \in \mathbb{N}$, $x\in \lang\cap \bit^\secpar$, and $w\in \rel_\lang(x)$.
\end{lemma}
\begin{proof}
Since $\siml_\abort$ perfectly simulates the real execution for $\ver^*_\secpar$ when it aborts, Lemma \ref{lem:siml_abort} immediately follows.
\end{proof}
\begin{lemma}[$\siml_\nonabort$ simulates the non-aborting case.]\label{lem:siml_nonabort}
For any non-uniform QPT malicious verifier $\ver^*=\{\ver^*_\secpar,\rho_\secpar\}_{\secpar\in\mathbb{N}}$,
let $\OUT_{\ver^*_{\nonabort}}\execution{\pro(w)}{\ver^*_\secpar(\rho_\secpar)}(x)$ be the $\ver^*_\secpar$'s final output that is replaced with $\fail$ if $\ver^*_\secpar$ aborts.
Then we have
\begin{align*}
\{\OUT_{\ver^*_{\nonabort}}\execution{\pro(w)}{\ver^*_\secpar(\rho_\secpar)}(x)\}_{\secpar,x,w}  \compind_{\delta}
  \{\siml_\nonabort(x,1^{\epsilon^{-1}},\ver^*_\secpar, \rho_\secpar)\}_{\secpar,x,w}
\end{align*}
where $\secpar \in \mathbb{N}$, $x\in \lang \cap \bit^{\secpar}$, and $w\in \rel_\lang(x)$.
\end{lemma}
\begin{proof}
Here, we analyze $\siml_\nonabort(x,1^{\epsilon^{-1}},\ver^*_\secpar, \rho_\secpar)$.
In the following, we consider hybrid simulators $\simlnai$ for $i=1,2,3$.
We remark that they also take the witness $w$ as input unlike $\siml_\nonabort$.

\begin{description}
\item[$\simlnaone$:] This simulator works similarly to $\siml_\nonabort(x,1^{\epsilon^{-1}},\ver^*_\secpar, \rho_\secpar)$ except that it generates $(a,\st)\sample \Sigma.\pro_1(x)$ and $z\sample \Sigma.\pro_3(\st,w,e_{\ext})$ instead of $(a,z)\sample \siml_{\Sigma}(x,e_\ext)$ in Step \ref{step:simna_simlate_Sigma}.

By the special honest-verifier zero-knowledge property of the $\Sigma$-protocol, we have 
\begin{align*}
\{\siml_\nonabort(x,1^{\epsilon^{-1}},\ver^*_\secpar, \rho_\secpar)\}_{\secpar,x,w}
\compind 
\{\{\siml_{\nonabort,1}(x,w,1^{\epsilon^{-1}},\ver^*_\secpar,\rho_\secpar)\}_{\secpar,x,w}\}_{\secpar,x,w}
\end{align*}
where $\secpar\in \mathbb{N}$, $x\in \lang \cap \bit^{\secpar}$, and $w\in \rel_\lang(x)$.

\item[$\simlnatwo$:] This simulator works similarly to $\simlnaone$ except that the generation of $z$ is delayed until Step \ref{step:simna_send_z} and it is generated as $z\sample \Sigma.\pro_3(\st,w,e)$ instead of $z\sample \Sigma.\pro_3(\st,w,e_{\ext})$.

The modification does not affect the output distribution since it outputs $\fail$ if $e\neq e_\ext$ and if $e=e_\ext$, then this simulator works in exactly the same way as the previous one.
Therefore we have 
\begin{align*}
\{\siml_{\nonabort,1}(x,w,1^{\epsilon^{-1}},\ver^*_\secpar,\rho_\secpar)\}_{\secpar,x,w}  \equiv \{\siml_{\nonabort,2}(x,w,1^{\epsilon^{-1}},\ver^*_\secpar,\rho_\secpar)\}_{\secpar,x,w}
\end{align*}
where $\secpar\in \mathbb{N}$, $x\in \lang \cap \bit^{\secpar}$, and $w\in \rel_\lang(x)$.

\item[$\simlnathree$:] This simulator works similarly to $\simlnatwo$ except that Step \ref{step:simna_ext} and \ref{step:simna_reinitialize} are deleted and the check of $e\neq e_\ext$ in Step \ref{step:simna_send_z} is omitted.
That is, it outputs $\fail$ in  Step \ref{step:simna_send_z} if and only if we have $\CBCommit(\pp,e;r)\neq \com$.
We note that $e_\ext$ and $\rho_\ext$ are no longer used at all and thus need not be generated.

We can see that 
Step \ref{step:simna_generate_com} is exactly the same as executing $(\com,\rho_\st)\sample \A_{\commit,\secpar}(\pp;\rho_\secpar)$  and
Step \ref{step:simna_simlate_Sigma} and \ref{step:simna_generate_er} of previous and this experiments are exactly the same as executing 
$(e,r,\out=(a,\st),\rho'_\st)\sample \A_{\open,\secpar}(\rho_\ext)$ and
$(e,r,\out=(a,\st),\rho'_\st)\sample \A_{\open,\secpar}(\rho_\st)$, respectively where we define $\rho'_\st$ in simulated experiments as $\ver^*_\secpar$'s internal state after Step \ref{step:simna_generate_er}. 
Moreover, the rest of execution of the simulators can be done given $(\pp,\com,e,r,\out=(a,\st),\rho'_\st)$. 
Therefore, by a straightforward reduction to Lemma \ref{lem:extraction}, we have 
\begin{align*}
\{\siml_{\nonabort,2}(x,w,1^{\epsilon^{-1}},\ver^*_\secpar,\rho_\secpar)\}_{\secpar,x,w}  \compind_\delta \{\siml_{\nonabort,3}(x,w,1^{\epsilon^{-1}},\ver^*_\secpar,\rho_\secpar)\}_{\secpar,x,w}
\end{align*}
where $\secpar\in \mathbb{N}$, $x\in \lang \cap \bit^{\secpar}$, and $w\in \rel_\lang(x)$.
\end{description}

We can see that $\simlnathree$ perfectly simulates the real execution for $\ver^*_\secpar$ and outputs $\ver^*_\secpar$'s output conditioned on that $\ver^*_\secpar$ does not abort, and just outputs $\fail$ otherwise.
Therefore, we have 
\begin{align*}
\{\siml_{\nonabort,3}(x,w,1^{\epsilon^{-1}},\ver^*_\secpar,\rho_\secpar)\}_{\secpar,x,w} \equiv \{\OUT_{\ver^*_{\nonabort}}\execution{\pro(w)}{\ver^*_\secpar(\rho_\secpar)}(x)\}_{\secpar,x,w} 
\end{align*}
where $\secpar\in \mathbb{N}$, $x\in \lang \cap \bit^{\secpar}$, and $w\in \rel_\lang(x)$. 
Combining the above, Lemma \ref{lem:siml_nonabort} is proven.
\end{proof}

By combining  Lemmas \ref{lem:siml_abort} and \ref{lem:siml_nonabort}, we can prove the following lemma.

\begin{lemma}[$\siml_\comb$ simulates $\ver^*_\secpar$'s output with probability almost $1/2$]\label{lem:siml_comb}
For any non-uniform QPT malicious verifier $\ver^*=\{\ver^*_\secpar,\rho_\secpar\}_{\secpar\in\mathbb{N}}$, 
let $\pcombsuc$ be the probability that $\siml_\comb(x,1^{\epsilon^{-1}},\ver^*_\secpar, \rho_\secpar)$ does not return $\fail$ and $D_{\mathsf{sim},\comb}( x,\allowbreak 1^{\epsilon^{-1}},\ver^*_\secpar,\rho_\secpar)$ be a conditional  distribution of $\siml_\comb(x,1^{\epsilon^{-1}},\ver^*_\secpar, \rho_\secpar)$, conditioned on that it does not return $\fail$.
There exists a negligible function $\negl$ such that for any $x=\{x_\secpar \in \lang \cap \bit^\secpar\}_{\secpar \in \mathbb{N}}$,  we have 
\begin{align}
    \left|\pcombsuc-1/2\right|\leq \delta/2+\negl(\secpar). \label{eq:pcombsuc}
\end{align}
Moreover, we have 
\begin{align}
   \{\OUT_{\ver^*}\execution{\pro(w)}{\ver^*_\secpar(\rho_\secpar)}(x)\}_{\secpar,x,w}
    \compind_{4\delta} 
    \{D_{\mathsf{sim},\comb}(x,1^{\epsilon^{-1}},\ver^*_\secpar,\rho_\secpar)\}_{\secpar,x,w} \label{eq:siml_comb}
\end{align}
where $\secpar\in \mathbb{N}$, $x\in \lang \cap \bit^{\secpar}$, and $w\in \rel_\lang(x)$. 
\end{lemma}
\begin{proof}(sketch.)
Intuition of the proof is very easy: 
By Lemma \ref{lem:siml_abort} and \ref{lem:siml_nonabort}, $\siml_\abort$ and $\siml_\nonabort$ almost simulate the real output distribution of $\ver^*_\secpar$ conditioned on that $\ver^*_\secpar$ aborts and does not abort, respectively.
Therefore, if we randomly guess  if $\ver^*_\secpar$ aborts and runs either of $\siml_\abort$ and $\siml_\nonabort$ that successfully works for the guessed case, the output distribution is close to  the real output distribution of $\ver^*_\secpar$ conditioned on that the guess is correct, which happens with probability almost $1/2$.

Indeed, the actual proof is based on the above idea, but for obtaining concrete bounds as in Eq. \ref{eq:pcombsuc} and \ref{eq:siml_comb}, we need some tedious calculations. 
We give a full proof in
\ifnum\cameraready=1
the full version 
\else 
Appendix \ref{sec:siml_comb} 
\fi 
since the proof is easy and  very similar to that in \cite{STOC:BitShm20} (once we obtain Lemma \ref{lem:siml_abort} and \ref{lem:siml_nonabort}).
\end{proof}

Then, we convert $\siml_\comb$ to a full-fledged simulator that does not return $\fail$ by using the quantum rewinding lemma (Lemma \ref{lem:quantum_rewinding}).
Namely, we let $\sfQ$ be a quantum algorithm that takes $\rho_\secpar$ as input and outputs $\siml_\comb(x,1^{\epsilon^{-1}},\ver^*_\secpar, \rho_\secpar)$ where $b\defeq 0$ if and only if  it does not return $\fail$, $p_0\defeq \frac{1}{4}$, $q\defeq \frac{1}{2}$, $\gamma\defeq \delta$, and $T\defeq 2\log (1/\delta)$.
Then it is easy to check that the conditions for Lemma \ref{lem:quantum_rewinding} is satisfied by Eq. \ref{eq:pcombsuc} in Lemma \ref{lem:siml_comb} (for sufficiently large $\secpar$).
Then by using Lemma \ref{lem:quantum_rewinding}, we can see that $\sfR(1^T,\sfQ,\rho_\secpar)$ runs in time $T\cdot |\sfQ|=\poly(\secpar)$ and its output (seen as a mixed state) has a trace distance bounded by $4\sqrt{\gamma}\frac{\log(1/\gamma)}{p_0(1-p_0)}$ from  $D_{\mathsf{sim},\comb}(x,1^{\epsilon^{-1}},\ver^*_\secpar,\rho_\secpar)$.
Since we have $\gamma=\delta=\frac{\epsilon^2}{3600\log^4(\secpar)}=1/\poly(\secpar)$, we have $4\sqrt{\gamma}\frac{\log(1/\gamma)}{p_0(1-p_0)}< 30\sqrt{\gamma} \log^2 (\secpar)=\frac{\epsilon}{2}$ for sufficiently large $\secpar$
 where we used $\log(1/\gamma)=\log(\poly(\secpar))=o(\log^2(\secpar))$.
Thus, by combining the above and Eq. \ref{eq:siml_comb} in Lemma \ref{lem:siml_comb}, if we define $\siml(x,1^{\epsilon^{-1}},\ver^*_\secpar,\rho_\secpar)\defeq \sfR(1^T,\sfQ,\rho_\secpar)$, then we have 
\begin{align*}
   \OUT_{\ver^*}\execution{\pro(w)}{\ver^*_\secpar(\rho_\secpar)}(x)
    \compind_{\frac{\epsilon}{2}+4\delta} 
    \siml(x,1^{\epsilon^{-1}},\ver^*_\secpar,\rho_\secpar).
\end{align*}
We can conclude the proof of quantum $\epsilon$-zero-knowledge by noting that  we have
    $\frac{\epsilon}{2}+4\delta< \epsilon$
since we have $\delta= \frac{\epsilon^2}{3600\log^4(\secpar)} < \frac{\epsilon}{8}$.

\paragraph{Black-Box Simulation.}
Here, we explain that the simulator $\siml$ constructed as above only needs black-box access to the verifier. 
What we need to show are that $\siml$ applies the unitary part $U_{\ver^*_\secpar}$ of $\ver^*_\secpar$ and its inverse $U_{\ver^*_\secpar}^\dagger$ only as oracles and $\siml$ does not directly act on $\ver^*_\secpar$'s internal register. 
There are two parts of the construction of $\siml$ that are not obviously black-box.
The first is Step \ref{step:simna_ext} and \ref{step:simna_reinitialize} of $\siml_\nonabort$ where it runs the extraction algorithm $\ext$ of Lemma \ref{lem:extraction}, and the second is the conversion from $\siml_\comb$ to $\siml$ using $\sfR$ in Lemma \ref{lem:quantum_rewinding}.
In the following, we explain that both steps can be implemented by black-box access to the verifier.

\begin{enumerate}
    \item By Lemma \ref{lem:extraction}, $\ext$ uses the unitary part of $\A_{\open,\secpar}$ and its inverse only in a black-box manner, and they can be implemented by black-box access to $U_{\ver^*_\secpar}$ and $U_{\ver^*_\secpar}^\dagger$.
    Moreover, since register $\regst$ in the notation of   Lemma \ref{lem:extraction} corresponds to the internal register of $\ver^*_\secpar$ in our context, the lemma ensures that $\ext$ does not directly act on it.
    Also, $\siml_{\nonabort}$ need not explicitly set $\ver^*_\secpar$'s internal register to $\rho_\ext$ in Step \ref{step:simna_reinitialize}   if we do the above black-box simulation since a state in the register automatically becomes $\rho_\ext$ after the execution as stated in Lemma \ref{lem:extraction}.
    Therefore, this step can be implemented by black-box access to $\ver^*_\secpar$.
    \item Given the above observation, we now know that both $\siml_\abort$ and $\siml_\nonabort$ only need black-box access to $\ver^*_\secpar$.
    This means that $\sfQ$ only needs black-box access to $\ver^*_\secpar$.
    Since $\sfR$ only uses $\sfQ$ as oracles that perform the unitary part of $\sfQ$ and its inverse as stated in  Lemma \ref{lem:quantum_rewinding} and they can be implemented by black-box access to $\ver^*_\secpar$, 
    $\sfR$ uses $U_{\ver^*_\secpar}$ and $U_{\ver^*_\secpar}^\dagger$ only as oracles.
    Moreover, since the register $\reginp$ in Lemma \ref{lem:quantum_rewinding} corresponds to the internal register of $\ver^*_\secpar$ in our context, $\sfR$ does not directly act on it.
 \end{enumerate}
 By the above observations, we can see that the simulator $\siml$ only needs black-box access to $\ver^*_\secpar$.
 
 \subsection{Instantiation from Collapsing Hash Function}\label{sec:generalization_collapsing}
 Our construction in Figure \ref{fig:ZK_from_Collapse} is based on two building blocks: a statistically hiding and strong collapse-binding commitment scheme and a delayed-witness $\Sigma$-protocol.
 Though the former can be instantiated by a collapsing hash function, we do not know how to instantiate the latter by a collapsing hash function since it needs non-interactive commitment that is not known to be implied by  collapsing hash functions.
 However, we can just use a $4$-round version of a delayed-witness $\Sigma$-protocol where the first message ``commitment" in the $\Sigma$-protocol is instantiated based on Naor's commitments \cite{JC:Naor91} instead of a non-interactive one.
 Since Naor's commitments  can be instantiated under any OWF and collapsing hash function is trivially also one-way, we can instantiate  the 4-round version of a delayed-witness $\Sigma$-protocol based on a collapsing hash function. 
We can prove security of the construction based on 4-round version of a delayed-witness $\Sigma$-protocol in essentially the same manner as the security proofs in Sec. \ref{sec:statistical_soundness} and \ref{sec:proof_ZK_collapsing}. We also note that this does not increase the number of rounds of our construction. 
Based on these observations, we obtain Theorem \ref{thm:ZK_proof_from_collapsing}. 

\section{Post-Quantum \texorpdfstring{$\epsilon$}{epsilon}-Zero-Knowledge Argument from OWF}\label{sec:ZK_from_OWF}
In this section, we construct a constant-round $\epsilon$-zero-knowledge argument from any post-quantum OWF.
\begin{theorem}\label{thm:ZK_argument_from_OWF}
If post-quantunm OWF exists, then there exists a 9-round  post-quantum black-box $\epsilon$-zero-knowledge argument for all $\NP$ languages.
\end{theorem}

\subsection{Preparation}\label{sec:preparation}
Before giving our construction, we prepare a formalization of a variant of Blum's Graph Hamiltonicity protocol \cite{Blum86} by Pass and Wee \cite{TCC:PasWee09}. 
For clarity of exposition, we describe the construction and its properties in an abstracted form that is sufficient for our purpose.
\begin{definition}[Modified Hamiltonicity Protocol]\label{def:modified_hamiltonicity}
Let $(\Sigma.\Setup,\Sigma.\Commit)$ be a statistically binding and computationally hiding commitment scheme with message space $\calM_\Sigma$  randomness space $\calR_\Sigma$, and public parameter space $\ppspace_\Sigma$.
A modified Hamiltonicity protocol for an $\NP$ language $L$ instantiated with $(\Sigma.\Setup,\allowbreak \Sigma.\Commit)$ is a 4-round interactive proof for $\NP$  with the following syntax.

\noindent\textbf{Common Input:} An instance $x\in \lang \cap \bit^{\secpar}$ for security parameter $\secpar \in \mathbb{N}$.\\
\textbf{$\pro$'s Private Input:} A classical witness $w\in \rel_\lang(x)$ for $x$. 
\begin{enumerate}
\item $\ver$ generates $\pp_\Sigma\sample \Sigma.\Setup(1^\secpar)$ and sends $\pp_\Sigma$ to $\pro$.
\item $\pro$ generates $\{m_i\}_{i\in[\secpar]}\in \calM_\Sigma^\secpar$ by using $x$ (without using $w$).
We denote this procedure by $\{m_i\}_{i\in[\secpar]}\sample \Sigma.\Samp(x)$.
Then $\pro$ picks $r_i\sample \calR_\Sigma$ and computes $\com_i:=\Sigma.\Commit(\pp_\Sigma,m_i;r_i)$ for all $i\in [\secpar]$ and sets $a:=\{\com_i\}_{i\in [\secpar]}$ and $\st\defeq \{m_i,r_i\}_{i\in[\secpar]}$. 
Then it sends $a$ to $\ver$, and keeps $\st$ as a state information.  
\item $\ver$ chooses a ``challenge" $e\sample \bit^\secpar$ and sends $e$ to $\pro$.
\item $\pro$ generates a ``response" $z$ from $\st$, witness $w$, and $e$.   
We denote this procedure by $z\sample \Sigma.\pro_{\resp}(\st,w,e)$.
Then it sends $z$ to $\ver$.
\item $\ver$ verifies the transcript $(\pp_\Sigma,a,e,z)$ and outputs $\top$ indicating acceptance or $\bot$ indicating rejection.
We denote this procedure by $\top/\bot \sample \Sigma.\ver(x,\pp_\Sigma,a,e,z)$.
\end{enumerate}
We require the protocol to satisfy the following properties (in addition to perfect completeness and statistical soundness).
\paragraph{Special Honest-Verifier Zero-Knowledge.}
There exist PPT simulators $\SimSamp$ and $\siml_\resp$  such that we have
\begin{align*}
&\left\{\left(\{\com_i\}_{i\in[\secpar]},z\right):
\begin{array}{ll}
\{m_i\}_{i\in\secpar}\sample \Sigma.\Samp(x),\\
r_i\sample \calR_\Sigma \textrm{~for~}i\in[\secpar],\\ 
\com_i\sample \Sigma.\Commit(\pp_\Sigma,m_i,r_i)\textrm{~for~}i\in[\secpar],\\ 
z\sample \Sigma.\pro_{\resp}(\{m_i,r_i\},w,e)
\end{array}
\right\}_{\secpar,x,w,e,\pp_\Sigma}\\
\compind
&
\left\{\left(\{\com_i\}_{i\in[\secpar]},z\right):
\begin{array}{ll}
\{m_i\}_{i\in\secpar}\sample \SimSamp(x,e),\\
r_i\sample \calR_\Sigma \textrm{~for~}i\in[\secpar],\\ 
\com_i\sample \Sigma.\Commit(\pp_\Sigma,m_i,r_i)\textrm{~for~}i\in[\secpar],\\ 
z\sample \siml_\resp(\{m_i,r_i\},e)
\end{array}
\right\}_{\secpar,x,w,e,\pp_\Sigma}
\end{align*}
where $x\in \lang \cap \bit^\secpar$, $w\in \rel_\lang(x)$, $e\in \bit^\secpar$, and $\pp_\Sigma\in \ppspace_\Sigma$.
\paragraph{Bad Challenge Searchability.}
Intuitively, this property requires that for any $x\notin \lang$ and  $a$, if one is given a decommitment of $a$, one can efficiently compute a ``bad challenge" $e$ for which there may exist a valid response $z$.
Note that we do not require that a valid response exists for the bad challenge.
Rather, we require that a valid response can exist only for the bad challenge if it exists at all. 
A formal definition is given below.

There exists an efficiently computable function $f_\bad:\calM_\Sigma^\secpar \rightarrow \bit^\secpar$ such that  
for any binding $\pp_\Sigma$,\footnote{See definition \ref{def:commitment} for the definition of binding public parameters.} $x\in \bit^\secpar \setminus \lang$,
$a=\{\com_i=\Sigma.\Commit(\pp_\Sigma,m_i;r_i)\}_{i\in [\secpar]}$, $e\in \bit^\secpar\setminus \{f_\bad(\{m_i\}_{i\in[\secpar]})\}$, and $z$, 
 we have $\Sigma.\ver(x,\pp_\Sigma,a,e,z)=\bot$. 
\end{definition}
\begin{remark}
 Looking ahead, bad challenge searchability is needed for the reduction from computational soundness of our protocol to computational hiding of a commitment scheme. 
\end{remark}
 \paragraph{Instantiations.}
 The above definition is an abstraction of the modified version of Blum's Graph Hamiltonicity protocol by  Pass and Wee \cite{TCC:PasWee09}, and this construction only requires the existence of OWF.
 For completeness, we briefly explain more details.

 First, we recall (unparallelized version of) Blum's Graph Hamiltonicity protocol \cite{Blum86}.
 In the protocol, a prover is going to prove that a graph $G$ has a cycle where the prover is given a cycle $w$ as a witness.  
 In the first round, the prover picks a random permutation $\pi$, and commits to $\pi(G)$.
 In the second round, the verifier returns a challenge $e\in \bit$.
 In the third round, if $e=0$, the prover opens all commitments and sends $\pi$, and if $e=1$, the prover only opens commitments corresponding to the cycle $\pi(w)$.
 The verifier verifies the response in an obvious way.
 If we just implement the parallel version of Blum's Graph Hamiltonicity protocol by using (bit-wise) statistically binding commitments (e.g., Naor's commitments \cite{JC:Naor91}), then we can see that it already satisfies the syntactic requirements and the special honest-verifier zero-knowledge property. 
 However, it does not seem to satisfy bad challenge searchability.
 The reason is that, even if one is given a decommitment $G'$ of the commitment, there is no efficient way to check if $G$ and $G'$ are isomorphic, and thus one cannot know for which challenge bit a cheating prover may answer correctly.
To resolve this issue, the idea of Pass and Wee \cite{TCC:PasWee09} is to let the prover commit not only to $\pi(G)$, but also to $\pi$. 
In this case, 
for $G$ that does not have a cycle, 
if a decommitment is of the form $(\pi(G),\pi)$ for some permutation $\pi$, there does not exist a valid response for $e=1$ since $\pi(G)$ does not have a cycle,  and otherwise 
there clearly does not exist a valid response for $e=0$.
By instantiating the parallel repetition of this construction with a Naor's commitment, we obtain a protocol that satisfies the above requirements under the existence of OWF.



\subsection{Construction}
Our construction is inspired by that of $\cite{TCC:PasWee09}$, but it is not a simple instantiation of their construction since they rely on extractable commitments whose quantum security is unclear.\footnote{We could use the extractable commitment in \cite{STOC:BitShm20}, but that construction relies on a constant-round post-quantum zero-knowledge argument, which is stronger than our goal.}
Our idea is to replace extractable commitments in their construction with usual commitment combined with witness indistinguishable proof of knowledge.

Our construction is built on the following ingredients:
\begin{itemize}
    \item A commitment scheme $(\SBSetup,\SBCommit)$ that is computationally hiding and  statistically binding with message space $\bit^\secpar$ and randomness space $\calR$.
    As noted in Sec. \ref{sec:commitment}, such a commitment scheme exists under the existence of post-quantum OWF.
    \item A $4$-round witness indistinguishable proof of knowledge $(\wipok.\pro,\wipok.\ver)$ for an $\NP$ language $\widetilde{L}$ described in Figure \ref{fig:ZK_from_OWF}.
    As noted in Sec. \ref{sec:wipok}, this exists under the existence of post-quantum OWF.
    \item Modified Hamiltonicity protocol for an $\NP$ language $\lang$ instantiated with
    a computationally hiding and  statistically binding commitment scheme
    $(\Sigma.\Setup,\Sigma.\Commit)$ as defined in Definition \ref{def:modified_hamiltonicity}. 
    As discussed in Sec. \ref{sec:preparation}, this exists under the existence of post-quantum OWF.
    We denote by $\calM_\Sigma$ and $\calR_\Sigma$ message space and randomness space of the commitment scheme.
\end{itemize}
Then our construction of post-quantum black-box $\epsilon$-zero-knowledge argument is given in Figure \ref{fig:ZK_from_OWF}.

\protocol
{Protocol \ref{fig:ZK_from_OWF}}
{Constant-Round Post-Quantum $\epsilon$-Zero-Knowledge Argument for $\lang \in \NP$}
{fig:ZK_from_OWF}
{
\textbf{Common Input:} An instance $x\in \lang \cap \bit^{\secpar}$ for security parameter $\secpar \in \mathbb{N}$.\\
\textbf{$\pro$'s Private Input:} A classical witness $w\in \rel_\lang(x)$ for $x$. 
\begin{enumerate}
\item \textbf{$\ver$'s Commitment to Challenge:}
\begin{enumerate}
    \item $\pro$ computes $\pp\sample \SBSetup(1^\secpar)$ and sends $\pp$ to $\ver$.
    \label{step:ZK_from_OWF_send_pp}
    \item $\ver$ chooses $e\sample \bit^\secpar$ and $r\sample \randspace$, computes $\com \sample \SBCommit(\pp,e;r)$, and sends $\com$ to $\pro$.
    \label{step:ZK_from_OWF_send_com}
\end{enumerate}
\item \textbf{First Half of Modified Hamiltonicity Protocol:}
\begin{enumerate}
    \item $\ver$ generates $\pp_\Sigma\sample \Sigma.\Setup(1^\secpar)$.
    \label{step:ZK_from_OWF_send_pp_sigma}
    \item $\pro$ generates 
    $\{m_i\}_{i\in[\secpar]}\sample \Sigma.\Samp(x)$ and $\com_i:=\Sigma.\Commit(\pp_\Sigma,m_i;r_i)$ 
    where $r_i\sample \calR_\Sigma$  for all $i\in [\secpar]$.
  It  sends $a:=\{\com_i\}_{i\in[\secpar]}$ to $\ver$ and keeps $\st\defeq \{m_i,r_i\}_{i\in[\secpar]}$ as a state information.
  \label{step:ZK_from_OWF_sned_a}
\end{enumerate}
\item \textbf{Proof of Knowledge of Decommitments:}
  $\pro$ and $\ver$ interactively run the protocol $\execution{\wipok.\pro(\{m_i,r_i\}_{i\in[\secpar]})}{\wipok.\ver}(\pp_\Sigma,x,a)$ where the language $\widetilde{L}$ is defined as follows:
 \begin{align*}
&(\pp_\Sigma,x,a=\{\com_i\}_{i\in[\secpar]}) \in \widetilde{L}\\  \iff &(x\in \lang)~ \lor \\
&(\exists \{m_i,r_i\}_{i\in[\secpar]}\in (\calM_\Sigma\times \calR_\Sigma)^\secpar \text{~s.t.~} \com_i=\Sigma.\Commit(\pp_\Sigma,m_i;r_i)
\text{~for~all~} i \in [\secpar])
 \end{align*}
\item \textbf{Second Half of Modified Hamiltonicity Protocol:}
\begin{enumerate}
    \item $\ver$ sends $(e,r)$ to $\pro$.
    \label{Step_protocol_ZK_argument_open}
    \item $\pro$ aborts if $\SBCommit(\pp,e;r)\neq \com$.\\ Otherwise, it generates $z\sample \Sigma.\pro_{\resp}(\st,w,e)$ 
    and sends $z$ to $\ver$.
        \label{Step_protocol_ZK_argument_check}
    \item $\ver$ outputs $\Sigma.\ver(x,\pp_\Sigma,a,e,z)$. 
\end{enumerate}
\end{enumerate}
}

The completeness of the protocol clearly follows from that of the underlying $\Sigma$-protocol.
In Sec. \ref{sec:computational_soundness} and \ref{sec:proof_ZK_OWF}, we prove that this protocol satisfies computational soundness and quantum black-box $\epsilon$-zero-knowledge.
Then we obtain Theorem \ref{thm:ZK_argument_from_OWF}.

\subsection{Computational Soundness}\label{sec:computational_soundness}
Suppose that computational soundness does not hold.
This means that there exists a non-uniform QPT adversary $\pro^*=\{\pro^*_\secpar,\rho_\secpar\}$ and a sequence of false statements $\{x\in \bit^\secpar \setminus \lang\}$ such that $\Pr[\OUT_{\ver}\execution{\pro^*_\secpar(\rho_\secpar)}{\ver}(x_\secpar)=\top]$ is non-negligible.
We denote this probability by $p_\win$.
By an averaging argument, for at least $p_\win/2$-fraction of $(\pp,e,\com,\pp_\Sigma,\{\com_i\}_{i\in[\secpar]})$ generated in Step
\ref{step:ZK_from_OWF_send_pp}, 
\ref{step:ZK_from_OWF_send_com}, \ref{step:ZK_from_OWF_send_pp_sigma}, and
\ref{step:ZK_from_OWF_sned_a} and $\pro^*$'s internal state $\rho_{\pro^*}$ after Step \ref{step:ZK_from_OWF_sned_a}, 
we have $\Pr[\OUT_{\ver}\execution{\pro^*_\secpar(\rho_\secpar)}{\ver}(x_\secpar)=\top\mid (\pp,e,\com,\pp_\Sigma,\{\com_i\}_{i\in[\secpar]},\rho_{\pro^*})]\geq p_\win/2$  where the above probability means a conditional probability that $\ver$ returns $\top$ conditioned on $(\pp,e,\com,\pp_\Sigma,\{\com_i\}_{i\in[\secpar]},\rho_{\pro^*})$.
Moreover, by the statistical binding property of the commitment scheme, $\pp_\Sigma$ is binding except for negligible probability by Lemma \ref{lem:overwhelming_fraction_is_binding}.
Therefore, if we define a set $S$ consisting of $(\pp,e,\com,\pp_\Sigma,\{\com_i\}_{i\in[\secpar]},\rho_{\pro^*})$ such that 

\begin{enumerate}
    \item $\pp_\Sigma$ is binding, and
    \item $\Pr[\OUT_{\ver}\execution{\pro^*_\secpar(\rho_\secpar)}{\ver}(x_\secpar)=\top\mid (\pp,e,\com,\pp_\Sigma,\{\com_i\}_{i\in[\secpar]},\rho_{\pro^*})]\geq p_\win/2$, 
\end{enumerate}
then the probability that $(\pp,e,\com,\pp_\Sigma,\{\com_i\}_{i\in[\secpar]},\rho_{\pro^*})$ is in $S$ is non-negligible over the randomness of the execution of $\execution{\pro^*_\secpar(\rho_\secpar)}{\ver}(x_\secpar)$.

We fix $(\pp,e,\com,\pp_\Sigma,\{\com_i\}_{i\in[\secpar]},\rho_{\pro^*})\in S$.
Since $\pp_\Sigma$ is binding, for each $i\in[\secpar]$, there is a unique $m_i\in\calM$ such that 
$\SBCommit(\pp_\Sigma,m_i;r_i)=\com_i$ for some $r_i\in \calR$.
By the bad challenge searchability, a valid response $z$ can exist only for $e=f_\bad(\{m_i\}_{i\in[\secpar]})$. 
Therefore, for letting $\ver$ accept, we must have  $e=f_\bad(\{m_i\}_{i\in[\secpar]})$. 
Since 
$\ver$ accepts with probability 
$p_\win/2>0$, we must have $e=f_\bad(\{m_i\}_{i\in[\secpar]})$.\footnote{Strictly speaking, we may have $p_\win/2=0$ for a finite number of $\secpar$. This can be easily dealt with by considering sufficiently large $\secpar$. For simplicity, we omit this.}
Moreover we can use the knowledge extractor $\mathcal{K}$ of $\wipok$ to extract $\{m_i\}_{i\in[\secpar]}$ from $\pro^*$. 
That is, since the verification of $\wipok$ accepts with probability at least $p_\win/2$ (since otherwise the overall accepting probability should be smaller than $p_\win/2$), 
 we have 
\begin{align*}
\Pr[\widetilde{w}\in \rela_{\widetilde{\lang}}(\pp_\Sigma,x_\secpar,\{\com_i\}_{i\in\secpar}):\widetilde{w}\sample \kext^{\pro^*_\secpar(\rho_{\pro^*})}(\pp_\Sigma,x_\secpar,\{\com_i\}_{i\in[\secpar]})]
    \geq \frac{1}{\poly(\secpar)}\cdot (p_\win/2)^d - \negl(\secpar)
\end{align*}
where one can see that the RHS is non-negligible. 
Since we assume $x_\secpar \notin \lang$  and $\pp_\Sigma$ is binding, when we have $\widetilde{w}\in \rela_{\widetilde{\lang}}(\pp_\Sigma,x_\secpar,\{\com_i\}_{i\in\secpar})$, we have $\widetilde{w}=\{m_i,r'_i\}_{i\in[\secpar]}$ for some $\{r'_i\}_{i\in[\secpar]}$. 
By using this extracted witness, one can compute $e=f_\bad(\{m_i\}_{i\in[\secpar]})$.
We can use this to show contradiction to the unpredictability of the commitment scheme.

Specifically, we construct an adversary $\A=\{\A_\secpar,\rho_{\A,\secpar}\}_{\secpar \in \mathbb{N}}$ that breaks the unpredictability of $\SBCom$, which contradicts the computational hiding property as noted in Lemma \ref{lem:hiding_to_unpredictable}.
\begin{description}
\item[Advice:]
$\A$ gets an advice $\rho_{\A,\secpar}=(x_\secpar,\rho_\secpar)$
\item[$\A_\secpar(\rho_{\A,\secpar})$:]
It sets $\pro^*_\secpar$'s internal register to $\rho_\secpar$, gives $x_\secpar$ as input to $\pro^*_\secpar$, receives $\pp$ from $\pro^*_\secpar$, and sends $\pp$ to its external challenger.
Let $\com$ be challenger's response. (Here, $e$ is implicitly chosen by the challenger.)
$\A$ gives $\com$ to $\pro^*_\secpar$ as a message from $\ver$, and simulates the protocol between $\pro^*_\secpar$ and $\ver$ until Step \ref{step:ZK_from_OWF_sned_a}.
Let $\rho_{\pro^*}$ be $\pro^*_\secpar$'s internal state at this point.
Then it runs $\{m^*_i,r^*_i\}_{i\in[\secpar]}\sample \kext^{\pro^*_\secpar(\rho_{\pro^*})}(\pp_\Sigma,x_\secpar,\{\com_i\}_{i\in[\secpar]})$ computes $e^*=f_\bad(\{m^*_i\}_{i\in[\secpar]})$, and outputs $e^*$.
\end{description}
We can see that $\A$ perfectly simulates the soundness game until Step \ref{step:ZK_from_OWF_sned_a}.
Therefore, we have $(\pp,e,\com,\pp_\Sigma,\{\com_i\}_{i\in[\secpar]},\rho_{\pro^*})\in S$ with non-negligible probability, and for any fixed such values, we have $e^*=e$ with non-negligible probability as observed in the previous paragraph.
Therefore, $\A$ succeeds in finding $e$ with non-negligible probability overall, and breaks the unpredictability.

Since this contradicts the unpredictability, which follows from the  assumed computational hiding property, there does not exist non-uniform QPT $\pro^*$ that breaks soundness.

\subsection{Quantum Black-Box \texorpdfstring{$\epsilon$}{epsilon}-Zero-Knowledge}\label{sec:proof_ZK_OWF}
The  proof is similar to that in Sec. \ref{sec:proof_ZK_collapsing} except that we need to deal with WIPoK.
Similarly to the proof there, we show the quantum $\epsilon$-zero-knowledge property ignoring that the simulator should be black-box for clarity of exposition. One can see that the simulator is indeed  black-box by similar observations made at the end of Sec. \ref{sec:proof_ZK_collapsing}. 

In quantum $\epsilon$-zero-knowledge, we need to show a simulator $\siml$ that takes an accuracy parameter $1^{\epsilon^{-1}}$ as part of its input.
We assume $\epsilon(\secpar)= o(1)$ without loss of generality since the other case trivially follows from this case.
Without loss of generality, we can assume that a malicious verifier $\ver^*$ does not terminate the protocol before the prover aborts since it does not gain anything by declaring the termination.
We say that $\ver^*$ aborts if it fails to provide a valid opening $(e,r)$ to $\com$ in Step \ref{Step_protocol_ZK_argument_open} (i.e., the prover aborts in Step \ref{Step_protocol_ZK_argument_check}).

First, we construct a simulator $\siml_\comb$, which returns a special symbol $\fail$ with probability roughly $1/2$ but almost correctly simulates the output of $\ver^*_\secpar$ conditioned on that it does not return $\fail$.
The simulator $\siml_\comb$ uses simulators $\siml_\abort$ and $\siml_\nonabort$ as sub-protocols:

\begin{description}
\item[$\siml_\comb(x,1^{\epsilon^{-1}},\ver^*_\secpar,\rho_\secpar)$:]~
\begin{enumerate}
    \item Choose $\mathsf{mode}\sample \{\abort,\nonabort\}$.
    \item Run $\siml_{\mathsf{mode}}(x,1^{\epsilon^{-1}},\ver^*_\secpar,\rho_\secpar)$. 
    \item Output what $\siml_{\mathsf{mode}}$ outputs.
\end{enumerate}
\item[$\siml_\abort(x,1^{\epsilon^{-1}},\ver^*_\secpar,\rho_\secpar)$:]~
\begin{enumerate}
   \item Set $\ver^*_\secpar$'s internal state to $\rho_\secpar$.
    \item Compute $\pp\sample \SBSetup(1^\secpar)$ and send $\pp$ to $\ver^*_\secpar$.
    \item $\ver^*_\secpar$ returns $\com$ and $\pp_\Sigma$.
    \item 
    Compute
       $\{m_i\}_{i\in[\secpar]}\sample \Sigma.\Samp(x)$ and $\com_i:=\Sigma.\Commit(\pp_\Sigma,m_i;r_i)$ 
    where $r_i\sample \calR_\Sigma$  for all $i\in [\secpar]$, and
  sends $a:=\{\com_i\}_{i\in[\secpar]}$ to $\ver^*_\secpar$.
  Let $\rho_{\ver^*_\secpar}$ be $\ver^*_\secpar$'s internal state at this point.
  \item Interactively execute $\execution{\wipok.\pro(\{m_i,r_i\}_{i\in[\secpar]})}{\wipok.\ver^*_\secpar(\rho_{\ver^*_\secpar})}(\pp_\Sigma,x,a)$ 
  where $\wipok.\ver^*_\secpar$ is the corresponding part of $\ver^*_\secpar$.
    \item $\ver^*_\secpar$ returns $(e,r)$.
    \item Return $\fail$ and abort if $\SBCommit(\pp,e;r)= \com$.\\
    Otherwise, let $\ver^*_\secpar$ output the final output notifying that the prover aborts.
    \item The final output of $\ver^*_\secpar$ is treated as the output $\siml_\abort$.  
\end{enumerate}
\item[$\siml_\nonabort(x,1^{\epsilon^{-1}},\ver^*_\secpar, \rho_\secpar)$:]~
\begin{enumerate}
\item Set $\ver^*_\secpar$'s internal state to $\rho_\secpar$.
   \item Compute $\pp\sample \SBSetup(1^\secpar)$ and send $\pp$ to $\ver^*_\secpar$.
    \item $\ver^*_\secpar$ returns $\com$. 
    Let $\rho_\st$ be the internal state of $\ver^*_\secpar$ at this point.\footnote{Though $\com$ and $\pp_\Sigma$ can be sent simultaneously in the real protocol, we consider that they are sent one by one, and the state $\rho_\st$  is defined to be $\ver^*_\secpar$'s internal state in between them. We stress that this is just for convenience of the proof, and the protocol satisfies the same security even if the verifier sends $\com$ and $\pp_\Sigma$ simultaneously.}
    \label{step:arg_simna_generate_com}
    \item Compute $(e_\ext,\rho_\ext)\sample \ext(1^\secpar,1^{\delta^{-1}},x,\pp,\com,\A_{\open,\secpar},\rho_{\st})$ where $\ext$ is as in Lemma \ref{lem:extraction} for the commitment scheme $\SBCom$, $\delta\defeq \frac{\epsilon^2}{3600\log^4(\secpar)}$, and $\A=(\A_{\commit,\secpar},\A_{\open,\secpar})$ is defined as follows:
    \begin{description}
    \item[$\A_{\commit,\secpar}(\pp;\rho_\secpar)$:]
    It sets  $\ver^*_\secpar$'s internal state  to $\rho_\secpar$ and sends $\pp$ to $\ver^*_\secpar$. Let $\com$ be the response by $\ver^*_\secpar$ and $\rho_\st$ be the internal state of $\ver^*_\secpar$ at this point. 
    It outputs $(\com,\rho_\st)$. 
    \item[$\A_{\open,\secpar}(\rho_\st)$:] It
    sets $\ver^*_\secpar$'s internal state to $\rho_\st$, and receives $\pp_\Sigma$ from $\ver^*_\secpar$.
    It generates 
       $\{m_i\}_{i\in[\secpar]}\sample \Sigma.\Samp(x)$ and $\com_i:=\Sigma.\Commit(\pp_\Sigma,m_i;r_i)$ 
    where $r_i\sample \calR_\Sigma$  for all $i\in [\secpar]$ and 
    sends $a:=\{\com_i\}_{i\in[\secpar]}$ to $\ver^*_\secpar$.
    Let $\rho_{\ver^*_\secpar}$ be $\ver^*_\secpar$'s internal state at this point.
    Then it executes  $\execution{\wipok.\pro(\{m_i,r_i\}_{i\in[\secpar]})}{\wipok.\ver^*_\secpar(\rho_{\ver^*_\secpar})}(\pp_\Sigma,x,a)$ where  $\wipok.\ver^*_\secpar$ is the corresponding part of $\ver^*_\secpar$. 
    After completing the execution of WIPoK, $\ver^*_\secpar$ returns $(e,r)$. Let $\rho'_{\st}$ be the internal state of $\ver^*_\secpar$ at this point. 
    It outputs $(e,r,\out:=(a,\{m_i,r_i\}_{i\in[\secpar]}),\rho'_{\st})$.
    \end{description}
   Here, we remark that 
   $\ver^*_\secpar$'s internal register corresponds to $\regst$ and 
   $e$ corresponds to $m$ in the notation of Lemma \ref{lem:extraction}.    \label{step:arg_simna_ext}
  \item Set the verifier's internal state to $\rho_\ext$.
     \label{step:arg_simna_reinitialize}
 \item $\ver^*$ returns $\pp_\Sigma$.
 \label{step:arg_simna_pp_Sigma}
    \item 
    Compute $\{m_i\}_{i\in\secpar}\sample \SimSamp(x,e_\ext)$,
$\com_i\sample \Sigma.\Commit(\pp_\Sigma,m_i,r_i)$ where $r_i\sample \calR_\Sigma$ for all $i\in[\secpar]$, and  
$z\sample \siml_\resp(\{m_i,r_i\},e_\ext)$.
   Send $a:=\{\com_i\}_{i\in[\secpar]}$ to $\ver^*_\secpar$. \label{step:arg_simna_simlate_Sigma}
    Let $\rho_{\ver^*_\secpar}$ be $\ver^*_\secpar$'s internal state at this point.
 \item Interactively execute $\execution{\wipok.\pro(\{m_i,r_i\}_{i\in[\secpar]})}{\wipok.\ver^*_\secpar(\rho_{\ver^*_\secpar})}(\pp_\Sigma,x,a)$
  where $\wipok.\ver^*_\secpar$ is the corresponding part of $\ver^*_\secpar$.
  \label{step:arg_simna_wipok}
    \item $\ver^*_\secpar$ returns $(e,r)$.
    \label{step:arg_simna_generate_er}
    \item Return $\fail$ and abort if $e\neq e_\ext$ or $\SBCommit(\pp,e;r)\neq \com$.\\
    Otherwise, send $z$ to $\ver^*_\secpar$.
    \label{step:arg_simna_send_z}
    \item The final output of $\ver^*_\secpar$ is treated as the output $\siml_\nonabort$.   \label{step:arg_simna_final}
\end{enumerate}
\end{description}

Intuitively, $\siml_\abort$ (resp. $\siml_\nonabort$) is a simulator that simulates the verifier's view in the case that verifier aborts (resp. does not abort).

More formally, we prove the following lemmas.
\begin{lemma}[$\siml_\abort$ simulates the aborting case.]\label{lem:arg_siml_abort}
For any non-uniform QPT malicious verifier $\ver^*=\{\ver^*_\secpar,\rho_\secpar\}_{\secpar\in\mathbb{N}}$, let $\OUT_{\ver^*_{\abort}}\execution{\pro(w)}{\ver^*_\secpar(\rho_\secpar)}(x)$ be the $\ver^*_\secpar$'s final output that  is replaced with $\fail$ if $\ver^*_\secpar$ does not abort.
Then we have
\begin{align*}
  \{\OUT_{\ver^*_{\abort}}\execution{\pro(w)}{\ver^*_\secpar(\rho_\secpar)}(x)\}_{\secpar,x,w}  \equiv
  \{\siml_\abort(x,1^{\epsilon^{-1}},\ver^*_\secpar,\rho_\secpar)\}_{\secpar,x,w}.
\end{align*}
where $\secpar \in \mathbb{N}$, $x\in \lang\cap \bit^\secpar$, and $w\in \rel_\lang(x)$.
\end{lemma}
\begin{proof}
Since $\siml_\abort$ perfectly simulates the real execution for $\ver^*_\secpar$ when it aborts, Lemma \ref{lem:arg_siml_abort} immediately follows.
\end{proof}
\begin{lemma}[$\siml_\nonabort$ simulates the non-aborting case.]\label{lem:arg_siml_nonabort}
For any non-uniform QPT malicious verifier $\ver^*=\{\ver^*_\secpar,\rho_\secpar\}_{\secpar\in\mathbb{N}}$,
let $\OUT_{\ver^*_{\nonabort}}\execution{\pro(w)}{\ver^*_\secpar(\rho_\secpar)}(x)$ be the $\ver^*_\secpar$'s final output that is replaced with $\fail$ if $\ver^*_\secpar$ aborts.
Then we have
\begin{align*}
\{\OUT_{\ver^*_{\nonabort}}\execution{\pro(w)}{\ver^*_\secpar(\rho_\secpar)}(x)\}_{\secpar,x,w}  \compind_{\delta}
  \{\siml_\nonabort(x,1^{\epsilon^{-1}},\ver^*_\secpar, \rho_\secpar)\}_{\secpar,x,w}
\end{align*}
where $\secpar \in \mathbb{N}$, $x\in \lang \cap \bit^{\secpar}$, and $w\in \rel_\lang(x)$.
\end{lemma}
\begin{proof}
Here, we analyze $\siml_\nonabort(x,1^{\epsilon^{-1}},\ver^*_\secpar, \rho_\secpar)$.
In the following, we consider hybrid simulators $\simlnai$ for $i=1,2,3,4,5$.
We remark that they also take the witness $w$ as input unlike $\siml_\nonabort$.

\begin{description}
\item[$\simlnaone$:] This simulator works similarly to $\siml_\nonabort(x,1^{\epsilon^{-1}},\ver^*_\secpar, \rho_\secpar)$ except that in the simulation of $\wipok$ in Step \ref{step:arg_simna_wipok}, it uses witness $w$ instead of $\{m_i,r_i\}_{i\in[\secpar]}$.

By witness indistinguishability of $\wipok$,  we have 
\begin{align*}
\{\siml_\nonabort(x,1^{\epsilon^{-1}},\ver^*_\secpar, \rho_\secpar)\}_{\secpar,x,w}
\compind 
\{\simlnaone\}_{\secpar,x,w}
\end{align*}
where $\secpar\in \mathbb{N}$, $x\in \lang \cap \bit^{\secpar}$, and $w\in \rel_\lang(x)$.
\item[$\simlnatwo$:] This simulator works similarly to $\simlnaone$ except that it generates $(a=\{\com_i\}_{i\in[\secpar]},z)$  as in the real protocol for the challenge $e_\ext$
instead of using the simulator in Step \ref{step:arg_simna_simlate_Sigma}.
That is, it generates $\{m_i\}_{i\in[\secpar]}\sample \Sigma.\Samp(x)$, $\com_i:=\Sigma.\Commit(\pp_\Sigma,m_i;r_i)$  where $r_i\sample \calR_\Sigma$  for all $i\in [\secpar]$, and $z\sample \Sigma.\pro_{\resp}(\st,w,e_\ext)$ where  $\st\defeq \{m_i,r_i\}_{i\in[\secpar]}$.

By the special honest-verifier zero-knowledge property of the modified Hamiltonicity protocol, we have 
\begin{align*}
\{\simlnaone\}_{\secpar,x,w}
\compind 
\{\simlnatwo\}_{\secpar,x,w}
\end{align*}
where $\secpar\in \mathbb{N}$, $x\in \lang \cap \bit^{\secpar}$, and $w\in \rel_\lang(x)$.

\item[$\simlnathree$:] This simulator works similarly to $\simlnatwo$ except that in the simulation of $\wipok$ in Step \ref{step:arg_simna_wipok}, it uses witness $\{m_i,r_i\}_{i\in[\secpar]}$ instead of $w$.

By witness indistinguishability of $\wipok$,  we have 
\begin{align*}
\{\simlnatwo\}_{\secpar,x,w}
\compind 
\{\simlnathree\}_{\secpar,x,w}
\end{align*}
where $\secpar\in \mathbb{N}$, $x\in \lang \cap \bit^{\secpar}$, and $w\in \rel_\lang(x)$.

\item[$\simlnafour$:] This simulator works similarly to $\simlnathree$ except that the generation of $z$ is delayed until Step \ref{step:arg_simna_send_z} and it is generated as $z\sample \Sigma.\pro_{\resp}(\st,w,e)$ instead of $z\sample \Sigma.\pro_{\resp}(\st,w,e_{\ext})$.

The modification does not affect the output distribution since it outputs $\fail$ if $e\neq e_\ext$ and if $e=e_\ext$, then this simulator works in exactly the same way as the previous one.
Therefore we have 
\begin{align*}
\{\simlnathree\}_{\secpar,x,w}
\equiv
\{\simlnafour\}_{\secpar,x,w}
\end{align*}
where $\secpar\in \mathbb{N}$, $x\in \lang \cap \bit^{\secpar}$, and $w\in \rel_\lang(x)$.

\item[$\simlnafive$:] This simulator works similarly to $\simlnafour$ except that Step \ref{step:arg_simna_ext} and \ref{step:arg_simna_reinitialize} are deleted and the check of $e\neq e_\ext$ in Step \ref{step:arg_simna_send_z} is omitted.
That is, it outputs $\fail$ in  Step \ref{step:arg_simna_send_z} if and only if we have $\SBCommit(\pp,e;r)\neq \com$.
We note that $e_\ext$ and $\rho_\ext$ are no longer used at all and thus need not be generated.

We can see that 
Step \ref{step:arg_simna_generate_com} is exactly the same as executing $(\com,\rho_\st)\sample \A_{\commit,\secpar}(\pp;\rho_\secpar)$  and
Step 
\ref{step:arg_simna_pp_Sigma},
\ref{step:arg_simna_simlate_Sigma},
\ref{step:arg_simna_wipok}, and 
\ref{step:arg_simna_generate_er} of previous and this simulators are exactly the same as executing 
$(e,r,\out=(a,\{m_i,r_i\}_{i\in[\secpar]}),\rho'_\st)\sample \A_{\open,\secpar}(\rho_\ext)$ and
$(e,r,\out=(a,\{m_i,r_i\}_{i\in[\secpar]}),\rho'_\st)\sample \A_{\open,\secpar}(\rho_\st)$, respectively where we define $\rho'_\st$ in simulated experiments as $\ver^*_\secpar$'s internal state after Step \ref{step:arg_simna_generate_er}. 
Moreover, the rest of execution of the simulators can be done given $(\pp,\com,e,r,\out=(a,\st),\rho'_\st)$.  
Therefore, by a straightforward reduction to Lemma \ref{lem:extraction}, we have 
\begin{align*}
\{\simlnafour\}_{\secpar,x,w}
\statind_\delta
\{\simlnafive\}_{\secpar,x,w}
\end{align*}
where $\secpar\in \mathbb{N}$, $x\in \lang \cap \bit^{\secpar}$, and $w\in \rel_\lang(x)$.
\end{description}

We can see that $\simlnafive$ perfectly simulates the real execution for $\ver^*_\secpar$ and outputs $\ver^*_\secpar$'s output conditioned on that $\ver^*_\secpar$ does not abort, and just outputs $\fail$ otherwise.
Therefore, we have 
\begin{align*}
\{\simlnafive\}_{\secpar,x,w} \equiv \{\OUT_{\ver^*_{\nonabort}}\execution{\pro(w)}{\ver^*_\secpar(\rho_\secpar)}(x)\}_{\secpar,x,w} 
\end{align*}
where $\secpar\in \mathbb{N}$, $x\in \lang \cap \bit^{\secpar}$, and $w\in \rel_\lang(x)$. 
Combining the above, Lemma \ref{lem:arg_siml_nonabort} is proven.
\end{proof}

By combining  Lemmas \ref{lem:arg_siml_abort} and \ref{lem:arg_siml_nonabort}, we can prove the following lemma.

\begin{lemma}[$\siml_\comb$ simulates $\ver^*_\secpar$'s output with probability almost $1/2$]\label{lem:arg_siml_comb}
For any non-uniform QPT malicious verifier $\ver^*=\{\ver^*_\secpar,\rho_\secpar\}_{\secpar\in\mathbb{N}}$, 
let $\pcombsuc$ be the probability that $\siml_\comb(x,1^{\epsilon^{-1}},\ver^*_\secpar, \rho_\secpar)$ does not return $\fail$ and $D_{\mathsf{sim},\comb}( x,1^{\epsilon^{-1}},\ver^*_\secpar,\rho_\secpar)$ be a conditional  distribution of $\siml_\comb(x,1^{\epsilon^{-1}},\ver^*_\secpar, \rho_\secpar)$, conditioned on that it does not return $\fail$.
There exists a negligible function $\negl$ such that for any $x=\{x_\secpar \in \lang \cap \bit^\secpar\}_{\secpar \in \mathbb{N}}$,  we have 
\begin{align}
    \left|\pcombsuc-1/2\right|\leq \delta/2+\negl(\secpar). \label{eq:arg_pcombsuc}
\end{align}
Moreover, we have 
\begin{align}
   \{\OUT_{\ver^*}\execution{\pro(w)}{\ver^*_\secpar(\rho_\secpar)}(x)\}_{\secpar,x,w}
    \compind_{4\delta} 
    \{D_{\mathsf{sim},\comb}(x,1^{\epsilon^{-1}},\ver^*_\secpar,\rho_\secpar)\}_{\secpar,x,w} \label{eq:arg_siml_comb}
\end{align}
where $\secpar\in \mathbb{N}$, $x\in \lang \cap \bit^{\secpar}$, and $w\in \rel_\lang(x)$. 
\end{lemma}
\begin{proof}
The proof of this lemma is exactly the same as that of Lemma \ref{lem:siml_comb} in Appendix \ref{sec:siml_comb}.
\end{proof}

Finally, we convert $\siml_\comb$ to a full-fledged simulator that does not return $\fail$ by using the quantum rewinding lemma (Lemma \ref{lem:quantum_rewinding}).
This part is exactly the same as that in Sec. \ref{sec:proof_ZK_collapsing}.
Finally, we can see that our simulator is black-box similarly to the last paragraph of Sec. \ref{sec:proof_ZK_collapsing}.
This completes the proof of quantum black-box $\epsilon$-zero-knowledge property.
\else 

\fi

\section*{Acknowledgement}

NHC's research is support by
the U.S. Department of Defense and NIST through the Hartree Postdoctoral Fellowship at QuICS and by NSF through IUCRC Planning Grant Indiana University: Center for Quantum Technologies (CQT) under award number 2052730. KMC's research is partially supported by MOST, Taiwan, under Grant no. MOST 109-2223-E-001-001-MY3 and Executive Yuan Data Safety and Talent Cultivation Project (ASKPQ-109-DSTCP).

\ifnum\submission=1
\bibliographystyle{myalpha}
\else
\bibliographystyle{alpha}
\fi
\bibliography{abbrev3,crypto,reference}

\ifnum\cameraready=0
\appendix
\ifnum\submission=0
\else
\newpage
 	\setcounter{page}{1}
 	{
	\noindent
 	\begin{center}
	{\Large SUPPLEMENTAL MATERIALS}
	\end{center}
 	}
	\setcounter{tocdepth}{2}
\fi
\ifnum\submission=0 
\section{Omitted Preliminaries}\label{app:omitted_definition}
\ifnum\submission=1
\subsection{One-Way Functions and Collapsing Hash Functions}\label{app:omitted_definition_OWF_Collapsing}
\fi
We define post-quantum one-way functions and collapsing hash functions.
\begin{definition}[Post-Quantum One-Way Functions.]
We say that a function $f:\bit^*\rightarrow \bit^*$ is post-quantum one-way function if $f$ is computable in classical polynomial time and for any non-uniform QPT adversary $\A$, we have 
\[
\Pr[f(x)=f(x'):x\sample \bit^{\secpar}, x'\sample \A(f(x))]=\negl(\secpar).
\]
\end{definition}
\begin{definition}[Collapsing Hash Function.]\label{def:collapsing}
A length-decreasing function family $\mathcal{H}=\{H_k:\bit^{L}\rightarrow \bit^\ell\}_{k\in \mathcal{K}}$ for $L>\ell$ is collapsing if the following is satisfied:
\item \textbf{Collapsing.}
For an adversary $\A$, we define an experiment $\Exp_{\A}^{\mathsf{collapse}}(1^\secpar)$ as follows:
\begin{enumerate}
    \item The challenger generates $k\sample \mathcal{K}$.
    \item $\A$ is given $k$ as input and generates a hash value $y\in \bit^\ell$ and a quantum state $\sigma$ over registers $(\regX,\regA)$ where
    $\regX$ stores an element of $\bit^L$ and $\regA$ is $\A$'s internal register.
    Then it sends $y$ and register $\regX$ to the challenger, and keeps $\regA$ on its side.
    \item The challenger picks $b\sample \bit$ If $b=0$, the challenger does nothing and if $b=1$, the challenger measures register $\regX$ in the computational basis.
    The challenger returns register $\regX$ to $\A$
    \label{step:collapsing_game_measure}
    \item $\A$ outputs a bit $b'$. The experiment outputs $1$ if $b'=b$ and $0$ otherwise. 
\end{enumerate}
We say that $\A$ is a valid adversary if we we have 
\[
\Pr[H_k(x)=y:k\sample\mathcal{K},(y,\sigma)\sample \A(k),x\leftarrow M_{\regX}\circ \sigma]=1.
\]
We say that a  hash function is collapsing if for any non-uniform QPT valid adversary $\A$ we have
\[
|\Pr[1\sample \Exp_{\A}^{\mathsf{collapsing}}(1^\secpar)]-1/2|= \negl(\secpar).
\]
\end{definition}

As shown in \cite{AC:Unruh16}, a collapsing hash function with arbitrarily long (or even unbounded) input-length exists under the QLWE assumption.

\ifnum\submission=1 
\subsection{Commitments}\label{sec:appendix_commitment}
We give definitions of commitments and their security.
Though mostly standard, we introduce one new security notion which we call \emph{strong collapse-binding}, which is a stronger variant of  collapse-biding introduced by Unruh \cite{EC:Unruh16}.
As shown in Appendix \ref{sec:strong_collapse_binding}, we can see that Unruh's construction of collapse-binding commitments actually also satisfies strong collapse-binding with almost the same (or even simpler) security proof. 

\begin{definition}[Commitment.]\label{def:commitment}
A (two-message) commitment scheme   with message space $\calM$, randomness space $\calR$, commitment space $\COM$, and a public parameter space $\ppspace$  consists of two classical PPT algorithms $(\Setup,\Commit)$:
\begin{description}
\item[$\Setup(1^\secpar)$:] The setup algorithm takes the security parameter $1^\secpar$ as input and outputs a public parameter $\pp\in \ppspace$.
\item[$\Commit(\pp,m)$:] The committing algorithm takes a public parameter $\pp\in \ppspace$ and a message $m\in \calM$ as input and outputs a commitment $\com\in \COM$. 
\end{description}
We say that a commitment scheme is non-interactive if a public parameter $\pp$ generated by $\Setup(1^\secpar)$ is always just the security parameter $1^\secpar$. For such a scheme, we omit to write $\Setup$.

We define the following security notions for a commitment scheme.

\paragraph{Statistical/Computational Hiding.}
For an adversary $\A$, we consider an experiment $\Exp_{\A}^{\hiding}(1^\secpar)$ defined below:
\begin{enumerate}
    \item $\A$ is given the security parameter $1^\secpar$ and sends a (possibly malformed) public parameter $\pp\in\ppspace$ and $(m_0,m_1)\in \calM^2$ to the challenger 
    \item The challenger randomly picks $b\sample \bit$, computes $\com\sample \Commit(\pp,m_b)$, and sends $\com$ to $\A$. 
    \item $\A$ is given a commitment $\com$ and outputs $b'\in \bit$.
    The experiment outputs $1$ if $b=b'$ and $0$ otherwise.
\end{enumerate}
We say that a commitment scheme satisfies statistical (resp. computational) hiding if for any unbounded-time (resp. non-uniform QPT) adversary $\A$, we have
\begin{align*}
    |\Pr[1\sample \Exp_{\A}^{\hiding}(1^\secpar)]-1/2|=\negl(\secpar).
\end{align*}
\begin{remark}
Our definition of hiding requires that the security should hold even if $\pp$ is maliciously generated. Thus, the hiding property holds even if a receiver  runs the setup algorithm. 
\end{remark}

\paragraph{Binding.}
\begin{itemize}
\item \textbf{Perfect/Statistical/Computational Binding.}
We say that a non-interactive commitment scheme satisfies statistical (resp. computational) binding if for any unbounded-time (resp. non-uniform QPT) adversary $\A$, we have  
\ifnum\submission=0
\begin{align*}
    \Pr[\Commit(\pp,m;r)=\Commit(\pp,m';r')\land m\neq m' : \pp \sample \Setup(1^\secpar),(m,m',r,r')\sample \A(\pp)]=\negl(\secpar).
\end{align*}
\else
\begin{align*}
    \Pr\left[
    \begin{array}{ll}
     &\Commit(\pp,m;r)=\Commit(\pp,m';r')\\
     &\land~ m\neq m' 
    \end{array}
    : 
    \begin{array}{ll}
    &\pp \sample \Setup(1^\secpar),\\
    &(m,m',r,r')\sample \A(\pp)
    \end{array}
    \right]=\negl(\secpar).
\end{align*}
\fi
We say that a scheme satisfies perfect binding if the above probability is $0$ for all unbounded-time adversary $\A$.

\item \textbf{Strong Collapse-Binding.}
For an adversary $\A$, we define an experiment $\Exp_{\A}^{\clbinding}(1^\secpar)$ as follows:
\begin{enumerate}
    \item The challenger generates $\pp\sample \Setup(1^\secpar)$.
    \item $\A$ is given the public parameter $\pp$ as input and generates a commitment $\com\in\COM$ and a quantum state $\sigma$ over registers $(\regM,\regR,\regA)$ where
    $\regM$ stores an element of $\calM$, $\regR$ stores an element of $\calR$, and $\regA$ is $\A$'s internal register.
    Then it sends $\com$ and registers $(\regM,\regR)$ to the challenger, and keeps $\regA$ on its side.
    \item The challenger picks $b\sample \bit$ If $b=0$, the challenger does nothing and if $b=1$, the challenger measures registers $(\regM,\regR)$ in the computational basis.
    The challenger returns registers $(\regM,\regR)$ to $\A$
    \label{step:collapse_binding_game_measure}
    \item $\A$ outputs a bit $b'$. The experiment outputs $1$ if $b'=b$ and $0$ otherwise. 
\end{enumerate}
We say that $\A$ is a valid adversary if we we have 
\[
\Pr[\Commit(\pp,m;r)=\com:\pp\sample\setup(1^\secpar),(\com,\sigma)\sample \A(\pp),(m,r)\leftarrow M_{\regM,\regR}\circ \sigma]=1.
\]
We say that a  commitment is strongly collapse-binding if for any non-uniform QPT valid adversary $\A$,  we have
\[
|\Pr[1\sample \Exp_{\A}^{\clbinding}(1^\secpar)]-1/2|= \negl(\secpar).
\]
\end{itemize}
\begin{remark}\label{rem:statistical_to_collapse_binding}
The difference of strong collapse-binding  from the original collapse-binding  is that the challenger measures both registers $(\regM,\regR)$ in Step \ref{step:collapse_binding_game_measure} in the case of $b=1$ whereas the challenger of the original collapse-binding game only measures $\regM$.
We note that the statistical binding property immediately implies the (original) collapse-binding property, but it does not imply the strong collapse-binding property.
\end{remark}
\begin{remark}\label{rem:collapse_to_computational_binding}
One can easily see that the strong collapse-binding property implies the computational binding property. 
Indeed, if one can find $(m,r)\neq (m',r')$ such that $\Commit(\pp,m;r)=\Commit(\pp,m';r')=\com$, then we can break the strong collapse-binding property by sending 
$\com$ and
$\ket{\psi}:=\frac{1}{\sqrt{2}}(\ket{m,r}_{\regM,\regR}+\ket{m',r'}_{\regM,\regR})$ to the challenger and performing a measurement $(\ket{\psi}\bra{\psi},I-\ket{\psi}\bra{\psi})$ on the returned state to distinguish if the state is measured.
\end{remark}
\end{definition}

We introduce the following definition for convenience.
\begin{definition}[Binding Public Parameter.]\label{def:binding_pp}
We say that $\pp\in \ppspace$ is \emph{binding}
if for any commitment $\com\in \COM$, there is at most one $m\in\calM$ such that $\Commit(\pp,m;r)=\com$ for some $r\in \calR$.
\end{definition}
The following lemma is easy to see.
\begin{lemma}\label{lem:overwhelming_fraction_is_binding}
If a commitment scheme is statistically binding, then overwhelming fraction of $\pp$ generated by $\Setup(1^\secpar)$ is binding.
\end{lemma}

We also consider an additional security definition.
\begin{definition}[Unpredictability.]\label{def:unpredictability}
For an adversary $\A$, we consider an experiment $\Exp_{\A}^{\mathsf{unpre}}(1^\secpar)$ defined below:
\begin{enumerate}
    \item $\A$ is given the security parameter $1^\secpar$ and sends a (possibly malformed) public parameter $\pp\in\ppspace$ to the challenger.
    \item The challenger randomly picks $m\sample \calM$, computes $\com\sample \Commit(\pp,m)$, and sends $\com$ to $\A$. 
    \item $\A$ returns $m^*$.
    The experiment outputs $1$ if $m=m^*$ and $0$ otherwise.
\end{enumerate}
We say that a commitment scheme is  unpredictable if for any non-uniform QPT adversary $\A$, we have
\begin{align*}
    \Pr[1\sample \Exp_{\A}^{\mathsf{unpre}}(1^\secpar)]=\negl(\secpar).
\end{align*}
\end{definition}
The following lemma is a folklore, and easy to prove.
\begin{lemma}\label{lem:hiding_to_unpredictable}
If a commitment scheme is computationally hiding and $|\calM|=2^{\omega(\secpar)}$, then the scheme is unpredictable. 
\end{lemma}

\paragraph{Instantiations.}
A computationally hiding and statistically binding commitment scheme exists under the existence of OWF \cite{JC:Naor91,SICOMP:HILL99}. 
A computationally  hiding and perfectly binding non-interactive commitment scheme exists under  the QLWE assumption \cite{TCC:GHKW17,ePrint:LomSch19}.

A statistically hiding and strong collapse-binding commitment scheme exists assuming the existence of collapsing hash functions (and thus under the QLWE assumption) \cite{EC:Unruh16,AC:Unruh16}. 
This can be seen by observing that the proof of  (original) collapse-binding property from collapsing hash functions in \cite{EC:Unruh16,AC:Unruh16} already implicitly proves the strong collapse-binding property.
For completeness, we give a proof in  Appendix \ref{sec:strong_collapse_binding}.

\subsection{Witness Indistinguishable Proof of Knowledge}\label{sec:wipok}
\begin{definition}[Witness Indistinguishable Proof of Knowledge]\label{def:WIPoK}
A witness indistinguishable proof of knowledge for an $\NP$ language $\lang$ is an interactive proof for $\lang$ that satisfies the following properties (in addition to perfect completeness and statistical soundness):
\paragraph{Witness Indistinguishability.}
For any non-uniform QPT malicious verifier $\ver^*$, we have 
\begin{align*}
\{\OUT_{\ver^*}\execution{\pro(w_0)}{\ver^*}(x)\}_{\secpar,x,w_0,w_1}\compind \{\OUT_{\ver^*}\execution{\pro(w_1)}{\ver^*}(x)\}_{\secpar,x,w_0,w_1}
\end{align*}
where $\secpar\in \mathbb{N}$, $x\in \lang \cap \bit^\secpar$,  and $w_0,w_1\in \rel_\lang(x)$.
\paragraph{(Non-Adaptive) Knowledge Extractability.}
There is an oracle-aided QPT algorithm $\kext$, a polynomial $\poly$, a negligible function $\negl$, and a constant $d\in \mathbb{N}$ such that for any quantum unbounded-time  malicious prover $\pro^*=\{\pro^*_\secpar,\rho_\secpar\}_{\secpar\in \mathbb{N}}$, $\secpar \in \mathbb{N}$, and $x\in \bit^\secpar$, we have   
\begin{align*}
\Pr[w\in \rel_L(x):w\sample \kext^{\pro^*_\secpar(\rho_\secpar)}(x)]
    \geq \frac{1}{\poly(\secpar)}\cdot \Pr[\OUT_{\ver}\execution{\pro^*_\secpar(\rho_\secpar)}{\ver}(x)=\top]^d - \negl(\secpar).
\end{align*}
\end{definition}
\paragraph{Instantiations.}
We can construct a constant round (actually $4$-round) witness indistinguishable proof of knowledge for $\NP$ only assuming the existence of post-quantum OWFs by some tweak of existing works \cite{EC:Unruh12,EC:Unruh16}. We briefly explain this below.

A constant round witness indistinguishable proof of knowledge that satisfies the above requirements was first constructed by Unruh \cite{EC:Unruh12} based on \emph{strict-binding commitments}, where a commitment perfectly binds not only a message but also randomness. 
Due to the usage of strict-binding commitment, an instantiation of this protocol requires one-to-one OWF, for which there is no post-quantum candidate under standard assumptions. 
Later, Unruh \cite{EC:Unruh16} proved that the protocol in \cite{EC:Unruh12} can be instantiated using collapse-binding commitments instead of strict-binding commitments if we relax the knowledge extractability requirement to computational one.
Since the statistical binding property trivially implies the collapse-binding property as noted in Remark \ref{rem:statistical_to_collapse_binding}, we can just use statistically binding commitments as  collapse-binding commitments in the construction of \cite{EC:Unruh16}.
Moreover, since the statistically binding property can be seen as a ``statistical version" of collapse-binding, we obtain statistical knowledge extractability.\footnote{If one is not convinced by this informal explanation, one can think of our claim as the existence of a constant round witness indistinguishable \emph{argument} of knowledge for $\NP$ under the existence of OWF. Actually, this computational version of knowledge extractability suffices for the purpose of this paper.} 
In summary, the construction in \cite{EC:Unruh16} instantiated with statistically binding (and computational hiding) commitments suffices for our purpose.

We also give another more concrete explanation.
The protocol in \cite{EC:Unruh12} is a modification of Blum's Graph Hamiltonicity protocol \cite{Blum86}. For proving the knowledge extractability, Unruh introduced a rewinding technique that enables the extractor to run a prover twice for different challenges.
In his extraction strategy, the extractor records \emph{both committed messages and randomness} when it runs the prover for the first time. 
For ensuring that this does not collapse the prover's state too much, he assumed the strict-binding property.
Here, we observe that the extractor actually need not record both committed message and randomness, and it only need to record the committed message. (Indeed, the security proof in \cite{EC:Unruh16} does so).
In this case, the statistical binding property instead of the strict-binding property suffices to ensure that the prover's state is not collapsed too much since the randomness register is not measured by the extractor.
\fi
\section{Construction of Strong Collapse-Binding Commitments}\label{sec:strong_collapse_binding}
In this section, we show that (bounded-length) Halevi-Micali commitments \cite{C:HalMic96,EC:Unruh16} satisfies the strong collapse-binding property if we instantiate it based on a collapsing hash function.

\paragraph{Construction.}
In the following, we give a description of (bounded-length) Halevi-Micali commitments.
Let $\mathcal{H}=\{H_k:\bit^{L}\rightarrow \bit^\ell\}_{k\in \mathcal{K}}$ be a family of collapsing hash functions and 
$\mathcal{F}$ be a family of universal hash functions $f:\bit^L\rightarrow \bit^n$ where $L=4\ell+2n+4$.
Then Halevi-Micali commitments over message space $\bit^n$ is described as follows:
\begin{description}
\item[$\Setup(1^\secpar)$:] 
It chooses $k\sample \mathcal{K}$ and outputs $\pp:=k$.
\item[$\Commit(\pp=k,m)$:] 
It picks $f\sample \mathcal{F}$ and $r\sample \bit^L$ conditioned on that $f(r)=m$, computes $y:=H_k(r)$,  and outputs $\com:=(y,f)$.
\end{description}
This completes the description of the scheme.
In the following, we prove security.
\paragraph{Statistical Hiding.}
This proof is completely identical to that in \cite{C:HalMic96}.
\paragraph{Strong Collapse-Binding.}
Suppose that there exists a non-uniform QPT adversary $\A$ that breaks the strong collapse-binding property of the above construction.
Then we construct non-uniform QPT $\B$ that breaks the collapsing property of $\mathcal{H}$ as follows:
\begin{description}
\item[$\B(k)$:]
Given $k\in \mathcal{K}$, it sets $\pp:=k$, sends $\pp$ to $\A$ as input, and receives $(\com=(y,f),\regM,\regR)$ from $\A$.
Then $\B$ sends $(y,\regR)$ to its external challenger where $\regR$ plays the role of $\regX$ in the collapsing game.
Then it receives the register $\regR$ (which is either measured or not) returned from the challenger and sends $(\regM,\regR)$ to $\A$ as a response from the challenger.
Finally, when $\A$ outputs $b'$, then $\B$ also outputs $b'$.
\end{description}
First, if $\A$ is valid, then $\B$ is also valid since if $(m,r)$ is a valid opening for $\com=(y,f)$, then we have $H_k(r)=y$.
 Moreover, 
 since $(\regM,\regR)$ contains a superposition of valid openings to $\com$ and the value of $r$ completely determines $m$ for a valid opening $(m,r)$,  
 measuring register $\regR$ is equivalent to measuring both registers $(\regM,\regR)$.
 Based on this observation, we can see that cases of $b=0$ and $b=1$ for the strong collapse-binding and collapsing games perfectly match. Therefore $\B$ breaks the collapsing with the same advantage as $\A$ breaks the strong collapse-binding. 

\begin{remark}
By a similar proof to that in \cite{EC:Unruh16,AC:Unruh16}, we can also prove the strong collapse-binding property for the unbounded-length version.
We omit this since this is not needed for our purpose.
\end{remark}

\section{Equivalence between Definitions of Zero-Knowledge}\label{app:equivalence_ZK}
Here, we introduce a seemingly stronger definition of quantum black-box $\epsilon$-zero-knowledge than Definition \ref{def:post_quantum_ZK}, 
which captures entanglement between auxiliary input of a verifier and a distinguisher.
Then we show that they are actually equivalent. 

We define a seemingly stronger definition which we call \emph{quantum black-box $\epsilon$-zero-knowledge with entanglement}.
\begin{definition}[Post-Quantum Black-Box  $\epsilon$-Zero-Knowledge with Entanglement]\label{def:post_quantum_ZK_entanglement}
For an interactive proof or argument for $\lang$, we define the following property:
\paragraph{Quantum Black-Box $\epsilon$-Zero-Knowledge with Entanglement.}
There exists an oracle-aided QPT simulator $\siml$ that satisfies the following:
For any sequences of polynomial-size quantum circuits (referred to as malicious verifiers) $\ver^*=\{\ver^*_\secpar\}_{\secpar\in \mathbb{N}}$  that takes $x$ as input and a  state in a register $\regV$ as advice and any noticeable function $\epsilon(\secpar)$, 
we define quantum channels $\Psi_{\real,\secpar,x,w}^{\ver^*}$ and $\Psi_{\siml,\secpar,x}^{\ver^*}$ as follows:
\begin{description}
\item[$\Psi_{\real,\secpar,x,w}^{\ver^*}$:] Take a state $\sigma$ in the register $\regV$ as input and output $\OUT_{\ver^*_\secpar}\execution{\pro(w)}{\ver^*_\secpar(\sigma)}(x)$.
\item[$\Psi_{\siml,\secpar,x}^{\ver^*}$:]
Take a state $\sigma$ in the register $\regV$ as input and output
$\OUT_{\ver^*_\secpar}(\siml^{\ver^*_\secpar(\sigma)}(x,1^{\epsilon^{-1}}))$.
\end{description}
Then for any sequence of polynomial-size states  $\{\rho_\secpar\}_{\secpar\in \mathbb{N}}$ in registers $\regV$ and any additional register $\regR$, we have
\[
\{(\Psi_{\real,\secpar,x,w}^{\ver^*}\otimes I_{\regR})(\rho_\secpar)\}_{\secpar,x,w}
\compind_{\epsilon}
\{(\Psi_{\siml,\secpar,x}^{\ver^*}\otimes I_{\regR})(\rho_\secpar)\}_{\secpar,x,w}
\]
where $\secpar \in \mathbb{N}$, $x\in \lang\cap \bit^\secpar$, $w\in \rel_\lang(\secpar)$.
\end{definition}

\begin{lemma}\label{lem:equivalence_ZK}
If an interactive proof or argument satisfies
quantum black-box $\epsilon$-zero-knowledge (Definition \ref{def:post_quantum_ZK}), then it also satisfies  quantum black-box  $\epsilon$-zero-knowledge with entanglement (Definition \ref{def:post_quantum_ZK_entanglement}).
\end{lemma}
\begin{proof}
Suppose that we have sequences $\{\ver^*_\secpar\}_{\secpar\in \mathbb{N}}$ and $\{\rho_\secpar\}_{\secpar\in \mathbb{N}}$ as in the definition of quantum $\epsilon$-zero-knowledge with entanglement in Definition \ref{def:post_quantum_ZK_entanglement}.
For each $\secpar$, we consider a modified circuit $\widetilde{\ver}^*_\secpar$ that works similarly to $\ver^*_\secpar$ except that it also takes a state on the register $\regR$ as part of its advice but does not touch $\regR$ at all.  
Then clearly we have 
\[
(\Psi_{\real,\secpar,x,w}^{\ver^*}\otimes I_{\regR})(\rho_\secpar)
=
\OUT_{\widetilde{\ver}^*_\secpar}\execution{\pro(w)}{\widetilde{\ver}^*_\secpar(\rho_\secpar)}(x).
\]
By quantum black-box $\epsilon$-zero-knowledge (Definition
\ref{def:post_quantum_ZK}), 
there is a QPT simulator $\siml$ such that we have
\[
\{\OUT_{\widetilde{\ver}^*_\secpar}\execution{\pro(w)}{\widetilde{\ver}^*_\secpar(\rho_\secpar)}(x)\}_{\secpar,x,w}\compind_\epsilon \{\OUT_{\widetilde{\ver}^*_\secpar}(\siml^{\widetilde{\ver}^*_\secpar(\rho_\secpar)}(x,1^{\epsilon^{-1}}))\}_{\secpar,x,w}
\]
where $\secpar \in \mathbb{N}$, $x\in \lang\cap \bit^\secpar$, $w\in \rel_\lang(\secpar)$.
By the definition of $\widetilde{\ver}^*_\secpar(\rho_\secpar)$, it does not act on register $\regR$ and thus $\siml^{\widetilde{\ver}^*_\secpar(\rho_\secpar)}$ does not act on $\regR$ either.
Therefore, we have 
\[
\OUT_{\widetilde{\ver}^*_\secpar}(\siml^{\widetilde{\ver}^*_\secpar(\rho_\secpar)}(x,1^{\epsilon^{-1}}))=(\siml^{\ver^*_\secpar(\cdot)}(x,1^{\epsilon^{-1}})\otimes I_\regR)(\rho_\secpar).
\]
By combining the above, we can conclude that the protocol also satisfies the quantum $\epsilon$-zero-knowledge with entanglement  
w.r.t. the same simulator $\siml$.
\end{proof}
\begin{remark}
In the above proof, we do not use the full power of black-box simulation. Indeed, it suffices to assume a very mild condition that a simulator does not act on any register on which the verifier does not act. 
We also remark that the same proof works for equivalence between quantum black-box zero-knowledge and quantum black-box zero-knowledge with entanglement (which can be defined analogously).
\end{remark}

\section{Proof of Lemma \ref{lem:amplification}}\label{sec:proof_amplification}
For proving Lemma \ref{lem:amplification}, we first introduce the following lemma taken from \cite{NWZ09}, which is an easy consequence of Jordan's lemma.

\begin{lemma}[{\cite[Section 2.1]{NWZ09}}]\label{lem:decomposition}
Let $\Pi_0$ and $\Pi_1$ be projectors on an $N$-dimensional Hilbert space $\hil$. Then there is an orthogonal decomposition of $\hil$ into two-dimensional subspaces $S_j$ for $j\in[\ell]$ and $N-2\ell$ one-dimensional subspaces $T_j^{(bc)}$ for $b,c\in \bit$ that satisfies the following properties:
\begin{enumerate}
\item For each two-dimensional subspace $S_j$, there exist two orthonormal bases $(\ket{\alpha_j},\ket{\alpha_j^{\bot}})$ and $(\ket{\beta_j},\ket{\beta_j^{\bot}})$ of $S_j$ such that 
\begin{align*}
    \Pi_0\ket{\alpha_j}=\ket{\alpha_j},~~~ \Pi_0\ket{\alpha_j^{\bot}}=0,\\
    \Pi_1\ket{\beta_j}=\ket{\beta_j},~~~ \Pi_1\ket{\beta_j^{\bot}}=0.
\end{align*}
Moreover, if we let 
\begin{align*}
    p_j\defeq \bra{\alpha_j}\Pi_1 \ket{\alpha_j},
\end{align*}
then we have $0<p_i<1$ and
\begin{align*}
\ket{\alpha_j}=\sqrt{p_j}\ket{\beta_j}+\sqrt{1-p_j}\ket{\beta_j^{\bot}},
\ket{\beta_j}=\sqrt{p_j}\ket{\alpha_j}+\sqrt{1-p_j}\ket{\alpha_j^{\bot}}.
\end{align*}
\label{item:decomposition_two_dimensional}
\item  Each one-dimensional subspace $T_j^{(bc)}$ is spanned by a unit vector $\ket{\alpha_j^{(bc)}}$ such that $\Pi_0 \ket{\alpha_j^{(bc)}}=b \ket{\alpha_j^{(bc)}}$ and $\Pi_1 \ket{\alpha_j^{(bc)}}=c \ket{\alpha_j^{(bc)}}$.
\label{item:decomposition_one_dimensional}
\end{enumerate}
\end{lemma}

Then we prove Lemma \ref{lem:amplification}.
\begin{proof}(of Lemma \ref{lem:amplification}.)
We define projections $\Pi_0$ and $\Pi_1$ over $\hil=\hil_\regX\times \hil_\regY$ as
\begin{align*}
\Pi_0:=I_{\regX} \ot (\ket{0}\bra{0})_{\regY}, \Pi_1:=\Pi,
\end{align*} 
and apply Lemma \ref{lem:decomposition} for them.
In the following, we use the notations in Lemma \ref{lem:decomposition} for this particular application.
We define 
\[
S_{<t}:=\left(\bigoplus_{j:p_j<t} S_j\right) \oplus \left( \bigoplus_{j,b} T_j^{(b0)}\right)
\]
and 
\[
S_{\geq t}:=\left(\bigoplus_{j:p_j\geq t} S_j\right) \oplus \left(\bigoplus_{j,b} T_j^{(b1)}\right).
\]
Then it is easy to see that they are an orthogonal decomposition of $\hil$.
Since each subspace is invariant under
$\Pi_0=I_{\regX} \ot (\ket{0}\bra{0})_{\regY}$ and 
$\Pi_1=\Pi$ by the Item \ref{item:decomposition_two_dimensional} of Lemma \ref{lem:decomposition},  Item \ref{item:amplification_invariance_projection} of Lemma \ref{lem:amplification} immediately follows.

For any quantum state $\ket{\phi}_{\regX}$ such that $\ket{\phi}_{\regX}\ket{0}_{\regY} \in S_{<t}$, we can write 
\begin{align*}
\ket{\phi}_{\regX}\ket{0}_{\regY}=\sum_{j:p_j<t} d_j \ket{\alpha_j}+ \sum_{j}  d_j^{(10)} \ket{\alpha_{j}^{(10)}}
\end{align*}
by using $d_j \in \mathbb{C}$ and $d_j^{(10)}\in \mathbb{C}$ for each $j$ such that $\sum_{j:p_j<t} |d_j|^2+ \sum_{j}  |d_j^{(10)}|^2=1$.
 
Then we have
\begin{align*}
\bra{\phi}_{\regX}\bra{0}_{\regY} \Pi \ket{\phi}_{\regX}\ket{0}_{\regY}
&=\sum_{j:p_j<t} |d_j|^2 \bra{\alpha_j} \Pi \ket{\alpha_j}+ \sum_{j}  |d_j^{(10)}|^2 \bra{\alpha_j^{(10)}} \Pi \ket{\alpha_{j}^{(10)}}\\
&=\sum_{j:p_j<t} |d_j|^2 p_j\\
&<t.
\end{align*}
Similarly, for any $\ket{\phi}_{\regX}$ such that $\ket{\phi}_{\regX}\ket{0}_{\regY} \in S_{\geq t}$, we can show
\begin{align*}
\bra{\phi}_{\regX}\bra{0}_{\regY} \Pi \ket{\phi}_{\regX}\ket{0}_{\regY}
\geq t.
\end{align*}
This completes the proof of Item \ref{item:amplification_success_probability} of Lemma \ref{lem:amplification}.

For proving the Item \ref{item:amplification} and \ref{item:amplification_efficiency} of Lemma \ref{lem:amplification}, we first consider an algorithm $\widetilde{\Amp}$  described as follows:
\begin{description}
\item[$\widetilde{\Amp}(1^T,\ket{\psi}_{\regX,\regY})$:] 
This algorithm takes a repetition parameter $T$ and a quantum state $\ket{\psi}_{\regX,\regY}\in \hil$ as input and works as follows:\footnote{Strictly speaking, we need to consider descriptions of quantum circuits to perform measurements $\{\Pi_0, I_{\regX,\regY}-\Pi_0\}$  and $\{\Pi_1, I_{\regX,\regY}-\Pi_1\}$ as part of its input so that we can make the description of $\widetilde{\Amp}$ independent on them. (Looking ahead, this is needed for showing Item \ref{item:amplification_efficiency} in Lemma \ref{lem:amplification} where $\Amp$ is required to be a uniform QPT machine.) We omit to explicitly write them in the input of $\widetilde{\Amp}$ for notational simplicity.}
\begin{enumerate}
\item Repeat the following $T$ times:
\begin{enumerate}
\item Perform a measurement $\{\Pi_1, I_{\regX,\regY}-\Pi_1\}$ If the outcome is $1$, i.e., if $\Pi_1$ is applied, then output the state in the registers $(\regX,\regY)$ and a classical bit $b=1$ indicating a success and immediately halt.  
\item Perform a measurement $\{\Pi_0, I_{\regX,\regY}-\Pi_0\}$.
\end{enumerate}
\item Output the state in the registers $(\regX,\regY)$ and a classical bit $b=0$ indicating a failure.
\end{enumerate}
\end{description}

Then we prove the following claim:
\begin{myclaim}\label{claim:amp}
The following hold
\begin{enumerate}
\item For any quantum state $\ket{\psi}_{\regX,\regY}$, if we run $(\ket{\psi'}_{\regX,\regY},b)\sample \widetilde{\Amp}(1^T,\ket{\psi}_{\regX,\regY})$ and we have $b=1$, then $\ket{\psi'}_{\regX,\regY}$ is in the span of $\Pi_1$ with probability $1$.
\label{item:amp_map_to_Pi}
\item 
\label{item:amp_amplification}
\revise{For any noticeable function $\nu=\nu(\secpar)$, there is $T=\poly(\secpar)$ such that} 
for any quantum state $\ket{\phi}_{\regX}$ such that $\ket{\phi}_{\regX}\ket{0}_{\regY} \in S_{\geq t}$, we have 
\begin{align*}
\Pr[b=1: (\ket{\psi'}_{\regX,\regY},b)\sample \widetilde{\Amp}(1^T,\ket{\phi}_{\regX}\ket{0}_{\regY})]\geq \revise{1- \nu}.
\end{align*}
\item 
\label{item:amp_invariance}
$S_{<t}$ and $S_{\geq t}$ are invariant under $\widetilde{\Amp}(1^T,\cdot)$. 
 More precisely, 
for any  quantum state $\ket{\psi_{<t}}_{\regX,\regY}\in  S_{<t}$ if we run $(\ket{\psi'}_{\regX,\regY},b)\sample \widetilde{\Amp}(1^T,\ket{\psi_{<t}}_{\regX,\regY})$, then we have   $\ket{\psi'}_{\regX,\regY}\in S_{<t}$ with probability 1.
Similarly, for any  quantum state $\ket{\psi_{\geq t}}_{\regX,\regY}\in  S_{\geq t}$ if we run $(\ket{\psi'}_{\regX,\regY},b)\sample \widetilde{\Amp}(1^T,\ket{\psi_{\geq t}}_{\regX,\regY})$, then we have   $\ket{\psi'}_{\regX,\regY}\in S_{\geq t}$ with probability 1.
\end{enumerate}
\end{myclaim}
\begin{proof}(of Claim \ref{claim:amp}.)
Item \ref{item:amp_map_to_Pi} immediately follows from the description of $\widetilde{\Amp}$ since it returns $b=1$ only after succeeding in applying $\Pi_1$.   
Item \ref{item:amp_invariance} is easy to see noting that $\widetilde{\Amp}$ just sequentially applies measurements $\{\Pi_1, I_{\regX,\regY}-\Pi_1\}$ and $\{\Pi_0, I_{\regX,\regY}-\Pi_0\}$ over the registers $\regX$ and $\regY$ and each subspace $S_j$ or $T_j^{(bc)}$ is invariant under these measurements by Lemma \ref{lem:decomposition}.
In the following, we prove Item \ref{item:amp_amplification}.
We note that essentially the same statement was proven in \cite{TCC:ChiChuYam20}.
We include a proof for completeness.

For any quantum state $\ket{\phi}_{\regX}$ such that $\ket{\phi}_{\regX}\ket{0}_{\regY} \in S_{\geq t}$, we can write 
\begin{align*}
\ket{\phi}_{\regX}\ket{0}_{\regY}=\sum_{j:p_j\geq t} d_j \ket{\alpha_j}+ \sum_{j}  d_j^{(11)} \ket{\alpha_{j}^{(11)}}
\end{align*}
by using $d_j \in \mathbb{C}$ and $d_j^{(11)}\in \mathbb{C}$ for each $j$ such that $\sum_{j:p_j\geq t} |d_j|^2+ \sum_{j}  |d_j^{(11)}|^2=1$.
Since $\widetilde{\Amp}$ just sequentially applies measurements $\{\Pi_1, I_{\regX,\regY}-\Pi_1\}$ and $\{\Pi_0, I_{\regX,\regY}-\Pi_0\}$ over the registers $\regX$ and $\regY$ and each subspace $S_j$ or $T_j^{(bc)}$ is invariant under these measurements by Lemma \ref{lem:decomposition}, states in different subspaces do not interfere with each other. Therefore, it suffices to prove Item \ref{item:amp_amplification} of  \ref{claim:amp} assuming that 
$\ket{\phi}_{\regX}\ket{0}_{\regY}=\ket{\alpha_j}$ for some $j$ such that $p_j \geq t$ or $\ket{\phi}_{\regX}\ket{0}_{\regY}=\ket{\alpha_j^{(11)}}$ for some $j$.

The latter case is easy: If $\ket{\phi}_{\regX}\ket{0}_{\regY}=\ket{\alpha_j^{(11)}}$ for some $j$, then $\widetilde{\Amp}(1^T,\ket{\alpha_j^{(11)}})$ outputs $(\ket{\alpha_j^{(11)}},b=1)$ and halts at the very first step with probability $1$ since we have $\Pi_1\ket{\alpha_j^{(11)}}=\ket{\alpha_j^{(11)}}$ by Lemma \ref{lem:decomposition}.

In the following, we analyze the case of  $\ket{\phi}_{\regX}\ket{0}_{\regY}=\ket{\alpha_j}$ for some $j$ such that $p_j \geq t$.
For $k\in\mathbb{N}$, let  $P_k$ and $P_k^\bot$ be the probability that $\widetilde{\Amp}(1^k,\ket{\alpha_j})$ and $\widetilde{\Amp}(1^k,\ket{\alpha_j^\bot})$ succeed, respectively. (We define $P_0=P_0^\bot:=0$.)
By using 
\begin{align*}
\ket{\alpha_j}=\sqrt{p_j}\ket{\beta_j}+\sqrt{1-p_j}\ket{\beta_j^{\bot}},
\ket{\beta_j}=\sqrt{p_j}\ket{\alpha_j}+\sqrt{1-p_j}\ket{\alpha_j^{\bot}},
\end{align*}
we can see that we have
\begin{align*}
&P_{k+1}=p_j+(1-p_j)^2 P_{k}+ (1-p_j)p_j P_{k}^\bot, \\
&P_{k+1}^\bot=(1-p_j)+ p_j(1-p_j) P_{k}+ p_j^2 P_{k}^\bot.
\end{align*}

Solving this, we have 
\begin{align*}
P_T=1-(1-2p_j+2p_j^2)^{T-1}(1-p_j)  
\end{align*}
for $T\geq 1$. 
\revise{Since $p_j\ge t$ and $\nu$ are noticeable, we can take $T=\poly(\secpar)$ in such a way that $P_T\ge 1-\nu$.}
This completes the proof of Claim \ref{claim:amp}.
\end{proof}

We define an algorithm $\Amp$ as a purified version of $\widetilde{\Amp}$.
That is, $\Amp$ works similarly to $\widetilde{\Amp}$ except that intermediate measurement results are stored in designated registers in  $\reganc$ without being measured and the output $b$ is stored in register $\regB$. 
Let $U_{\amp,T}$ be the unitary part of $\Amp(1^T,\cdot)$.
Then Item \ref{item:amplification} of Lemma \ref{lem:amplification} directly follows from the corresponding statements of Claim \ref{claim:amp}.
Finally, $\Amp$ clearly runs in QPT by the definition, and thus Item \ref{item:amplification_efficiency} of Lemma \ref{lem:amplification} follows.
\end{proof}

\section{Proof of Lemma \ref{lem:siml_comb}}\label{sec:siml_comb}
Here, we give a proof of Lemma \ref{lem:siml_comb}.
\begin{proof}(of Lemma \ref{lem:siml_comb}).
We consider the following probabilities:
\begin{description}
\item[$\psuc$:] A probability that $\ver^*_\secpar$ does not abort in an execution $\execution{\pro(w)}{\ver^*_\secpar(\rho_\secpar)}(x)$.
\item[$\pasuc$:] A probability that $\siml_\abort(x,1^{\epsilon^{-1}},\ver^*_\secpar,\rho_\secpar)$ does not return $\fail$.
\item[$\pnasuc$:] A probability that $\siml_\nonabort(x,1^{\epsilon^{-1}},\ver^*_\secpar,\rho_\secpar)$ does not return $\fail$.
\end{description}
Lemma \ref{lem:siml_abort} and \ref{lem:siml_nonabort} immediately imply that there is a negligible function $\negl$ such that for all $x=\{x_\secpar\in \lang \cap \bit^\secpar\}_{\secpar \in \mathbb{N}}$ we have
\begin{align}
    \pasuc=\psuc \label{eq:pasuc}
\end{align}
and  
\begin{align}
    \left|\pnasuc-(1-\psuc)\right|\leq \delta+\negl(\secpar). \label{eq:pnasuc}
\end{align}
By the construction of $\siml_\comb(x,\ver^*_\secpar,\rho_\secpar)$, we have
\begin{align}
    \pcombsuc
    =\frac{1}{2}\left(\pasuc+\pnasuc\right). \label{eq:pcombsuc_average}
\end{align}
Combining Eq. \ref{eq:pasuc}, \ref{eq:pnasuc}, and \ref{eq:pcombsuc_average}, we obtain Eq. \ref{eq:pcombsuc}.

For proving Eq. \ref{eq:siml_comb}, We consider the following distributions:
\begin{description}
\item[$D_{\mathsf{real},\abort}( x,w,1^{\epsilon^{-1}},\ver^*_\secpar,\rho_\secpar)$]:
A conditional distribution of $\OUT_{\ver^*_\secpar}\execution{\pro(w)}{\ver^*_\secpar(\rho_\secpar)}(x)$, conditioned on that $\ver^*_\secpar$ aborts.
\item[$D_{\mathsf{real},\nonabort}( x,w,1^{\epsilon^{-1}},\ver^*_\secpar,\rho_\secpar)$]: 
A conditional distribution of $\OUT_{\ver^*_\secpar}\execution{\pro(w)}{\ver^*_\secpar(\rho_\secpar)}(x)$, conditioned on that $\ver^*_\secpar$ does not abort.
\item[$D_{\mathsf{sim},\abort}( x,1^{\epsilon^{-1}},\ver^*_\secpar,\rho_\secpar)$]: 
A conditional distribution of $\siml_\abort(x,\ver^*_\secpar,\rho_\secpar)$, conditioned on that the output is not $\fail$.
\item[$D_{\mathsf{sim},\nonabort}(x,1^{\epsilon^{-1}},\ver^*_\secpar,\rho_\secpar)$]: A conditional distribution of $\siml_\nonabort(x,1^{\epsilon^{-1}},\ver^*_\secpar,\rho_\secpar)$, conditioned on that the output is not $\fail$.
\end{description}

Then we 
consider the following sequence of distributions implicitly indexed by $\secpar \in \mathbb{N}$, $x\in\lang \cap \bit^\secpar$, and $w\in \rel_\lang(x)$.
\begin{description}
\item[$D_1$]$\defeq D_{\mathsf{sim},\comb}(x,1^{\epsilon^{-1}},\ver^*_\secpar,\rho_\secpar)$.
We note that this can be rephrased as follows:

It samples from  $D_{\mathsf{sim},\abort}( x,\ver^*_\secpar,\rho_\secpar)$ with probability 
\[
\frac{\pasuc}{\pasuc+\pnasuc}
\]
and from $D_{\mathsf{sim},\nonabort}( x,\ver^*_\secpar,\rho_\secpar)$ with probability 
\[
\frac{\pnasuc}{\pasuc+\pnasuc}.
\]
\item[$D_2$]: It samples from $D_{\mathsf{sim},\abort}( x,1^{\epsilon^{-1}},\ver^*_\secpar,\rho_\secpar)$ with probability \redunderline{$\psuc$}
and 
from \\
$D_{\mathsf{sim},\nonabort}( x,1^{\epsilon^{-1}}, \ver^*_\secpar,\rho_\secpar)$ with probability \redunderline{$1-\psuc$}.
\item[$D_3$]: It samples from \redunderline{$D_{\mathsf{real},\abort}( x,w,1^{\epsilon^{-1}},\ver^*_\secpar,\rho_\secpar)$}  with probability $\psuc$
and 
from\\
$D_{\mathsf{sim},\nonabort}( x,1^{\epsilon^{-1}}, \ver^*_\secpar,\rho_\secpar)$ with probability $1-\psuc$.
\item[$D_4$]: It samples from $D_{\mathsf{real},\abort}( x,w,1^{\epsilon^{-1}},\ver^*_\secpar,\rho_\secpar)$  with probability \redunderline{$1-\pnasuc$}
and 
from\\
\redunderline{$D_{\mathsf{real},\nonabort}( x,1^{\epsilon^{-1}}, \ver^*_\secpar,\rho_\secpar)$} with probability \redunderline{$\pnasuc$}.
\item[$D_5$]: It samples from $D_{\mathsf{real},\abort}( x,w,1^{\epsilon^{-1}},\ver^*_\secpar,\rho_\secpar)$  with probability \redunderline{$\psuc$}
and 
from\\
$D_{\mathsf{real},\nonabort}( x,1^{\epsilon^{-1}}, \ver^*_\secpar,\rho_\secpar)$ with probability \redunderline{$1-\psuc$}.

We can see that this is exactly equal to $\OUT_{\ver^*_\secpar}\execution{\pro(w)}{\ver^*_\secpar(\rho_\secpar)}(x)$
\end{description}
In the following, we give an upper bound for advantage to distinguish each neighboring distributions.
We denote by $\calD_j$ to mean $\{D_j\}_{\secpar \in \mathbb{N}, x\in \lang \cap \bit^\secpar, w\in\rel_\lang(x)}$.
\begin{itemize}
\item $\calD_1\statind_{2\delta}\calD_2$:
By Eq.  \ref{eq:pasuc} and \ref{eq:pnasuc}, 
we have 
\begin{align*}
    1-\delta-\negl(\secpar) \leq \pasuc+\pnasuc\leq  1+\delta+\negl(\secpar).  
\end{align*}
Then, by using Eq.  \ref{eq:pasuc} again, we have
\[
\left|\frac{\pasuc}{\pasuc+\pnasuc}-\psuc\right|\leq 2\delta +\negl(\secpar)
\] 
where we use  
$\delta<\epsilon=o(1)$ and in particular $\delta<1/2$ for a sufficiently large $\secpar$ and $1-z< \frac{1}{1+z}$ and $ \frac{1}{1-z}< 1+2z$ for all reals $0<z< 1/2$.
Then $\calD_1\statind_{2\delta}\calD_2$ immediately follows.\footnote{Indeed, we can prove $\calD_1\statind_{c\delta}\calD_2$ for any constant $c>1$ similarly.}
\item $\calD_2\equiv \calD_3$: This immediately follows from Lemma \ref{lem:siml_abort} since it implies  $D_{\mathsf{sim},\abort}( x,\ver^*_\secpar,\rho_\secpar)$ and $D_{\mathsf{real},\abort}( x,w,1^{\epsilon^{-1}},\ver^*_\secpar,\rho_\secpar)$ are exactly the same distributions.
\item $\calD_3\compind_\delta \calD_4$: 
Here, we denote by $D_{j,\secpar,x,w}$ to mean $D_j$ for clarifying the dependence on $\secpar,x,w$. 
We consider distributions $D_{3,\secpar,x,w,\rho^*}$ and $D_{4,\secpar,x,w,\rho^*}$ for any state $\rho^*$ in the support of $D_{\mathsf{real},\abort}( x,w,1^{\epsilon^{-1}},\ver^*_\secpar,\rho_\secpar)$ defined as follows:
\begin{description}
\item[$D_{3,\secpar,x,w,\rho^*}$]: It outputs $\rho^*$  with probability $\psuc$
and 
samples from 
$D_{\mathsf{sim},\nonabort}(x,1^{\epsilon^{-1}}, \ver^*_\secpar,\rho_\secpar)$ with probability $1-\psuc$.
\item[$D_{4,\secpar,x,w,\rho^*}$]: It outputs  $\rho^*$  with probability $1-\pnasuc$
and 
samples from 
$D_{\mathsf{real},\nonabort}( x,1^{\epsilon^{-1}}, \ver^*_\secpar,\rho_\secpar)$ with probability $\pnasuc$.
\end{description}
Suppose that $\calD_3 \compind_\delta \calD_4$ does not hold. 
This means that there exists a non-uniform PPT distinguisher $\A=\{\A_\secpar,\rho_{\A,\secpar}\}_{\secpar \in\mathbb{N}}$ and a sequence $\{(x_\secpar,w_\secpar) \in (\lang \cap \bit^\secpar)\times \rel_\lang(x)\}_{\secpar\in \mathbb{N}}$ such that
\[
|\Pr[1\sample \A_\secpar(D_{3,\secpar,x_\secpar,w_\secpar})]-\Pr[1\sample \A(D_{4,\secpar,x_\secpar,w_\secpar})]-\delta
\]
is non-negligible. 
 By an averaging argument, there exists a sequence $\{\rho^*_\secpar\}_{\secpar \in \mathbb{N}}$ such that \[
|\Pr[1\sample \A_\secpar(D_{3,\secpar,x_\secpar,w_\secpar,\rho^*_\secpar})]-\Pr[1\sample \A(D_{4,\secpar,x_\secpar,w_\secpar,\rho^*_\secpar})]-\delta
\]
is non-negligible. 
We fix such $\{\rho^*_\secpar\}_{\secpar \in \mathbb{N}}$.
By using $\A$, we construct a non-uniform QPT distinguisher $\A'$
that is given $\{\rho_{\A,\secpar},\rho^*_\secpar\}_{\secpar \in \mathbb{N}}$ as an advice and
distinguishes  $\OUT_{\ver^*_{\nonabort}}\execution{\pro(w_\secpar)}{\ver^*_\secpar(\rho_\secpar)}(x_\secpar)$ and $\siml_\nonabort(x_\secpar,1^{\epsilon^{-1}},\ver^*_\secpar,\rho_\secpar)$ as follows:
\begin{description}
\item[$\A'(\rho')$:]
It is given an input $\rho'$, which is sampled from either  $\OUT_{\ver^*_{\nonabort}}\execution{\pro(w_\secpar)}{\ver^*_\secpar(\rho_\secpar)}(x_\secpar)$ or $\siml_\nonabort(x_\secpar,1^{\epsilon^{-1}},\ver^*_\secpar,\rho_\secpar)$.
It sets $\rho:=\rho'$ if $\rho'\neq \fail$ and otherwise sets $\rho:=\rho^*_\secpar$.
Then it runs $\A$ on input $\rho$ and outputs as $\A$ outputs.
\end{description}
Clearly, if $\rho'$ is sampled from   $\OUT_{\ver^*_{\nonabort}}\execution{\pro(w_\secpar)}{\ver^*_\secpar(\rho_\secpar)}(x_\secpar)$, then $\rho$ is distributed according to $D_{3,x_\secpar,w_\secpar,\rho^*_\secpar}$ and if $\rho'$ is sampled from   $\siml_\nonabort(x_\secpar,\ver^*_\secpar,\rho_\secpar)$, then $\rho$ is distributed according to $D_{4,x_\secpar,w_\secpar,\rho^*_\secpar}$.
Therefore, 
\[
|\Pr[1\sample \A'(\OUT_{\ver^*_{\nonabort}}\execution{\pro(w_\secpar)}{\ver^*_\secpar(\rho_\secpar)}(x_\secpar))]-\Pr[1\sample \A'(\siml_\nonabort(x_\secpar,1^{\epsilon^{-1}},\ver^*_\secpar,\rho_\secpar))]|-\delta
\]
is non-negligible.
This contradicts Lemma \ref{lem:siml_nonabort}.
Therefore, we have $\calD_3\compind_\delta \calD_4$.

\item $\calD_4\compind_\delta \calD_5$:
This immediately follows from Eq. \ref{eq:pnasuc}.
\end{itemize} 
By combining the above, we obtain  Eq. \ref{eq:siml_comb}.
\end{proof}

\else  

\input{LNCS_Appendix_Extraction_Proof}

\input{LNCS_Appendix_Argument}
\fi
\newpage
  \tableofcontents
  \thispagestyle{empty}
 \fi
\end{document}